\documentclass[runningheads]{llncs}
\pdfoutput=1
\usepackage{amsmath,amssymb}

\usepackage{ifthen}
\newif\ifarxiv
\arxivtrue

\newcommand{\qedhere}{\qed}

\newcommand{\au}{\mathsf{au}}       
\newcommand{\bound}{\mathit{bound}} 
\newcommand{\bouyerstrong}{\rightsquigarrow_s}
\newcommand{\bouyerweak}{\rightsquigarrow_w}
\newcommand{\floor}[1]{\lfloor #1 \rfloor}

\newcommand{\initval}{\mathbf{0}}    
\newcommand{\clockequiv}{\cong}
\newcommand{\clocks}{\mathbb{C}}
\newcommand{\clocksF}{\clocks_F}
\newcommand{\clocksA}{\clocks_A}
\newcommand{\genTA}{(L, L_0, \CX, I, E, \Lab)}  
\newcommand{\genTTS}{(Q, \ttrans{}, \Lab, Q_0)}    
\newcommand{\initsTA}{Q_{0,\TA}}    
\newcommand{\initsTAM}{Q_{0,\TA_{\mathcal{M}}}}
\newcommand{\LabTA}{\Lab_{\TA}}     
\newcommand{\LabTAM}{\Lab_{\TA_\mathcal{M}}} 
\newcommand{\Lc}{L_c}
\newcommand{\Lmunu}{L_{\mu,\nu}}
\newcommand{\Lnu}{L_\nu}
\newcommand{\Lrelmunu}{L^{\textit{rel}}_{\mu,\nu}}
\newcommand{\monus}{\dotdiv}
\newcommand{\musem}[3]{\semop{#1}_{#2,#3}}      
\newcommand{\musemTA}[2]{\musem{#1}{\TA}{#2}}   
\newcommand{\musemTAtheta}[1]{\musemTA{#1}{\theta}}  
\newcommand{\musemTAnotheta}[1]{\semop{#1}_{\TA}}
\newcommand{\musemrel}[4]{\semop{#1}^{#2}_{#3,#4}}   
\newcommand{\musemrelphi}[3]{\musemrel{#1}{\phi}{#2}{#3}} 
\newcommand{\musemrelphiTA}[2]{\musemrelphi{#1}{\TA}{#2}} 
\newcommand{\musemrelphiTAtheta}[1]{\musemrelphiTA{#1}{\theta}} 

\newcommand{\musemsym}[3]{\musemrel{#1}{s}{#2}{#3}}      
\newcommand{\musemsymTA}[2]{\musemsym{#1}{\TA}{#2}}   
\newcommand{\musemsymTAtheta}[1]{\musemsymTA{#1}{\theta}}  

\newcommand{\Phirelmunu}{\Phi^{\textit{rel}}_{\mu,\nu}}
\newcommand{\prefix}[2]{#1_{#2}}            
\newcommand{\QTA}{Q_{\TA}}                  
\newcommand{\QTAM}{Q_{\TA_\mathcal{M}}}     
\newcommand{\suffix}[2]{#1_{#2}}            
\newcommand{\TCTL}{\mathit{TCTL}}
\newcommand{\tsucc}{\mathit{succ}} 
\newcommand{\Tmu}{T_\mu}
\newcommand{\TStop}{\mathsf{TS}}
\newcommand{\ttransTA}[1]{\ttrans{#1}_{\TA}}    
\newcommand{\ttransTAM}[1]{\ttrans{#1}_{\TA_{\mathcal{M}}}} 
\newcommand{\ttsTA}{(\QTA,\ttransTA{},\LabTA, \initsTA)}  
\newcommand{\ttsTAM}{(\QTAM,\ttransTAM{},\LabTAM, \initsTAM)} 

\providecommand{\dotdiv}{
  \mathbin{
    \vphantom{+}
    \text{
      \mathsurround=0pt 
      \ooalign{
        \noalign{\kern-.35ex}
        \hidewidth$\smash{\cdot}$\hidewidth\cr 
        \noalign{\kern.35ex}
        $-$\cr 
      }%
    }%
  }%
}
\usepackage{caption}    
\usepackage[inline,shortlabels]{enumitem}
\usepackage[labelformat=simple]{subcaption} 
\usepackage{stmaryrd} 
\usepackage{tikz}
\usetikzlibrary{automata} 

\newcommand{\barsep}{\ensuremath{\ | \ }}

\newcommand{\lgtrue}{\ensuremath{\mathsf{tt}}}
\newcommand{\lgfalse}{\ensuremath{\mathsf{ff}}}

\graphicspath{{Figures/}}

\usepackage{booktabs} 

\usepackage{centernot}
\usepackage[inline]{enumitem}
\usepackage[a]{esvect} 
\usepackage{mathtools}
\usepackage{graphicx}
\makeatletter
\providecommand{\bigsqcap}{%
  \mathop{%
    \mathpalette\@updown\bigsqcup
  }%
}
\newcommand*{\@updown}[2]{%
  \rotatebox[origin=c]{180}{$\m@th#1#2$}%
}
\makeatother
\usepackage{xspace}

\usepackage[disable,textwidth=3.2cm]{todonotes}
\newcommand{\jk}[2][]{\todo[inline,#1]{JK: #2}}
\newcommand{\rc}[2][]{\todo[inline,color=green!40,#1]{RC: #2}}

\newcommand{\dia}[1]{\langle #1 \rangle}		

\newcommand{\node}[1]{\mathbf{#1}} 

\newcommand{\emptyL}{\varepsilon}               

\newcommand{\indices}{\mathbb{I}}               
\newcommand{\ints}{\mathbb{Z}}                  
\newcommand{\lab}{\mathit{lab}}                 

\newcommand{\nats}{\mathbb{N}}

\newcommand{\preimg}[2]{#1^{-1}(#2)}            
\newcommand{\semop}[1]{||\, #1 \,||}            

\newcommand{\src}{\mathit{src}}     

\newcommand{\T}{\mathcal{T}}
\newcommand{\TDIV}{\mathsf{TDIV}} 
\newcommand{\tgt}{\mathit{tgt}}        

\newcommand{\true}{\texttt{t\!t}}

\newcommand{\Var}{\textnormal{Var}\xspace}

\usepackage{centernot}

\newcommand{\clockvals}{\mathcal{V}} 
\newcommand{\clockval}{v} 

\newcommand{\tcA}{\operatorname{\mathsf{A}}}
\newcommand{\tcE}{\operatorname{\mathsf{E}}}
\newcommand{\tcU}{\mathbin{\mathsf{U}}}
\newcommand{\tcR}{\mathbin{\mathsf{R}}}
\newcommand{\tcG}{\operatorname{\mathsf{G}}}
\newcommand{\tcF}{\operatorname{\mathsf{F}}}
\newcommand{\tcfreeze}[2]{#1 . #2}
\newcommand{\tcsem}[2]{\semop{#1}_{#2}}
\newcommand{\tcsemTA}[1]{\tcsem{#1}{\TA}}
\newcommand{\embed}{\mathsf{mu}} 

\newcommand{\AP}{\mathcal{A}}
\newcommand{\CX}{\mathit{CX}}
\newcommand{\delays}{\nnegreals}

\newcommand{\executions}[2]{\Pi_{#1}{(#2)}} 
\newcommand{\Lab}{\mathcal{L}}
\newcommand{\nnegreals}{\reals_{\geq 0}}
\newcommand{\reals}{\mathbb{R}}

\newcommand{\runs}[2]{R_{#1}(#2)} 
\newcommand{\TA}{\mathit{TA}}
\newcommand{\TAM}{\TA_{\mathcal{M}}}

\newcommand{\TS}[1]{\mathcal{T}_{#1}}

\newcommand{\ttrans}[1]{\xrightarrow{#1}}
\newcommand{\untils}[3]{U_{#1}(#2,#3)}   

\begin{document}

\title{Expressiveness Results for Timed Modal Mu-Calculi\thanks{Research of the first author supported by US Office of Naval Research Grant N000141712622.}}


\author{Rance Cleaveland\inst{1}\orcidID{0000-0002-4952-5380} \and
Jeroen J. A. Keiren\inst{2}\orcidID{0000-0002-5772-9527} \and
Peter Fontana\inst{1}}

\authorrunning{R. Cleaveland, J.\,J.\,A. Keiren and Peter Fontana}

\institute{Department of Computer Science, University of Maryland, College Park, MD, USA \\
\email{\{rance,pfontana\}@cs.umd.edu}
\and
Department of Mathematics and Computer Science, Eindhoven University of Technology, Eindhoven, The Netherlands\\
\email{j.j.a.keiren@tue.nl}
}

\maketitle

\begin{abstract}
This paper establishes relative expressiveness results for several modal mu-calculi interpreted over timed automata. These mu-calculi combine modalities for expressing passage of (real) time with a general framework for defining formulas recursively; several variants have been proposed in the literature. We show that one logic, which we call $L^{rel}_{\nu,\mu}$, is strictly more expressive than the other mu-calculi considered.  It is also more expressive than the temporal logic TCTL, while the other mu-calculi are incomparable with TCTL in the setting of general timed automata.
\end{abstract}

\section{Introduction}

Researchers have extensively studied the modal mu-calculus~\cite{kozen-results-on-1983} for the verification of transition systems because of the logic's expressiveness and support for model checking of a variety of temporal logics. The embeddings of Computation Tree Logic (CTL)~\cite{clarke-automatic-verification-1986}, Linear Temporal Logic (LTL)~\cite{pnueli-1977}, and CTL*~\cite{emerson-halpern-1986} into the modal mu-calculus~\cite{bhat-efficient-model-1996} illustrate its expressive power.  In addition, the modal mu-calculus can encode semantic equivalences such as bisimulation; so-called \emph{characteristic formulas} can be computed for states in finite-state labeled transition systems that are satisfied only by states that are semantically equivalent to the given state~\cite{graf1986modal,steffen1994characteristic}. Kleene's Fixpoint Theorem ~\cite{emerson-efficient-model-1986} also provides a way to efficiently compute fixpoints when the underlying semantic models are finite. Using this theorem, researchers have developed model-checking algorithms for various fragments of the modal mu-calculus~\cite{andersen-model-checking-1994,cleaveland-a-linear-time-1993,mader-verification-of-1997,mateescu-efficient-on-the-fly-2003}.

The situation of similarly foundational logics for timed systems is less resolved.  There are timed extensions of temporal logics, with Timed Computation Tree Logic (TCTL)~\cite{ACD1993}, a timed extension of CTL, being especially prominent for specifying systems modeled as timed automata. The state of the art for analogous modal mu-calculi is less well developed, with a variety of different extensions to the modal mu-calculus being proposed for different purposes~\cite{ACD1993,bouyer-timed-modal-2011,laroussinie-from-timed-1995,sokolsky-local-model-1995}, including model checking.  Despite results shows that model checking for the logic in~\cite{sokolsky-local-model-1995} is EXPTIME-complete~\cite{aceto-is-your-2002},
some of these logics have shown the potential to be model checked in practice~\cite{fontana-the-power-2014}, and there are tools that can efficiently model-check fragments of them, including UPPAAL~\cite{behrmann-a-tutorial-2004}, RED~\cite{wang-efficient-verification-2004}, CMC~\cite{laroussinie-cmc:-a-1998}, and the tools in~\cite{fontana-data-structure-2011,seshia-unbounded-fully-2003,zhang-fast-generic-2005}.



Despite these algorithms and tools for timed modal mu-calculus model checking, few expressivity results have been established for the different variants of the underlying logics as well as other timed temporal logics. The purpose of this paper is fill this gap by characterizing the relative expressiveness of these logics \emph{vis \`a vis} one another as well as TCTL.  Our key contributions include the following.
\begin{itemize}
\item
The definition of a reference timed modal mu-calculus, which we call $\Lrelmunu$.
\item
A complete characterization of the relative expressiveness of the timed modal mu-calculi in~\cite{ACD1993,bouyer-timed-modal-2011,laroussinie-from-timed-1995,sokolsky-local-model-1995} and $\Lrelmunu$ with respect to general timed automata, with the latter shown to be strictly more expressive than any of the former for the model of timed automata.
\item
A full characterization of the expressiveness of $\Lrelmunu$, $\Tmu$~\cite{henzinger-symbolic-model-1994} and $\Lmunu$~\cite{sokolsky-local-model-1995} with respect to TCTL, again in the setting of general timed automata, with $\Lrelmunu$ being strictly more expressive than TCTL and $\Tmu$ and $\Lmunu$ being incomparable with TCTL.
\end{itemize}

The rest of the paper is organized as follows.  The next section defines the timed automaton model and its semantics via timed transition systems.  Section~\ref{sec:timed-modal-mu-calculi} then presents the variants of the mu-calculi considered in this paper and subsequent section fully characterizes the relative expressiveness of them.  Section~\ref{sec:timed-ctl} then defines the version of TCTL considered in this paper and the section thereafter considers its relative expressiveness with respect to $\Lrelmunu$, $\Tmu$ and $\Lmunu$.  The final section gives our conclusions and directions for future research.

\paragraph{Related Work.}

Versions of several of the results reported in this paper may be found in~\cite{fontana-2014}.
With respect to mu-calculus expressiveness,~\cite{bouyer-timed-modal-2011} defines a logic, $\Lc$, that is similar to the logic $\Lnu$~\cite{laroussinie-from-timed-1995}, and proves that it is strictly more expressive than $\Lnu$.  Both logics include the greatest fixpoint operator and labeled modalities, but they differ in the modalities used to reason about the passage of time.  Neither includes least fixpoints or general negation, which could be used to encode least fixpoints.

The timed modal mu-calculus $\Tmu$ was introduced in~\cite{henzinger-symbolic-model-1994} and its expressiveness \emph{vis \`a vis} TCTL was studied there and in~\cite{penczek-advances-in-2006}.  Depending on the classes of timed automata considered, $\Tmu$ was shown to be either strictly more expressive than, or incomparable to, TCTL.

Finally, the expressiveness of different time constructs have been considered in the setting of timed extensions to linear-time temporal logics.  In particular,~\cite{bouyer-on-the-2010} proved that TPTL is strictly more expressive than MTL; both are timed extensions of LTL, but with the former using so-called \emph{freeze quantification} and the latter using intervals to place time bounds on temporal operators.

\section{Timed Automata and Labeled Transition Systems}
\label{sec:timed-automata-and-labeled-transition-systems}

\noindent
This section reviews timed automata and transition systems.  Below, we use $\nats = \{0, 1, \ldots\}$ for the set of natural numbers and $\nnegreals = \{\delta \in \reals \mid \delta \geq 0\}$ for the set of non-negative reals.

\subsection{Timed Automata}
\emph{Timed automata} are used to model systems whose behavior depends on the passage of continuous time.  To define them, we first introduce the notion of \emph{timed sort}.

\begin{definition}[Time-safe sort, timed sort]\label{def:sort}
\mbox{}
\begin{enumerate}
\item
    A \emph{sort} is a set $\Sigma$, and sort $\Sigma$ is \emph{time-safe} iff $\Sigma \cap \delays = \emptyset$.
\item
    The \emph{timed sort}, $\Delta(\Sigma)$, associated with time-safe sort $\Sigma$ is defined as $\Delta(\Sigma) = \Sigma \cup \delays$.
\end{enumerate}
\end{definition}

\noindent
A sort $\Sigma$ is used to record the set of actions a system may perform.  Sort $\Sigma$ is time-safe iff it is syntactically distinct from $\delays$.  The timed sort $\Delta(\Sigma)$ then enriches the time-safe sort $\Sigma$ with \emph{time elapses} $\delta \in \delays$ that a system may undergo during its execution.

Timed automata also extend traditional finite-state machines with a notion of \emph{clock}.

\begin{definition}[Clock structure, clock constraints]\label{def:cxcons}
\mbox{}
\begin{enumerate}
    \item
    Triple $(\clocks, \clocksA, \clocksF)$ is a \emph{clock structure} iff $\clocksA$ and $\clocksF$ are countably infinite and disjoint and $\clocks = \clocksA \cup \clocksF$. Elements of $\clocks$ are called \emph{clocks}.  If clock $x \in \clocksA$ then $x$ is an \emph{automaton clock}, while if $x \in \clocksF$ it is a \emph{freeze clock}.
    \item
    Let $(\clocks, \clocksA, \clocksF)$ be a clock structure.  Then the \emph{clock constraints} over this structure are defined by the following grammar, where $x, x' \in \clocks$, $c \in \nats$, and ${\bowtie} \in \{<, \leq, >, \geq\}$.
\begin{equation*}
\phi ::= x \bowtie c
\;\mid\; x - x' \,\mathrel{\bowtie}\, c
\;\mid\; \phi \land \phi
\end{equation*}
\end{enumerate}%
We write $\Phi(\clocks)$ for the set of clock constraints over $(\clocks,\clocksA,\clocksF)$ and $\Phi(\mathcal{C})$ for the subset of $\Phi(\clocks)$ whose constraints only mention clocks in $\mathcal{C} \subseteq \clocks$. We use the following abbreviations in what follows, where $x \in \clocks$: \lgtrue\ (``true'') for $x \geq 0$, \lgfalse\ (``false'') for $x < 0$, and $x = c$ / $x - x' = c$ for constraints $(x \leq c) \land (x \geq c)$ / $(x - x' \leq c) \land (x - x' \geq c)$. We call clock constraints of the form $x \bowtie c$ or $x - x' \bowtie c$ \emph{atomic}.  We use $\Phi_a(\clocks) \subsetneq \Phi(\clocks)$ to denote the set of all atomic clock constraints over $(\clocks,\clocksA,\clocksF)$, and $\Phi_a(\mathcal{C}) \subseteq \Phi_a(\clocks)$, where $\mathcal{C} \subseteq \clocks$, for the atomic clock constraints only involving clocks in $\mathcal{C}$.  If $\phi \in \Phi(\clocks)$ then we write $\bound(\phi) \in \nats$ for the \emph{bound}, or largest constant, and $cs(\phi) \subsetneq \clocks$ for the set of clocks, appearing in $\phi$ (see Figure~\ref{fig:bound-cs-defs} for the formal definitions of these).
\end{definition}

\begin{figure}
\centering
\begin{align*}
\bound(\phi)
&=
\begin{cases}
c   & \text{if $\phi = x \bowtie c$ or $\phi = x - x' \bowtie c$} \\
\max \{\bound(\phi_1), \bound(\phi_2)\}  & \text{if $\phi = \phi_1 \land \phi_2$.}
\end{cases}
\\
cs(\phi)
&=
\begin{cases}
\{ x \}     & \text{if $\phi = x \bowtie c$}\\
\{ x,x' \}  & \text{if $\phi = x - x' \bowtie c$} \\
cs(\phi_1) \cup cs(\phi_2)  & \text{if $\phi = \phi_1 \land \phi_2$}
\end{cases}
\end{align*}
\caption{Formal definitions of bound and clock set for $\phi$.}
    \label{fig:bound-cs-defs}
\end{figure}

A clock structure specifies a countably infinite set of clocks that are used to record the passage of time.  The set of clocks is in turn partitioned into automaton clocks, which may be used in timed automata, and freeze clocks, which are reserved for use in logical formulas.  Clock constraints represent properties about clock values that are used in the setting of timed automata.

In what follows, we fix the clock structure $(\clocks,\clocksA,\clocksF)$.
The final ingredient in the definition of timed automata are \emph{clock-safe atomic propositions}, which are required to be syntactically distinct from the set of clock constraints.

\begin{definition}[Clock safety]
Set $\AP$ of atomic propositions is \emph{clock-safe} iff $\AP \cap \Phi(\clocks) = \emptyset$.
If $\AP$ is a clock-safe set of atomic propositions and $\mathcal{C} \subseteq \clocks$ then we use $A_{\mathcal{C}} = \AP \cup \Phi_a(\mathcal{C})$ for the set $\AP$ enriched with the set of atomic clock constraints over clock set $\mathcal{C}$.
\end{definition}

Timed automata are now defined as follows.

\begin{definition}[Timed automaton]
\label{def:timedaut}
Let $\Sigma$ be a time-safe sort and $\AP$ be a clock-safe set of atomic propositions.  A \emph{timed automaton} over $\Sigma$ and $\AP$
is a tuple $(L, L_0, \CX, I, E, \Lab)$ where:
\begin{itemize}
\item $\emptyset \subsetneq L$ is the non-empty finite set of \emph{locations};
\item $\emptyset \subsetneq L_0 \subseteq L$ is the non-empty set of \emph{initial locations};
\item $\emptyset \subsetneq \CX \subseteq \clocksA$ is a nonempty finite set of \emph{clocks};
\item $I \in L \to \Phi(\CX)$ maps $l \in L$ to its \emph{location invariant} $I(l) \in \Phi(\CX)$;
\item $E \subseteq L \times \Sigma \times \Phi(\CX) \times 2^{\CX} \times L$ is the set of \emph{edges}; and
\item $\Lab \in L \to 2^{\AP}$ is the \emph{labeling function}.
\end{itemize}
In edge $e = (l, a, \phi, \mathcal{C}, l')$ $l$ and $l'$ are referred to as the \emph{source} and \emph{target} locations, respectively, while $a$ is the \emph{action}, $\phi \in \Phi(\CX)$ is the \emph{guard}, and $\mathcal{C} \subseteq \CX$ is the \emph{reset set} (clocks reset to $0$ when edge $e$ is executed).  If $l \in L$ then $\Lab(l)$ indicates which atomic propositions are true of $l$.  The \emph{bound}, $\bound(\TA)$, of $\TA$ is the largest constant appearing in the definition of $\TA$.  Formally:
$
\bound(\TA) =
\max \left(
        \{\bound(\phi) \mid (\exists l \in L:: \phi = I(l)) \lor (\exists e \in E :: e = (\ldots,\phi,\ldots))\}
     \right).
$
\end{definition}

\subsection{Semantics of Timed Automata}

Semantically, timed automata are interpreted as \emph{labeled transition systems}.

\begin{definition}[Labeled transition system (LTS)]
Let $\Sigma$ be a sort and $\AP$ be a set of atomic propositions.  Then a \emph{labeled transition system} over $\Sigma$ and $\AP$ is a tuple $(Q, \ttrans{}, \Lab, Q_0)$, where $Q$ is the set of \emph{states}, ${\ttrans{}} \subseteq Q \times \Sigma \times Q$ is the \emph{transition relation}, $\Lab \in Q \to 2^{\AP}$ is the \emph{labeling}, and $Q_0 \subseteq Q$ is the set of \emph{initial states}.
\end{definition}

\noindent
An LTS encodes the operational behavior of a system, with $Q$ representing the possible states the system can be in and the transition relation denoting which actions are possible in a state and what the possible target states are after the action is performed.  The labeling assigns to each state the atomic propositions that are true in that state; $Q_0$ gives the possible start states of the system.  In what follows, if $a \in \Sigma$ then we write $q \ttrans{a} q'$ when $(q, a, q') \in {\ttrans{}}$, $q \ttrans{a}$ if $q \ttrans{a} q'$ for some $q'$, and $q \centernot{\ttrans{a}}$ if $q \ttrans{a}$ does not hold.
We occasionally abuse notation and write $a(q)$ for $\{q' \in Q \mid q \ttrans{a} q'\}$, the set of states reachable from $q$ via action $a$.
If $n \in \nats$ we use $q \,(\ttrans{a})^n\, q'$ when there exists $q = q_0, q_1, \ldots, q_n = q'$ such that $q_0 \ttrans{a} q_1 \cdots \ttrans{a} q_n$.  Note that $q \,(\ttrans{a})^0\, q$ for all $q$ and $a$.
When $K \subseteq \Sigma$ we also write $q \ttrans{K} q'$, etc. when $q \ttrans{a} q'$, etc., for some $a \in K$.

To give an LTS semantics for timed automata we first define what the states in such an LTS must be.  Each such state will include a \emph{clock valuation}, which is defined as follows.

\begin{definition}[Clock valuation]
A \emph{valuation} over clock structure $(\clocks, \clocksA, \clocksF)$ is a function $\clockval \in \clocks \to \nnegreals$.  We use $\clockvals_{\clocks}$ to denote the set of all valuations over $(\clocks,\clocksA,\clocksF)$, and refer to such valuations as \emph{clock valuations} over $\clocks$.
\end{definition}
Intuitively, $\clockval(x)$ records the current the time of clock $x$ in clock valuation $\clockval$. If the clock structure $(\clocks,\clocksA,\clocksF)$ is clear from context we write $\clockvals$ instead of $\clockvals_\clocks$,
The following notions are standard.
\begin{enumerate}
\item
    Let $\clockval \in \clockvals_\clocks$, $x \in \clocks$, and $\delta \in \nnegreals$.
    \begin{itemize}
    \item
        $\clockval[x := \delta] \in \clockvals$ is the clock valuation that is the same as $\clockval$, except that clock $x$ is mapped to $\delta$:  $(\clockval[x := \delta])(x') = \delta$ if $x' = x$, and $(\clockval[x := \delta])(x') = \clockval(x')$ if $x' \neq x$.
    \item
        $\clockval+\delta \in \clockvals$ is the clock valuation $\delta$ time units in the future from $\clockval$:  $(\clockval + \delta)(x) = \clockval(x) + \delta$ for all $x \in \clocks$.
    \end{itemize}
\item
    $\initval_\clocks \in \clockvals_\clocks$ assigns $0$ to every clock in $\clocks$:  $\initval_\clocks(x) = 0$ for all $x \in \clocks$.
\item
    If $\mathcal{C} \subseteq \clocks$ and $\clockval \in \clockvals_\clocks$ then $\clockval[\mathcal{C} := 0]$ is $\clockval$ with every clock in $\mathcal{C}$ reset to $0$:
    $(\clockval[\mathcal{C} :=0])(x) = 0$ if $x \in \mathcal{C}$, and $(\clockval[\mathcal{C} :=0])(x) = \clockval(x)$ if $x \not\in \mathcal{C}$.
\item
    If $\phi \in \Phi(\clocks)$ is a clock constraint and $\clockval \in \clockvals_\clocks$ then $\clockval \models \phi$ holds iff the values assigned to the clocks by $\clockval$ satisfies $\phi$ in the usual sense.  Formally, let ${\bowtie} \in \{<, \leq, >, \geq\}$.  Then $\models$ is given inductively as follows.
\[
\begin{array}{rcll}
v   & \models
    & x \bowtie c
    & \text{iff $v(x) \bowtie c$}
\\
v   & \models
    & x - x' \bowtie c
    & \text{iff $v(x) - v(x') \bowtie c$}
\\
v   & \models
    & \phi_1 \land \phi_2
    & \text{iff $v \models \phi_1$ and $v \models \phi_2$}
\end{array}
\]
\end{enumerate}

Given clock structure $(\clocks,\clocksA,\clocksF)$,
the LTS semantics of timed automaton $\TA$ is the following.

\begin{definition}[Timed automaton semantics]\label{d:ta-sem}
Let $\TA = \genTA$ be a timed automaton over time-safe $\Sigma$ and clock-safe $\AP$.
Then LTS $\TS{\TA} = \ttsTA$
over $\Delta(\Sigma)$ and $\AP_{\clocks} = \AP \cup \Phi_a(\clocks)$ is defined as follows.
\begin{itemize}
\item
    $Q_{\TA} = \{ (l, \clockval) \in L \times \clockvals \mid \clockval \models I(l)\}$ is the set of states.
\item
    ${\ttrans{}_{\TA}} \subseteq Q_{\TA} \times \Delta(\Sigma) \times Q_{\TA}$ is given as follows.
    \begin{itemize}
    \item
        \emph{Time elapse:}
        for $\delta \in \delays$, $(l, \clockval) \ttrans{\delta} (l, \clockval + \delta)$ iff
	for all $\delta'$ with $0 \leq \delta' \leq \delta$,
            $\clockval + \delta' \models I(l)$.
    \item
        \emph{Action:}
        for $a \in \Sigma$, $(l, \clockval) \ttrans{a} (l',\clockval [\mathcal{C} := 0])$ iff
        there is $(l,a,\phi,\mathcal{C},l') \in E$ such that
        $\clockval \models \phi$ and $\clockval [\mathcal{C} := 0] \models I(l')$.
    \end{itemize}
\item
    $\Lab_{\TA} \in Q_{\TA} \to 2^{\AP_{\clocks}}$ is
    $
    \Lab_{\TA}(l, \clockval) = \Lab(l) \cup \{\phi \in \Phi_a(\clocks) \mid \clockval \models \phi\}.
    $
\item $Q_{(0,TA)} = \{(l_0,\initval) \mid l_0 \in L_0 \land \initval \models I(l_0)\}$ is the set of initial states.
\end{itemize}
\end{definition}

Note that in $\TS{\TA}$ the set of atomic propositions includes atomic clock constraints as well as elements in $\AP$.
We adapt notions on clock valuations to states $q = (l, \clockval)$ in $\QTA$ in the obvious fashion:  $q(x) = \clockval(x)$, $q + \delta = (l, \clockval + \delta)$, $q[x := \delta] = (l, \clockval[x := \delta])$, etc.
Transition relation $\ttrans{}_{\TA}$ may also be seen to have the following properties.
\begin{description}
\item[Time-reflexivity.]
    For all $q \in Q_{\TA}$, $q \ttrans{0} q$.
\item[Time-determinacy.]
    For all $q, q', q'' \in Q_{\TA}$ and $\delta \in \delays$,
    if $q \ttrans{\delta} q'$
    and $q \ttrans{\delta} q''$ then $q' = q''$.
\item[Time-additivity.]
    For all $q, q', q'' \in Q_{\TA}$ and $\delta, \delta' \in \delays$, if $q \ttrans{\delta} q' \ttrans{\delta'} q''$ then $q \ttrans{\delta + \delta'} q''$.
\item[Time-continuity.]
    For all $q,q' \in Q_{\TA}$ and $\delta, \delta' \in \delays$ such that $q \ttrans{\delta+\delta'} q'$, there exists $q'' \in Q_{\TA}$ such that $q \ttrans{\delta} q'' \ttrans{\delta'} q'$.
\end{description}

\noindent
We call LTSs with these properties \emph{timed transition systems}.

\begin{definition}[Timed transition system (TTS)]
Let $\Sigma$ be a time-safe sort and $\AP$ be a set of atomic propositions.  Then $(Q,\ttrans{},\Lab,Q_0)$ is a \emph{timed transition system} over $\Sigma$ and $\AP$
iff $(Q, \ttrans{}, \Lab, Q_0)$ is an LTS over $\Delta(\Sigma)$ and $\AP$ and ${\ttrans{}}$ is time-reflexive, -determinate, -additive, and -continuous.
\end{definition}

Note that in TTS $(Q,\ttrans{},\Lab,Q_0)$ $\delta(q) \subseteq Q$
is either empty (because $q \centernot{\ttrans{\delta}}$) or a singleton $\{ q' \}$ (because $q \ttrans{\delta} q'$ and $\ttrans{\delta}$ is time-determinate) for any $\delta \in \delays$.\footnote{Recall that since $\genTTS$ is an LTS over $\Delta(\Sigma)$ and $\delta \in \Delta(\Sigma)$, $\delta(q) = \{ q' \in Q \mid q \ttrans{\delta} q' \}$.}
Also, time-additivity and -continuity guarantee that for any $\delta \in \delays$, if $\delta(q) = \emptyset$ then $\delta'(q) = \emptyset$ for all $\delta' > \delta$, and if $\delta(q) \neq \emptyset$ then $\delta'(q) \neq \emptyset$ for all $\delta' < \delta$.
Finally, in what follows we often want to assert that $q \ttrans{\delta} q'$ and $q'$ is in a set $S$.  While this can be written as $\delta(q) \cap S \neq \emptyset$, for reasons of conciseness we often abuse notation and write $\delta(q) \in S$ instead.
Finally, we often need to take a $\delta'$-duration piece of a time-elapse transition $q \ttrans{\delta} q'$, where $\delta' \leq \delta$.  We define $cut(q \ttrans{\delta} q', \delta') = q \ttrans{\delta'} \delta'(q)$ for this purpose.  We also use $rem(q \ttrans{\delta} q', \delta') = \delta'(q) \ttransTA{\delta-\delta'} q'$ for the remainder of the transition $q \ttransTA{\delta} q'$ after $cut(q \ttransTA{\delta} q', \delta')$.

\subsection{Region Automata}\label{sec:region-automata}

Analysis routines for timed automata, such as model checkers, often work with \emph{region automata}~\cite{ACD1993,alur-a-theory-1994}, which are discrete abstractions of timed automata that nevertheless preserve key behavioral properties.  We use region automata later in the paper and so define them here, albeit in somewhat more abstract way than in~\cite{ACD1993,alur-a-theory-1994}.
We begin by introducing the notions of \emph{logical equivalence} and \emph{bounded logical equivalence} on clock valuations.  In the remainder of this section, fix clock structure $(\clocks, \clocksA, \clocksF)$.

\begin{definition}[Logical equivalences for clock valuations]
Let $\mathcal{C} \subseteq \clocks$ be a set of clocks.
\begin{enumerate}
\item
Clock valuations $v, v' \in \clockvals_\clocks$ are \emph{logically equivalent} with respect to $\mathcal{C}$, notation $v \clockequiv_{\mathcal{C}} v'$, iff for all atomic clock constraints $\phi \in \Phi_a(\mathcal{C})$, $v \models \phi$ exactly when $v' \models \phi$.
\item
Fix $d \in \nats$, and define $\Phi_{a}^d(\mathcal{C}) = \{\phi \in \Phi_a(\mathcal{C}) \mid \bound(\phi) \leq d\}$.  Then $v,v' \in \clockvals_\clocks$ are \emph{bounded logically equivalent with respect to $\mathcal{C}$ and $d$}, notation $v \clockequiv_{\mathcal{C},d} v'$, iff for all $\phi \in \Phi_{a}^d(\mathcal{C})$, $v \models \phi$ exactly when $v' \models \phi$.
\end{enumerate}
\end{definition}

Thus $v \clockequiv_{\mathcal{C}} v'$ iff $v$ and $v'$ satisfy exactly the same atomic clock constraints in $\Phi_a(\mathcal{C})$ (and hence the same non-atomic clock constraints as well).
Also, $v \clockequiv_{\mathcal{C},d} v'$ iff $v$ and $v'$ satisfy the same atomic clock constraints in $\Phi_a(\mathcal{C})$ whose constant is no larger than $d$.
It is straightforward to see that $v \clockequiv_{\mathcal{C}} v'$ iff for all $d \in \nats, v \clockequiv_{\mathcal{C},d} v'$.
The next lemma establishes that $\clockequiv_{\mathcal{C},d}$ is preserved by clock resetting and time elapses.

\begin{lemma}[Preservation of $\clockequiv_{\mathcal{C},d}$]
\label{lem:preservation-of-eq}
Suppose that $v, v' \in \clockvals_\clocks$, $\mathcal{C} \subseteq \clocks$, and $d \in \nats$
are such that $v \clockequiv_{\mathcal{C},d} v'$.  Then the following hold.
\begin{enumerate}
\item\label{sublem:preservation-of-eq:reset}
For all $\mathcal{C}' \subseteq \clocks, v[\mathcal{C}' := 0] \clockequiv_{\mathcal{C},d} v'[\mathcal{C}' := 0]$.
\item\label{sublem:preservation-of-eq:time-transition}
For all $\delta \in \delays$ there exists $\delta' \in \delays$ such that $v + \delta \clockequiv_{\mathcal{C},d} v' + \delta'$.
\end{enumerate}
\end{lemma}
\begin{proof}
Follows from the definitions.\qedhere
\end{proof}

As $\clockequiv_{\mathcal{C},d}$ is an equivalence relation, it partitions $\clockvals_\clocks$ into a set of equivalence classes.
If $v \in \clockvals_\clocks$
then we write $[v]_{\mathcal{C},d} = \{v' \in \clockvals_\clocks \mid v \clockequiv_{\mathcal{C},d} v' \}$
for the equivalence class of $v$,
and $[\clockvals_\clocks]_{\mathcal{C},d} = \{[v]_{\mathcal{C},d} \mid v \in \clockvals_\clocks\}$
for the set of equivalence classes of $\clockequiv_{\mathcal{C},d}$. If $\mathcal{C}$ is finite then so is $\Phi_a^d(\mathcal{C})$ for any $d$, as is $[\clockvals_\clocks]_{\mathcal{C},d}$, since every equivalence class of $\clockequiv_{\mathcal{C},d}$ is associated with a subset of the finite set $\Phi_a^d(\mathcal{C})$ (specifically, the subset of $\Phi_a^d(\mathcal{C})$ such that every state in the equivalence classes satisfies every formula in the subset and violates all formulas not in the subset).
For any $[v]_{\mathcal{C},d}$ and $\phi \in \Phi_a^d(\mathcal{C})$,
we know that either every $v' \in [v]_{\mathcal{C},d}$ satisfies $\phi$, or no $v' \in [v]_{\mathcal{C},d}$ does.
We will write $[v]_{\mathcal{C},d} \models \phi$ for the former and $[v]_{\mathcal{C},d} \not\models \phi$ for the latter.

The relation $\clockequiv_{\mathcal{C},d}$ is \emph{anti-monotonic} in $\mathcal{C}$ and $d$, as the next lemma demonstrates.

\begin{lemma}[Anti-monotonicity of $\clockequiv_{\mathcal{C},d}$]
Suppose $\mathcal{C} \subseteq \mathcal{C}' \subseteq \clocks$, and let $d, d' \in \nats$ be such that $d \leq d'$.  Then $v \clockequiv_{\mathcal{C}',d'} v'$ implies $v \clockequiv_{\mathcal{C},d} v'$.
\end{lemma}
\begin{proof}
Follows from the definition of $\clockequiv_{\mathcal{C},d}$ and the fact that $\Phi_a^d(\mathcal{C}) \subseteq \Phi_a^{d'}(\mathcal{C}')$.\qedhere
\end{proof}

This lemma in effect says that as more clocks are added and as the constant grows, $\clockequiv_{\mathcal{C},d}$ shrinks, due to the fact that the number of atomic propositions to be satisfied / violated by equivalent clock valuations increases.  More formally, if $\mathcal{C} \subseteq \mathcal{C}'$ and $d \leq d'$, then $\clockequiv_{\mathcal{C}',d'}$ refines $\clockequiv_{\mathcal{C},d}$.  If also follows in this case that for every $v \in \clockvals_\clocks$, $[v]_{\mathcal{C}',d'} \subseteq [v]_{\mathcal{C},d}$.

Equivalence classes of $\clockequiv_{\mathcal{C},d}$, which we henceforth refer to as \emph{regions}, play a major role in the construction of region automata.
We review the following notions from~\cite{ACD1993,alur-a-theory-1994}.

\begin{definition}[Unbounded / successor region]
Let $\mathcal{C} \subseteq \clocks$ and $d \in \nats$.
\begin{enumerate}
\item
Region $r \in [\clockvals_\clocks]_{\mathcal{C},d}$ is \emph{unbounded} iff for every $v \in r$ and $\delta \in \delays$, $v + \delta \in r$.
\item
Let $r \in [\clockvals_\clocks]_{\mathcal{C},d}$.  Then the \emph{successor region}, $\tsucc(r) \in [\clockvals_\clocks]_{\mathcal{C},d}$, of $r$ is defined as follows.
\[
\tsucc(r) =
\begin{cases}
r
& \text{if $r$ is unbounded}
\\
r'
& \text{if $r \neq r'$ and}
\\
&\text{$\forall v \in r \colon \exists \delta_v \colon v + \delta_v \in r' \land
\forall \delta' \colon (\delta' < \delta \implies v+\delta' \in r \cup r')$}
\end{cases}
\]
\end{enumerate}
\end{definition}
Intuitively, a region is unbounded if every time successor of every clock valuation in the region remains in the region.  One region is a successor of another if either both are the same unbounded region, or if the successor is different (and hence disjoint) but consists of states that are reachable via a time delays from the first region that do not traverse any other regions.
It can be seen that $\tsucc$  is indeed a function over $[\clockvals_\clocks]_{\mathcal{C},d}$.  We define $\tsucc^0(r) = r$, and $\tsucc^{n+1}(r) = \tsucc(\tsucc^n(r))$ for $n \in \nats$.

The construction of the region automaton for timed automaton $\TA$ over $\Sigma$ and $\AP$ is given below.  Intuitively, given a finite set of auxiliary clocks $\mathcal{C}$ and a bound $d \in \nats$, region automaton $R_{\mathcal{C},d}(\TA)$ is a finite-state labeled transition system over sort $\Sigma_\varepsilon = \Sigma \cup \{\varepsilon\}$, where $\varepsilon \not\in \Sigma$ is a distinguished symbol, and an atomic proposition set consisting of elements of $\AP$ and atomic clock constraints that involve only clocks in $\TA$ or $\mathcal{C}$ and whose constants cannot exceed $d$ nor the largest constant in $\TA$ .
The states are pairs consisting of a location from $\TA$ and a region of clock valuations, where the equivalence relation is a bounded logical equivalence constructed based on $\TA$, $\mathcal{C}$ and $d$.
The parameters $\mathcal{C}$ and $d$ can be seen as coming from the analysis being performed on $\TA$; for example, they may correspond to clocks and a time bound used in a temporal formula that is being model-checked against $\TA$.

\begin{definition}[Region automaton]\label{def:region-automaton}
Let $\TA = \genTA)$ be a timed automaton over time-safe $\Sigma$ and clock-safe $\AP$, with $\mathcal{C} \subsetneq \clocks$ a finite set of clocks and $d \in \nats$.  Define:
\begin{itemize}
\item $\TS{\TA} = \ttsTA$;
\item $\mathcal{C}_\TA = \mathcal{C} \cup CX$;
\item $d_\TA = \max\{ d, \bound(\TA) \}$ (recall $\bound(\TA)$ is the largest constant in $\TA$);
\item $[v]_{\TA,\mathcal{C},d} = [v]_{\mathcal{C}_{\TA},d_{\TA}}$ for $v \in \clockvals_{\clocks}$;
\item $[\clockvals_{\clocks}]_{\TA, \mathcal{C}, d} = [\clockvals_{\clocks}]_{\mathcal{C}_{\TA},d_{\TA}}$;
\item $\AP_{(\TA,\mathcal{C},d)} = \AP \cup \Phi_a^{d_{\TA}}(\mathcal{C}_{\TA})$; and
\item $\varepsilon$ to be a distinguished symbol not in $\Sigma$.
\end{itemize}
Then the \emph{region automaton} of $\TA$
with respect to $\mathcal{C}$ and $d$
is LTS
$R_{\mathcal{C},d} (\TA) = (Q_R, {\ttrans{}_R}, \Lab_R, Q_{0,R})$
over $\Sigma_\varepsilon = \Sigma \cup \{\varepsilon\}$ and $\AP_{(\TA,\mathcal{C},d)}$ where:
\begin{enumerate}
\item
$Q_R = \{(l, [v]_{\TA,\mathcal{C}, d}) \mid (l,v) \in \QTA\}$;
\item $(l, r) \ttrans{a}_R (l',r')$ iff one of the following hold:
    \begin{enumerate}
    \item $a \in \Sigma$, and for all $v \in r$ there is $v' \in r'$ such that $(l,v) \ttrans{a}_{\TA} (l',v')$, or
    \item $a = \varepsilon, l = l', r' = \tsucc(r)$ and for every $v \in r$ there exists $\delta_v \in \delays$ such that $v + \delta_v \in r'$ and $(l,v) \ttrans{\delta_v} (l,v+\delta_v)$;
    \end{enumerate}
\item $\Lab_R(l,r) = \Lab(l) \cup \{\phi \in \Phi_a^{d'}(\mathcal{C}') \mid r \models \phi\}$; and
\item $Q_{0,R} = \{(l_0, [v]_{\mathcal{C}', d'}) \mid (l_0,v) \in Q_{0,\TA}\}$.
\end{enumerate}
\end{definition}

Region automata differ significantly from timed automata; they are finite-state, and in lieu of transitions labeled by concrete time elapses they contain transitions labeled by $\varepsilon$.  Nevertheless, there are strong connections between between $\TA$ and $R_{\mathcal{C},d}(\TA)$, as the next lemma indicates.

\begin{lemma}[Properties of region automata]\label{lem:properties-of-region-automata}
Let $\TA = \genTA$ be a timed automaton over time-safe $\Sigma$ and clock-safe $\AP$, with semantics $\TS{\TA} = \ttsTA$, and let $R_{\mathcal{C},d}(\TA) = (Q_R, \ttrans{}_R, \Lab_R, Q_{0,R})$ be the region automaton of $\TA$ with respect to $\mathcal{C}$ and $d$.  Then the following hold.
\begin{enumerate}
\item
    Let $A \in \AP_{(\TA,\mathcal{C},d)}$ and $(l,r) \in Q_R$.  Then $A \in \Lab_R(l,r)$ iff $A \in \Lab_{\TA}(l,v)$ for every $v \in r$.
\item\label{sublem:properties-of-region-automata:diamond}
    Let $a \in \Sigma$.  Then $(l,r) \ttrans{a}_R (l',r')$ iff for all $v \in r$ there exists $v' \in r'$ such that $(l,v) \ttrans{a}_{\TA} (l',v')$.
\item\label{sublem:properties-of-region-automata:n-fold}
    Suppose that $(l,r) \,(\ttrans{\varepsilon}_R)^n\, (l',r')$.  Then $l = l'$ and $r' = \tsucc^n(r)$.
\item\label{sublem:properties-of-region-automata:varepsilon-trans}
    Let that $(l,v), (l',v') \in \QTA$.  Then there exists $\delta \in \delays$ such that $(l,v) \ttrans{\delta}_{\TA} (l',v')$ iff there exists $n \in \nats$ such that $(l, r) \,(\ttrans{\varepsilon}_R)^n\, (l',r')$, where $v \in r$ and $v' \in r'$.
\end{enumerate}
\end{lemma}
\begin{proof}
Standard in the literature; see e.g.~\cite{tripakis-analysis-of-2001}.  Result (\ref{sublem:properties-of-region-automata:varepsilon-trans}) follows by induction on the number of regions that the $\ttrans{\delta}_{\TA}$ crosses en route from $(l,v)$ to $(l',v')$; note that $l = l'$.\qedhere
\end{proof}

Later in this paper we will establish correspondences between the properties satisfied by timed automata and their associated region automata.  In so doing we will need to map between subsets of $\QTA$ and subsets of $Q_R$.  We will use the following two functions, $con \in 2^{Q_R} \to 2^{\QTA}$ (for ``concretize'') and $abs \in 2^{\QTA} \to 2^{Q_R}$ (for ``abstract'') for this purpose.
\begin{align*}
con (S_R)
&= \{(l,v) \in \QTA \mid \exists\, (l,r) \in S_R \colon v \in r\}
\\
abs (S_{\TA})
&= \{(l,r) \in Q_{(\TA,\mathcal{C},d)} \mid \exists\, (l,v) \in S_{\TA} \colon v \in r \}
\end{align*}
Function $con$ concretizes states in the region automaton by converting them into states in the semantics of $\TA$, while $abs$ abstracts semantic states into region-automaton states.  Note that for any $S_R \subseteq Q_R$, $abs(con(S_R)) = S_R$, while for any $S_{\TA} \subseteq \QTA$, $con(abs(S_{\TA})) \subseteq S_{\TA}$.  When $S_{\TA} = con(abs(S_{\TA}))$ we will refer to $S_{\TA}$ as \emph{saturated with respect to $R_{\mathcal{C},d}(\TA)$}.  We will use the following properties of saturation later in the paper.

\begin{lemma}[Properties of saturation]\label{lem:properties-of-saturation}
Let $\TA$ be a timed automaton over $\Sigma$ and $\AP$, with $\TS{\TA} = (\QTA,\ldots)$.  Also let $\mathcal{C} \subseteq \clocks$ and $d \in \nats$, and let $R_{\mathcal{C},d}(\TA) = (Q_R, \ldots).$
\begin{enumerate}
\item\label{sublem:properties-of-saturation:con-sat-sym}
    Suppose $S_R \subseteq Q_R$.  Then $con(S_R)$ is saturated with respect to $R_{\mathcal{C},d}(\TA)$.
\item\label{sublem:properties-of-saturtion:sat-prop}
    Suppose $S_{\TA} \subseteq \QTA$ is saturated with respect to $R_{\mathcal{C},d}(\TA)$.  Then $(l,v) \in S_{\TA}$ iff $(l,[v]_{\mathcal{C},d}) \in abs(S_{\TA})$.
\item\label{sublem:properties-of-saturation:closure-props}
    Suppose that $S_1, S_2 \subseteq \QTA$ are saturated with respect to $R_{\mathcal{C},d}(\TA)$.  Then so are $S_1 \cup S_2$ and $S_1 \setminus S_2$.
\end{enumerate}
\end{lemma}
\begin{proof}
Immediate from the definitions.\qedhere
\end{proof}

Classical region-graph constructions, such as the one in~\cite{ACD1993}, give more operational accounts of $[v]_{\mathcal{C},d}$ in terms of the relationships between the floor $\floor{v(x)}$ and fractional parts $v(x) - \floor{v(x)}$ of the clocks $x$.  Our construction yields equivalent automata:  what is key in either approach is that region automata have finitely many states and preserve the key behavioral properties given in Lemma~\ref{lem:properties-of-region-automata}.

\section{Timed Modal Mu-Calculi}\label{sec:timed-modal-mu-calculi}

This paper is devoted to a study of the relative expressive power of different timed modal mu-calculi in the context of timed automata.  This section introduces a reference timed modal mu-calculus, $\Lrelmunu$, together with a collection of mu-calculi introduced in the literature, then compares the expressive power of these logics.

\subsection{A Reference Timed Modal Mu-Calculus:  $\Lrelmunu$}

We first define the syntax and semantics of our reference timed mu-calculus, $\Lrelmunu$.  The definition is parameterized with respect to:
the clock structure $(\clocks,\clocksA,\clocksF)$;
$\Var$, a countably infinite set of \emph{propositional variables};
$\Sigma$, a nonempty time-safe sort $\Sigma$; and
$\AP$, a set of clock-safe atomic propositions.
Sets $\clocks$, $\Var$, $\Sigma$ and $\AP$ are assumed to be pairwise-disjoint.

\begin{definition}[$\Lrelmunu$ syntax]\label{def:Lrelmunu}
The formulas of the timed mu-calculus $\Lrelmunu$ are defined via the following grammar, where $A \in \AP_{\clocks} = \AP \cup \Phi_a(\clocks)$, $Y \in \Var$, $K \subseteq \Sigma$ and $z \in \clocksF$.
\[
\phi
::=     A
\mid    Y
\mid    \lnot \phi
\mid    \phi \lor \phi
\mid    \dia{K} \phi
\mid    \exists_\phi \phi
\mid    z.\phi
\mid    \mu Y.\phi
\]
Formulas must also satisfy the following restriction: in any formula $\mu Y.\phi$ every free occurrence of $Y$ in $\phi$ must be in the scope of an even number of negations in $\phi$.
We use $\Phirelmunu$ to denote the set of all $\Lrelmunu$ formulas.
If $\AP' \subseteq \AP_{\clocks}$ then we write $\Phirelmunu(\AP')$ for the subset of formulas in $\Phirelmunu$ whose atomic / clock constraints are restricted to those in $\AP'$.
\end{definition}

\noindent
Constructs $A, Y, \lnot$ and $\lor$ are standard; note that for any timed automaton $\TA$ over $\Sigma$ and $\AP$, $\AP_{\TA} = \AP \cup \Phi_a(\clocks)$ by definition.
The operator $\dia{K}$ is a \emph{labeled modality}, with $K$
being a subset of $\Sigma$.  Construct $\exists$ is a \emph{relativized} time modality:  $\exists_{\phi_1} \phi_2$ corresponds to an until operator interpreted only over time transitions.  Operator $z.$ denotes \emph{freeze quantification}.  In formula $z.\phi$ the freeze clock, $z$, is set to 0 and the body $\phi$ interpreted in this updated clock state.  Finally, $\mu Y.\phi$ represents a recursively defined formula that may be seen as the strongest solution to equation $Y = \phi$.

Semantically, $\Lrelmunu$ formulas are interpreted with respect to states in the timed transition system $\TS{\TA}$ associated with timed automaton $\TA$.  The semantic function, $\musemTAtheta{\phi}$, maps formula $\phi$ to the set of states in $\TS{\TA}$ that satisfy $\phi$, with environment $\theta$ giving the semantics for the propositional variables.  In what follows, if $\theta$ is such an environment, $Y \in \Var$, and $S$ is a set of states, then $\theta[Y := S]$ is the environment $\theta$ with the value associated with $Y$ updated to $S$ in the usual sense.  We now give the semantics of $\Lrelmunu$.

\begin{definition}[$\Lrelmunu$ semantics] \label{def:Lrelmunu-semantics}
Let $\TA$ be a timed automaton over $\Sigma$ and $\AP$, with $\TS{\TA} = \ttsTA$ the timed transition system associated with $\TA$ and $\theta \in \Var \to 2^{\QTA}$.  Also let $\phi \in \Phirelmunu$.  Then $\musemTAtheta{\phi} \subseteq \QTA$ is defined inductively as follows.
\begin{align*}
\musemTAtheta{A}
    &= \Lab_{\TA}(A)
    && A \in \AP_{\clocks}
\\
\musemTAtheta{Y}
    &= \theta(Y)
    && Y \in \Var
\\
\musemTAtheta{\lnot \phi}
    &= \QTA \setminus \musemTAtheta{\phi}
\\
\musemTAtheta{\phi_1 \lor \phi_2}
    &= \musemTAtheta{\phi_1} \cup \musemTAtheta{\phi_2}
\\
\musemTAtheta{\dia{K} \phi}
    &= \{q \in \QTA \mid
    \exists q' \in \QTA
    \colon q \ttrans{K}_{\TA} q'
    \land q' \in \musemTAtheta{\phi}\}
    && K \subseteq \Sigma
\\
\musemTAtheta{\exists_{\phi_1} \phi_2}
    &= \{q \in \QTA \mid
        \exists \delta \in \delays \colon
        \delta(q) \in \musemTAtheta{\phi_2} \;\land
\\
    &   \hspace{23pt} \forall \delta' \in \delays
                \colon \delta' < \delta
                \implies \delta'(q) \in (\musemTAtheta{\phi_1 \lor \phi_2}) \}
\\
\musemTAtheta{z.\phi}
    &= \{q \in \QTA \mid q[z:=0] \in \musemTAtheta{\phi}\}
    && z \in \clocksF
\\
\musemTAtheta{\mu Y.\phi}
    &= \bigcap\, \{S \subseteq \QTA \mid \musemTA{\phi}{\theta[Y := S]} \subseteq S\}
\end{align*}
If $q \in \musemTAtheta{\phi}$ then we say that $q$ \emph{satisfies} $\phi$ for $\TA$ and $\theta$ and write $q \models_{\TA,\theta} \phi$.
\end{definition}

We now comment on this definition.
First, $\exists$ captures a notion of ``until", in the following sense.  Because of properties of $\delta(q)$, it follows that if $\delta(q) \cap \musemTAtheta{\phi_2} \neq \emptyset$ then it must hold that $\delta(q) = \{ q' \}$ for some $q' \in \QTA$, and that this $q' \in \musemTAtheta{\phi_2}$.  Thus, if $q$ satisfies $\exists_{\phi_1} \phi_2$ then there must be $\delta \in \delays$ and $q'$ such that $q \ttrans{\delta} q'$ and $q'$ satisfies $\phi_2$, and such that for every $\delta' < \delta$ and $q''$ such that $q \ttrans{\delta'} q''$, $q''$ satisfies either $\phi_1$ or $\phi_2$.
Second, $z.\phi$ is satisfied by $q$ iff the clock $z$, which is a freeze clock and thus cannot appear in $\TA$, when reset to $0$ in $q$, makes $\phi$ true.  Finally, formula $\mu Y.\phi$ is given meaning based on the Tarski-Knaster characterization of least fixpoints over complete lattices~\cite{tarski-a-lattice-theoretical-1955}.  In this case, the complete lattice in question is the subset lattice $2^{\QTA}$ ordered by $\subseteq$ and with union and intersection being the least-upper-bound / greatest-lower-bound operations.  Now let $Y \in \Var$ and $\phi \in \Phirelmunu$.  In the lattice, the semantic function $\musemTAtheta{Y.\phi} \in 2^{\QTA} \to 2^{\QTA}$ defined as
$
\musemTAtheta{Y.\phi}(S) = \musemTA{\phi}{\theta[Y:=S]}
$
is monotonic for every $\phi$, and therefore possesses a unique minimum fixpoint.  It turns out that $\musemTAtheta{\mu Y.\phi}$ coincides with this fixpoint.

Regarding the dual operators of $\Lrelmunu$, as usual, $\land$ is the dual of $\lor$:  $\phi_1 \land \phi_2 = \lnot (\lnot \phi_1 \lor \lnot \phi_2)$.  The dual of $\dia{K}$ is $[K]$:  $[K] \phi = \lnot (\dia{K} \lnot \phi)$.  The dual of $\exists$ is written as $\forall$: $\forall_{\phi_1} \phi_2 = \lnot(\exists_{\lnot\phi_1} (\lnot \phi_2))$, and represents a version of the \emph{release} operator from temporal logic.  The freeze operator $z.$ is self-dual, since $\lnot(z.\phi)$ is equivalent to $z.(\lnot \phi)$.  We also use $\nu Y.$\/ for the dual of $\mu Y.$:  $\nu Y.\phi = \lnot(\mu Y.(\lnot \phi[Y := \lnot Y]))$, where $\phi[Y := \lnot Y]$ represents the formula obtained from $\phi$ by replacing all free occurrences of $Y$ by $\lnot Y$ in the usual fashion.    Semantically, $\nu Y.\phi$ coincides with the greatest fixpoint of function $\musemTAtheta{Y.\phi}$ introduced above and can be characterized as $\bigcup \{S \subseteq \QTA \mid S \subseteq \musemTA{\phi}{\theta[Y := S]}\}$.  We will freely use these dual operators in $\Lrelmunu$ formulas.  It can also be shown that, if these dual operators are included in formulas, any formula can be rewritten into \emph{positive normal form}, in which $\lnot$ is only applied to elements of $\AP$ and free occurrences of propositional variables.

We close this section on $\Lrelmunu$ by establishing that the semantics of its formulas are preserved, in a very precise fashion, by the region-automaton construction given in Section~\ref{sec:region-automata}.  We first extend the bound and clock-set functions defined for clock constraints (Definition~\ref{def:cxcons}) to $\Lrelmunu$ formulas as follows.
\begin{definition}[Bounds, clocks of $\Lrelmunu$ formulas]
Let $\phi \in \Phirelmunu$ be a formula in $\Lrelmunu$.  Then the \emph{bound}, $\bound(\phi)$, and clocks, $cs(\phi) \subsetneq \clocks$, of $\phi$ are defined as follows.
\begin{align*}
    &(\bound(\phi), cs(\phi))\\
    &=
    \begin{cases}
        (0, \emptyset)              & \text{if $\phi \in \AP \cup \Var$} \\
        (\bound(\phi), cs(\phi))    & \text{if $\phi \in \Phi_a(\clocks)$} \\
        (\bound(\phi'), cs(\phi'))  & \text{if $\phi = \lnot \phi'$}\\
        (\max \{\bound(\phi_1), \bound(\phi_2)\}, cs(\phi_1) \cup cs(\phi_2))
                                    & \text{if $\phi = \phi_1 \lor \phi_2$} \\
        (\bound(\phi'), cs(\phi'))  & \text{if $\phi = \dia{K}\phi'$} \\
        (\max \{\bound(\phi_1), \bound(\phi_2)\}, cs(\phi_1) \cup cs(\phi_2))
                                    & \text{if $\phi = \exists_{\phi_1} \phi_2$} \\
        (\bound(\phi'), \{ z \} \cup cs(\phi'))
                                    & \text{if $\phi = z.\phi'$} \\
        (\bound(\phi'), cs(\phi'))   & \text{if $\phi = \mu Y.\phi$}
    \end{cases}
\end{align*}
\end{definition}

Intuitively, $\bound(\phi)$ returns the largest constant used in a clock constraint in $\phi$, while $cs(\phi)$ gives the (finite) set of clocks referenced in $\phi$.
We can now define $R_\phi(\TA)$, the region automaton for $\TA$ \emph{relativized} to $\phi$, and its atomic proposition set $\AP_{\TA,\phi}$, as follows.
\begin{definition}[Relativized region automaton and atomic proposition set]
Let $\TA = \genTA$ be a timed automaton over time-safe $\Sigma$ and clock-safe $\AP$, let $\phi$ be a formula in $\Lrelmunu$, and let $\mathcal{C} = cs(\phi)$ and $d = \bound(\phi)$. Then, following Definition~\ref{def:region-automaton}:
\begin{enumerate}
\item
    $[v]_{\TA,\phi} = [v]_{\TA,\mathcal{C},d}$ is the \emph{equivalence class}, or \emph{region, of $v \in \clockvals_{\clocks}$ induced by $\TA$ and $\phi$};
\item
    $[\clockvals_{\clocks}]_{\TA,\phi} = [\clockvals_{\clocks}]_{\TA,\mathcal{C},d}$ are the \emph{regions induced by $\TA$ and $\phi$};
\item
    $\AP_{\TA,\phi} = \AP_{\TA,\mathcal{C},d}$ is the \emph{atomic proposition set of $\TA$ relativized to $\phi$}; and
\item
    $R_\phi(\TA) = R_{\mathcal{C},d}(\TA)$ is the \emph{region automaton of $\TA$ relativized to $\phi$}.
\end{enumerate}
\end{definition}
Note that every atomic clock constraint appearing in $\TA$ or $\phi$ is in $\AP_{\TA,\phi}$.  Also, if $\phi$ and $\phi'$ are such that $\bound(\phi) = \bound(\phi')$ and $cs(\phi) = cs(\phi')$, then $\AP_{\TA,\phi} = \AP_{\TA,\phi'}$ and $R_{\TA,\phi} = R_{\TA,\phi'}$ for any $\TA$.

We now give a \emph{symbolic semantics} for $\Lrelmunu$ formulas in the spirit of Larousinie et al.~\cite{laroussinie-from-timed-1995} and Bouyer et al.~\cite{bouyer-timed-modal-2011} that interprets $\Lrelmunu$ formulas with respect to the region automata constructed from timed automata.
\begin{definition}[Symbolic semantics of $\Lrelmunu$]
Fix time-safe $\Sigma$ and clock-safe $\AP$, let $\phi \in \Phirelmunu$, and let
\begin{align*}
\TA         &= \genTA \\
\TS{\TA}    &= \ttsTA \\
R_\phi(\TA) &= (Q_R, \ttrans{}_R, \Lab_R,Q_{0,R})
\end{align*}
be a timed automaton over $\Sigma$ and $\AP$, the timed-automaton semantics of $\TA$, and the region automaton of $\TA$ relativized to $\phi$, respectively.
Then the \emph{symbolic semantics} of $\phi$, $\musemsymTAtheta{\phi} \subseteq Q_{R}$, is defined to be $\musemrelphiTAtheta{\phi}$, where  $\theta \in \Var \to 2^{\QTA}$ is such that, for all $Y \in \Var$, $\theta(Y) \subseteq \QTA$ is \emph{saturated} with respect to $R_\phi(\TA)$, and $\musemrelphiTAtheta{\gamma} \subseteq Q_R$ for $\gamma \in \Phirelmunu(\AP_{\TA,\phi})$ is given inductively as follows.
\begin{align*}
    \musemrelphiTAtheta{A}
        &= \Lab_R(A) \quad \text{for $A \in \AP_{\TA,\phi}$}
        \\
    \musemrelphiTAtheta{Y}
        &= abs(\theta(Y))
        \\
    \musemrelphiTAtheta{\lnot \gamma}
        &= Q_R \setminus \musemrelphiTAtheta{\gamma}
        \\
    \musemrelphiTAtheta{\gamma_1 \lor \gamma_2}
        &= \musemrelphiTAtheta{\gamma_1} \cup \musemrelphiTAtheta{\gamma_2}
        \\
    \musemrelphiTAtheta{\dia{K}\gamma}
        &= \{ q \in Q_R \mid \exists q' : q \ttrans{K}_R q' : q' \in \musemrelphiTAtheta{\gamma} \}
        \\
    \musemrelphiTAtheta{\exists_{\gamma_1}\gamma_2}
        &= \{ (l,r) \in Q_R \mid
        \\
        & \qquad \exists n \in \mathbb{N} \colon (l,r) \,(\ttrans{\varepsilon})^n\, (l,\tsucc^n(r)) \land (l,\tsucc^n(r)) \in \musemrelphiTAtheta{\gamma_2}
        \\
        & \qquad \land \forall m \in \nats \colon m < n \colon (l, \tsucc^m(r)) \in \musemrelphiTAtheta{\gamma_1} \} \\
    \musemrelphiTAtheta{z.\gamma}
        &= \{ (l,r) \in Q_R \mid (l, r[z := 0]) \in \musemrelphiTAtheta{\gamma} \}
        \\
    \musemrelphiTAtheta{\mu Y . \gamma}
        &= \bigcap \{ S \subseteq Q_R \mid \musemrelphiTA{\gamma}{\theta[Y := con(S)]} \subseteq S \}
\end{align*}
\end{definition}

Intuitively, $\musemsymTAtheta{\phi}$ returns a set states in the region automaton $R_\phi(\TA)$ associated with $\TA$ and $\phi$.  This set is guaranteed to be finite, since $Q_R$, the set of states in $R_\phi(\TA)$, is finite.  The semantics of $Y$ uses the $abs$ function to convert $\theta(Y) \subseteq \QTA$ into the corresponding subset of $Q_R$.
Similarly, the semantics of $\mu Y.\phi$ uses $con$ to concretize the set of region-automaton states $S$ into the corresponding subset of $\QTA$ when updating $\theta$.
The definition also uses the $(\ttrans{\varepsilon}_R)^n$ relation; note that if $(l,r) \;(\ttrans{\varepsilon}_R)^n\; (l, \tsucc^n(r))$ then Lemma~\ref{lem:properties-of-region-automata}(\ref{sublem:properties-of-region-automata:n-fold}) guarantees that $(l,r) \;(\ttrans{\varepsilon})^m\; (l,\tsucc^m(r))$ for all $m < n$.

The definition of $\musemsymTAtheta{\phi}$ uses an intermediate notion, $\musemrelphiTAtheta{\gamma}$, and we comment on the reason for this here. The symbolic semantics interprets $\phi$ with respect to a region automaton, $R_\phi(\TA)$, that is constructed from both $\TA$ and $\phi$.  A traditional inductive definition for $\musemsymTAtheta{\phi}$ would in general have to use different region automata for subformulas of $\phi$ than for $\phi$ itself.  The use of $\phi$ as a parameter in $\musemrelphiTAtheta{\gamma}$ solves this problem by fixing $\phi$ as the formula for constructing the region automaton, and then interpreting only formulas with respect to this automaton whose atomic propositions and clock constraints are consistent with those found in $\phi$.

The following correspondence between $\musemTAtheta{\phi}$ and $\musemsymTAtheta{\phi}$ can now be established.

\begin{lemma}\label{lem:region-graph-consistency-Lrelmunu}
Let $\TA = \genTA$ be a timed automaton over time-safe $\Sigma$ and clock-safe $\AP$, let $\phi \in \Phirelmunu$ be a formula in $\Lrelmunu$, and let $\theta$ be such that $\theta(Y)$ is saturated with respect to $R_\phi(\TA)$ for all $Y \in \Var$.  Then $\musemTAtheta{\phi} = con (\musemsymTAtheta{\phi})$.
\end{lemma}

\begin{proof}
Fix $\Sigma$, $\AP$, $\TA = \genTA$, $\phi$, $R_\phi(\TA) = (Q_R, \ttrans{}_R, \mathcal{L}_R, Q_{0,R})$.
We actually prove that for all $\gamma \in \Phirelmunu(\AP_{\TA,\phi})$, and all $\theta$ such that $\theta(Y)$ is saturated with respect to $R_\phi(\TA)$ for every $Y \in \Var$, $\musemTAtheta{\gamma} = con (\musemrelphiTAtheta{\gamma})$.
The desired result then follows since $\phi \in \Phirelmunu(\AP_{\TA,\phi})$ and $\musemsymTAtheta{\phi}$ is defined as $\musemrelphiTAtheta{\phi}$.
Below, for terminological conciseness, we say that a set $S_{\TA} \subseteq \QTA$ is saturated when it is saturated with respect to $R_\phi(\TA)$ and call $\theta \in \Var \rightarrow 2^{\QTA}$ saturated if $\theta(Y)$ is saturated for all $Y \in \Var$.

The proof now proceeds by induction on the structure of $\gamma \in \Phirelmunu(\AP_{\TA,\phi})$.
So fix $\gamma$; the induction hypothesis states that for all (strict) subformulas $\gamma'$ of $\gamma$ and saturated $\theta$, $\musemTAtheta{\gamma'} = con(\musemrelphiTAtheta{\gamma'})$.  Now fix saturated $\theta$; we must prove that $\musemTAtheta{\gamma} = con (\musemrelphiTAtheta{\gamma})$.  The arguments uses a case analysis on $\phi$.

\begin{description}
\item[$\gamma \in \AP_{\TA,\phi}$.]
    There are two subcases to consider.  In the first, $\gamma \in \AP$.  From the definitions it can be seen that the following hold.
    \begin{align*}
    \musemTAtheta{\gamma}
    &= \{l \in L \mid \gamma \in \Lab(l) \} \times \clockvals_{\clocks}
    \\
    \musemrelphiTAtheta{\gamma}
    &= \{l \in L \mid \gamma \in \Lab(l) \} \times [\clockvals_{\clocks}]_{\TA,\phi}
    \end{align*}
    That $\musemTA\theta{\gamma} = con(\musemsymTAtheta{\gamma})$ is immediate.  In the second subcase $\gamma \in \Phi_a^{d'}(\mathcal{C}')$, where $d' = \max\{\bound(\TA),\bound(\phi)\}$ and $\mathcal{C}' = \CX \cup cs(\phi)$.  From the definitions we observe the following.
    \begin{align*}
    \musemTAtheta{\gamma}
    &= L \times \{v \in \clockvals_{\clocks} \mid v \models \gamma\}
    \\
    \musemrelphiTAtheta{\gamma}
    &= L \times \{r \in [\clockvals_{\clocks}]_{\TA,\phi} \mid r \models \gamma \}
    \end{align*}
    It is clear that $abs(\musemTAtheta{\gamma}) =  \musemsymTAtheta{\gamma}$.  Moreover, $L \times \{v \in \clockvals_{\clocks} \mid v \models \phi\}$ is saturated, and thus $\musemTAtheta{\phi} = con(abs(\musemTAtheta{\phi})) = con(\musemsymTAtheta{\phi})$.
\item[$\gamma \in \Var$.]
    In this case, $\theta$ is saturated, and thus $\theta(\gamma) = con(abs(\theta(\gamma)))$.  Then
    \[
    \musemTAtheta{\gamma}
    = \theta(\gamma)
    = con \left( abs \left( \theta(\gamma) \right) \right)
    = con \left( \musemrelphiTAtheta{\gamma} \right).
    \]
\item[$\gamma = \lnot\gamma'$.]
    In this case the induction hypothesis guarantees that for all satured $\theta$, $\musemTAtheta{\gamma'} = con(\musemrelphiTAtheta{\gamma'})$.  We must show the same result for $\gamma$.  So fix saturated $\theta$. We begin by noting that since $\musemTAtheta{\gamma'} = con(\musemrelphiTAtheta{\gamma'})$, Lemma~\ref{lem:properties-of-saturation}(\ref{sublem:properties-of-saturation:con-sat-sym}) guarantees that $\musemTAtheta{\gamma'}$ is saturated.  It then follows from the same lemma that $\QTA \setminus \musemTAtheta{\gamma'}$ is saturated.  Also, $abs(\QTA \setminus \musemTAtheta{\gamma'}) = Q_R \setminus \musemrelphiTAtheta{\gamma'}$.  These observations give the following.
    \begin{align*}
    \musemTAtheta{\gamma}
    &=
    \QTA \setminus \musemTAtheta{\gamma'}
    =
    con(abs(\QTA \setminus \musemTAtheta{\gamma'}))
    \\
    &=
    con(Q_R \setminus \musemrelphiTAtheta{\gamma'})
    =
    con \left( \musemrelphiTAtheta{\gamma} \right)
    \end{align*}
\item[$\gamma = \gamma_1 \lor \gamma_2$.]
    In this case the induction hypothesis guarantees that for all saturated $\theta$, $\musemTAtheta{\gamma_1} = con(\musemrelphiTAtheta{\gamma_1})$ and
    $\musemTAtheta{\gamma_2} = con(\musemrelphiTAtheta{\gamma_2})$.
    We must show the same result for $\gamma$.  So fix saturated $\theta$.  We first note that for any sets $S_1, S_2 \subseteq Q_R$, $con(S_1 \cup S_2) = con(S_1) \cup con(S_2)$.  We now reason as follows.
    \begin{align*}
    \musemTAtheta{\gamma}
    &=
    \musemTAtheta{\gamma_1} \cup \musemTAtheta{\gamma_2}
    =
    con \left( \musemrelphiTAtheta{\gamma_1} \right) \cup
    con \left( \musemrelphiTAtheta{\gamma_2} \right)
    \\
    &=
    con \left( \musemrelphiTAtheta{\gamma_1} \cup \musemrelphiTAtheta{\gamma_2} \right)
    = con \left( \musemrelphiTAtheta{\gamma} \right)
    \end{align*}
\item[$\gamma = \dia{K}\gamma'$.]
    In this case the induction hypothesis guarantees that for all saturated $\theta$, $\musemTAtheta{\gamma'} = con(\musemrelphiTAtheta{\gamma'})$.  We must show the same result for $\gamma$.  So fix saturated $\theta$.
    We begin by noting that Lemma~\ref{lem:properties-of-saturation}(\ref{sublem:properties-of-saturation:con-sat-sym}) guarantees that $\musemTAtheta{\gamma'}$ is saturated. Based on the semantics of $\Lrelmunu$ and Lemma~\ref{lem:properties-of-region-automata}(\ref{sublem:properties-of-region-automata:diamond}) it is the case that $\musemTAtheta{\gamma}$ is also saturated.  This lemma also ensures that $abs(\musemTAtheta{\gamma}) = \musemrelphiTAtheta{\gamma}.$  Therefore,
    \[
    \musemTAtheta{\gamma}
    =
    con(abs(\musemTAtheta{\gamma}))
    = con \left( \musemrelphiTAtheta{\gamma} \right).
    \]
\item[$\gamma = \exists_{\gamma_1} \gamma_2$.]
    In this case the induction hypothesis guarantees that for all saturated $\theta$, $\musemTAtheta{\gamma_1} = con(\musemrelphiTAtheta{\gamma_1})$ and
    $\musemTAtheta{\gamma_2} = con(\musemrelphiTAtheta{\gamma_2})$.
    These two facts together also imply that $\musemTAtheta{\gamma_1 \lor \gamma_2} = con(\musemrelphiTAtheta{\gamma_1 \lor \gamma_2})$.
    We must show that for all saturated $\theta$, $\musemTAtheta{\gamma} = con(\musemrelphiTAtheta{\gamma})$.
    To this end, fix saturated $\theta$.
    We first show that $\musemTAtheta{\gamma} \subseteq con(\musemrelphiTAtheta{\gamma})$.  So suppose $(l,v) \in \musemTAtheta{\gamma}$.
    It suffices to establish that $(l,r_v) \in \musemrelphiTAtheta{\gamma}$, where $r_v$ here is short-hand for region $[v]_{\TA,\phi}$, as in this case
    \[
    (l,v)
    \in con \left( \left\{ (l,r_v) \right\} \right)
    \subseteq con \left(\musemrelphiTAtheta{\gamma} \right).
    \]Since $(l,v) \in \musemTAtheta{\gamma}$ there is a $\delta_v \in \delays$ such that the following hold.
    \begin{enumerate}
        \item $\delta_v(l,v) \cap \musemTAtheta{\gamma_2} \neq \emptyset$
        \item For all $\delta' \in \delays$ such that $\delta' < \delta_v$, $\delta'(l,v) \cap \musemTAtheta{\gamma_1 \lor \gamma_2} \neq \emptyset$.
    \end{enumerate}
    To establish that $(l,r_v) \in \musemrelphiTAtheta{\gamma}$ we must show that there is $n \in \nats$ such that:
    \begin{itemize}
        \item $(l, r_v) \,(\ttrans{\varepsilon})^n\, (l, \tsucc^n(r_v))$;
        \item $(l, \tsucc^n(l, r_v)) \in \musemrelphiTAtheta{\gamma_2}$; and
        \item for all $m \in \nats$ with $m < n$, $(l, succ^m(r_v)) \in \musemrelphiTAtheta{\gamma_1}$.
    \end{itemize}
    Because of the properties of $\ttrans{\delta}$ and $\delta_v$, Lemma~\ref{lem:properties-of-region-automata}(\ref{sublem:properties-of-region-automata:varepsilon-trans}) guarantees the existence of
    $n' \in \nats$ such that $(l, r_v) \,(\ttrans{\varepsilon})^{n'}\, (l,\tsucc^{n'}(r_v))$, $(l, \tsucc^{n'}(r_v)) \in \musemrelphiTAtheta{\gamma_2}$
    and for all $m' < n'$, $(l, \tsucc^{m'}(r_v)) \in \musemrelphiTAtheta{\gamma_1 \lor \gamma_2}$.
    Let $n$ be the smallest number such that $(l,\tsucc^n(l,r_v)) \in \musemrelphiTAtheta{\gamma_2}$; note that $n \leq n'$.  It follows that for all $m < n$, $(l, \tsucc^m(l,r_v)) \in \musemrelphiTAtheta{\gamma_1}$, thus establishing that $(l,r_v) \in \musemrelphiTAtheta{\gamma}$.

    We now prove that $con(\musemrelphiTAtheta{\gamma}) \subseteq \musemTAtheta{\gamma}$.
    It suffices to show that for every $(l,r) \in \musemrelphiTAtheta{\gamma}$ and $v \in r$, $(l,v) \in \musemTAtheta{\gamma}$.  The proof uses an inductive argument on $n(l,r) \in \nats$, which is the constant, guaranteed by the definition of $\musemrelphiTAtheta{\gamma}$, such that $(l,r) \,(\ttrans{\varepsilon})^{n(l,r)}\, (l, \tsucc^{n(l,r)}(r))$, $(l, \tsucc^{n(l,r)}(r)) \in \musemrelphiTAtheta{\gamma_2}$, and for all $m < n(l,r)$, $(l, \tsucc^m(r)) \in \musemrelphiTAtheta{\gamma_1}$.
    To avoid confusion, in the rest of this (inner) inductive argument we will refer to the corresponding induction hypothesis as the inner induction hypothesis, to distinguish it from what we will call the outer induction hypothesis, which handles structural reasoning about the semantics of formulas.
    For the base case of this inner inductive argument, consider $(l,r) \in \musemrelphiTAtheta{\gamma}$ such that $n(l,r) = 0$.
    In this case $(l,r) \in \musemrelphiTAtheta{\gamma_2}$, and the outer induction hypothesis then guarantees that for all $v \in r$, $(l,v) \in \musemTAtheta{\gamma_2}$.
    To conclude that for all $v \in r$, $(l,v) \in \musemTAtheta{\gamma}$ we must give, for each $v$, a $\delta_v$ such that $\delta_v(l,v) \cap \musemTAtheta{\gamma_2} \neq \emptyset$ and such that for all $\delta' \in \delays$ with $\delta' < \delta_v$, $\delta'(l,v) \cap \musemTAtheta{\gamma_1} \neq \emptyset$.
    Setting $\delta_v = 0$ satisfies these requirements.

    For the induction step of the inner inductive argument, assume $n' \geq 0$.  The inner induction hypothesis guarantees that for all $(l',r') \in \musemrelphiTAtheta{\gamma}$ such that $n(l',r') = n'$ and $v' \in r'$, $(l',v') \in \musemTAtheta{\gamma}$.  We must now prove this result for all $(l,r) \in \musemrelphiTAtheta{\gamma}$ such that $n(l,r) = n'+1$.  So fix such an $(l,r)$. It follows that $(l,r) \,(\ttrans{\varepsilon})^{n'+1}\, (l,\tsucc^{n'+1}(r))$, that $(l, \tsucc^{n'+1}(r)) \in \musemrelphiTAtheta{\gamma_2}$, and that for all $m \in \nats$ such that $m < n'+1$, $(l,\tsucc^m(r)) \in \musemrelphiTAtheta{\gamma_1}$.
    Now pick $v \in r$.
    We must construct $\delta_v \in \delays$ such that $\delta_v(l,v) \cap \musemTAtheta{\gamma_2} \neq \emptyset$ and such that for all $\delta' \in \delays$ such that $\delta' < \delta$, $\delta'(l,v) \cap \musemTAtheta{\phi_1 \lor \phi_2} \neq \emptyset$.
    Since $(l,r) \in \musemrelphiTAtheta{\gamma}$ and $(l,r) \,(\ttrans{\varepsilon})^{n'+1}\, (l,\tsucc^{n'+1}(r))$ it follows that
    $(l,r) \in \musemrelphiTAtheta{\gamma_1}$,
    $(l,r) \ttrans{\varepsilon} (l,\tsucc(r)) \,(\ttrans{\varepsilon})^{n'}\, (l, \tsucc^{n'+1}(r))$,
    $(l,\tsucc(r)) \in \musemrelphiTAtheta{\gamma}$,
    and $n(l,\tsucc(r)) = n'$.
    The inner induction hypothesis then guarantees that for all $v' \in \tsucc(r)$, $(l,v') \in \musemTAtheta{\gamma}$.
    Based on the definitions of $\ttrans{\varepsilon}$ and $\tsucc(r)$, it also follows that there is $\delta_1 \in \delays$ such that $v + \delta_1 \in \tsucc(r)$ and $(l,v) \ttrans{\delta_1} (l, v+\delta_1)$, and such that for all $\delta'_1 \in \delays$ such that $\delta'_1 < \delta_1$, $v + \delta'_1 \in r \cup \tsucc(r)$.
    Also, since $(l,v+\delta_1) \in \tsucc(r)$ it follows from the inner induction hypothesis that $(l, v+\delta_1) \in \musemrelphiTAtheta{\gamma}$, meaning there exists $\delta_{2} \in \delays$ such that $\delta_{2}(l, v+\delta_1) \cap \musemTAtheta{\gamma_2} \neq \emptyset$ and such that for all $\delta'_2 \in \delays$ with $\delta'_2 < \delta_2$, $\delta'_2(l,v+\delta_1) \in \musemTAtheta{\gamma_1 \lor \gamma_2}$.
    Now take $\delta_v = \delta_1 + \delta_{2}$.  It is straightforward to see that $\delta_v(l,v) = \delta_{2}(l,v+\delta_1)$ and thus $\delta_v(l,v) \cap \musemTAtheta{\gamma_2} \neq \emptyset$.  We must now show that for all $\delta' \in \delays$ such that $\delta' < \delta_v$, $\delta'(l,v) \in \musemTAtheta{\gamma_1 \lor \gamma_2}$.  There are two cases to consider.  In the first, $\delta_1 \leq \delta' < \delta_v$.  In this case it is easy to see that $\delta'(l,v) = (\delta_v - \delta')(l,v + \delta_1)$, and the result holds immediately.  In the second case, $0 \leq \delta' < \delta_1$.
    Pick such a $\delta'$;
    we know that $\delta'(l,v) = \{ (l, v+\delta') \}$ and that
    either $(l,v+\delta') \in r$ or $(l,v+\delta') \in \tsucc(r)$.
    In the former case we have established that $(l,v+\delta') \in \musemTAtheta{\gamma_1}$, and thus $\delta'(l,v) \cap \musemTAtheta{\gamma_1 \lor \gamma_2} \neq \emptyset$.
    In the latter case the inner induction hypothesis allows us to conclude that $(l,v+\delta') \in \musemTAtheta{\gamma}$; it is easy to see in this case that $(l,v+\delta') \in \musemTAtheta{\gamma_1 \lor \gamma_2}$, and thus $\delta'(l,v+\delta) \cap \musemTAtheta{\gamma_1 \lor \gamma_2} \neq \emptyset$.
    Thus $\delta_v$ has the required properties, $(l,v) \in \musemTAtheta{\gamma}$, and the inner inductive argument is complete.

\item[$\gamma = z.\gamma'$.]
    Here the induction hypothesis guarantees that for all saturated $\theta$, $\musemTAtheta{\gamma'} = con(\musemrelphiTAtheta{\gamma'})$.  We must show the same result for $\gamma$.  So fix saturated $\theta$.
    We first note that Lemma~\ref{lem:properties-of-saturation}(\ref{sublem:properties-of-saturation:con-sat-sym}) guarantees that $\musemTAtheta{\gamma'}$ is saturated. Based on the semantics of $\Lrelmunu$ and Lemma~\ref{lem:preservation-of-eq}(\ref{sublem:preservation-of-eq:reset}) it is the case that $\musemTAtheta{\gamma}$ is also.  This lemma moreover ensures that $abs(\musemTAtheta{\gamma}) = \musemrelphiTAtheta{\gamma}.$  Therefore,
    \[
    \musemTAtheta{\gamma}
    =
    con(abs(\musemTAtheta{\gamma}))
    = con \left( \musemrelphiTAtheta{\gamma} \right).
    \]
\item[$\gamma = \mu Y.\gamma'$.]
    Before proving this case we first remark on a simple result, which henceforth we will refer to as \textbf{IR} (``intersection result''), from set theory.
    \begin{quote}
        \textbf{(IR)}  Let $X$ be a set, and let $\mathcal{F}_1, \mathcal{F}_2 \subseteq 2^X$ be such that for every $S_2 \in \mathcal{F}_2$ there is $S_1 \in \mathcal{F}_1$ such that $S_1 \subseteq S_2$.  Then $\bigcap \mathcal{F}_1 \subseteq \bigcap \mathcal{F}_2$.
    \end{quote}
    \textbf{IR} in effect says that if every member set of $\mathcal{F}_2$ has a subset in $\mathcal{F}_1$ then the intersection of the sets in $\mathcal{F}_1$ is a subset of the intersection of $\mathcal{F}_2$.  A simple consequence of this result is that if $\mathcal{F}_2 \subseteq \mathcal{F}_1$ then $\bigcap \mathcal{F}_1 \subseteq \bigcap \mathcal{F}_2$.  The proof of \textbf{IR} is straightforward.

    We now continue with our inductive argument.  The induction hypothesis guarantees that for all saturated $\theta$, $\musemTAtheta{\gamma'} = con(\musemrelphiTAtheta{\gamma'})$.  We must show the same result for $\gamma$.  So fix saturated $\theta$.
    We begin by noting that Lemma~\ref{lem:properties-of-saturation}(\ref{sublem:properties-of-saturation:con-sat-sym}) guarantees that $con(S_R)$ is saturated for any $S_R \subseteq Q_R$; this means $\theta[Y := con(S_R)]$ is saturated as well.
    (For notational simplicity, in what follows we abbreviate $\theta[Y := S]$, where $S \subseteq \QTA$, as $\theta[S]$.)
    From the induction hypothesis we have that
    $\musemTA{\gamma'}{\theta[con(S_R)]} =
    con(\musemrelphiTA{\gamma'}{\theta[con(S_R)]})$
    for any $S_R \subseteq Q_R$.
    We prove that $\musemTAtheta{\gamma} = con(\musemrelphiTAtheta{\gamma})$ by showing that $\musemTAtheta{\gamma} \subseteq con(\musemrelphiTAtheta{\gamma})$ and $con(\musemrelphiTAtheta{\gamma}) \subseteq \musemTAtheta{\gamma}$.
    In what follows we use the following families of subsets.
    \begin{align*}
    \mathcal{F}_{\TA}
    &=  \left\{S_{\TA} \subseteq \QTA \mid \musemTA{\gamma'}{\theta[S_{\TA}]} \subseteq S_{\TA} \right\}
    &&  \subseteq 2^{\QTA}
    \\
    \mathcal{F}_R
    &=  \left\{ S_R \subseteq Q_R \mid \musemrelphiTA{\gamma'}{\theta[con(S_R)]} \subseteq S_R \right\}
    &&  \subseteq 2^{Q_R}
    \\
    con(\mathcal{F}_R)
    &=  \left\{ con(S_R) \mid S_R \in \mathcal{F}_R \right\}
    &&  \subseteq 2^{\QTA}
    \end{align*}
    By definition, $\musemTAtheta{\gamma} = \bigcap \mathcal{F}_{\TA}$ while $\musemrelphiTAtheta{\gamma} = \bigcap \mathcal{F}_R$.

    To prove that $\musemTAtheta{\gamma} \subseteq con(\musemrelphiTAtheta{\gamma})$ we will show that $con(\musemrelphiTAtheta{\gamma}) = \bigcap con(\mathcal{F}_R)$ and that $con(\mathcal{F}_R) \subseteq \mathcal{F_\TA}$.
    \textbf{IR} will then guarantee that $\bigcap \mathcal{F}_{\TA} \subseteq \bigcap con(\mathcal{F}_R)$, and thus $\musemTAtheta{\gamma} \subseteq con(\musemrelphiTAtheta{\gamma})$.
    That $con(\musemrelphiTAtheta{\gamma}) = \bigcap con(\mathcal{F}_R)$ is a consequence of Lemma~\ref{lem:properties-of-saturation}(\ref{sublem:properties-of-saturation:closure-props}),
    which guarantees that for any family $\mathcal{F} \subseteq 2^{Q_R}$, $con(\bigcap \mathcal{F}) = \bigcap \{con(S) \mid S \in \mathcal{F} \}$.
    To see that $con(\mathcal{F}_R) \subseteq \mathcal{F}_{\TA}$ it suffices to show that for any $S_R \in \mathcal{F}_R$, $\musemTA{\gamma'}{\theta[con(S_R)]} \subseteq con(S_R)$,
    for it then follows that $con(S_R) \in \mathcal{F}_{\TA}$.
    From the induction hypothesis we have $\musemTA{\gamma'}{\theta[con(S_R)]} = con(\musemrelphiTA{\gamma'}{\theta[con(S_R)]})$.
    Since $\musemrelphiTA{\gamma'}{\theta[con(S_R)]} \subseteq S_R$, the definition of $con$ gives that $con(\musemrelphiTA{\gamma'}{\theta[con(S_R)]}) \subseteq con(S_R)$, thus establishing the desired result.

    We now show that $con(\musemrelphiTAtheta{\gamma}) \subseteq \musemTAtheta{\gamma}$.
    Since we have already established that $con(\musemrelphiTAtheta{\gamma}) = \bigcap con(\mathcal{F}_R)$, it suffices to show that for all $S_{\TA} \in \mathcal{F}_{\TA}$ there is $S_R \subseteq Q_R$ such that $con(S_R) \subseteq S_{\TA}$,
    as then \textbf{IR} guarantees that $\bigcap con(\mathcal{F}_R) \subseteq \bigcap \mathcal{F}_{\TA} = \musemTAtheta{\gamma}$.
    This observation in turn follows if for any $S_{\TA} \in \mathcal{F}_{\TA}$ we can construct a saturated set $S'_{\TA} \in \mathcal{F}_{\TA}$ such that $S'_{\TA} \subseteq S_{\TA}$.  To see why, note that for every saturated $S'_{\TA} \in \mathcal{F}_{\TA}$ there must be $S'_R \subseteq Q_R$ such that $S'_{\TA} = con(S'_R)$; since the induction hypothesis guarantees that $\musemTA{\gamma'}{\theta[S'_{\TA}]} = con(\musemrelphiTA{\gamma'}{\theta[S'_{\TA}]})$, we have that $con(\musemrelphiTA{\gamma'}{\theta[S'_{\TA}]}) \subseteq S'_{\TA}$, whence $abs(con(\musemrelphiTA{\gamma'}{\theta[S'_{\TA}]})) \subseteq abs(S'_{\TA})$, and thus
    \[
    \musemrelphiTA{\gamma'}{\theta[S'_{\TA}]}
    =           abs(con(\musemrelphiTA{S'_R}{\theta[S'_{\TA}]}))
    \subseteq   abs(S'_{\TA})
    =           S'_R
    \]
    and thus $S'_R \in \mathcal{F}_R$, $con(S'_R) \in con(\mathcal{F}_R)$, and $con(S'_R) = S'_{\TA} \in \mathcal{F}_{\TA}$.
    So fix $S_{\TA} \in \mathcal{F}_{\TA}$; we must construct saturated $S'_{\TA} \in \mathcal{F}_{\TA}$ such that $S'_{\TA} \subseteq S_{\TA}$.  Define
    \[
    S'_{\TA} = con \left( \{ q \in Q_R \mid con(\{q\}) \subseteq S_{\TA} \} \right).
    \]
    Clearly $S'_{\TA} \subseteq S_{\TA}$ and is saturated; indeed, it is the unique largest subset of $S_{\TA}$ that is saturated.  We now show that $S'_{\TA} \in \mathcal{F}_{\TA}$ by establishing that $\musemTA{\gamma'}{\theta[S'_{\TA}]} \subseteq S'_{\TA}$.
    To begin with, we note that $\theta[S'_{\TA}]$ is saturated, meaning the induction hypothesis guarantees that $\musemTA{\gamma'}{\theta[S'_{\TA}]} = con(\musemrelphiTA{\gamma'}{\theta[S'_{\TA}]})$.  We now reason by contradiction; so assume that $S''_{\TA} = \musemTA{\gamma'}{\theta[S'_{\TA}]} \setminus S'_{\TA} \neq \emptyset$.
    Since $\musemTA{\gamma'}{\theta[S'_{\TA}]}$ and $S'_{\TA}$ are both saturated, Lemma~\ref{lem:properties-of-saturation}(\ref{sublem:properties-of-saturation:closure-props}) guarantees that $S''_{\TA}$ is also saturated.
    Since $S'_{\TA} \subseteq S_{\TA}$
    monotonicity also guarantees that $\musemTA{\gamma'}{\theta[S'_{\TA}]} \subseteq \musemTA{\gamma'}{\theta[S_{\TA}]}$, so we have that $S''_{\TA} \subseteq \musemTA{\gamma'}{\theta[S_{\TA}]} \subseteq S_{\TA}$, which is a contradiction, as $S'_{\TA} \cup S''_{\TA}$ would in this case be a saturated subset of $S_{\TA}$, contradicting the fact that $S'_{\TA}$ is the maximum such set.  Therefore $\musemTA{\gamma'}{\theta[S'_{\TA}]} \subseteq S'_{\TA}$, and the proof is complete.\qedhere
\end{description}
\end{proof}

This lemma asserts that, assuming $\theta(Y)$ is saturated for all $Y$, a state in $\TA$ satisfies $\phi$ iff the corresponding state in region automaton $R_\phi(\TA)$ satisfies $\phi$.
The net effect of this theorem is that we may freely move between the standard and symbolic semantics of $\Lrelmunu$, provided the environments used to interpret free variables are appropriately saturated.

\subsection{Other Timed Modal Mu-Calculi}\label{sec:other-timed-modal-mu-calculi}

We now list several modal mu-calculi that have been presented in the literature, and whose expressiveness we will assess \emph{vis \`a vis} each other as well as $\Lrelmunu$.  In what follows we fix time safe sort $\Sigma$, clock-safe atomic proposition set $\AP$, and countably infinite set $\Var$ of propositional variables.

\paragraph{Logic $\Lnu$.}
$\Lnu$ was introduced in~\cite{laroussinie-from-timed-1995} as a logic for defining so-called \emph{characteristic formulas} of timed automata.  The formulas of $\Lnu$ may be obtained by modifying $\Lrelmunu$ as follows.

\begin{definition}[$\Lnu$ syntax]\label{def:Lnu-syntax}
The formulas of $\Lnu$ are generated by the following grammar, where $Y \in \Var$; $c \in \nats
$; ${\bowtie} \in \{<,\leq,>,\geq, =\}$; $ x, y, z \in \clocksF$; and $K \subseteq \Sigma$.
\[
\phi ::=
x \bowtie c
\mid x - y \,\bowtie\, d
\mid Y
\mid \phi \lor \phi
\mid \phi \land \phi
\mid \dia{K} \phi
\mid [K] \phi
\mid \exists \phi
\mid \forall \phi
\mid z.\phi
\mid \nu Y.\phi\footnotemark
\]
\end{definition}
\footnotetext{Technically, $\Lnu$ as presented in~\cite{laroussinie-from-timed-1995} does not include the $\nu Y.$ construct.  Rather, each formula is associated with a set of equations between propositional variables and formulas, with the intended meaning that the weakest solution of the equations is intended.  It is well-known that this equational representation can be converted into single formulas involving the $\nu Y.$ operator, however.}
Syntactically, $\Lnu$ differs from $\Lrelmunu$ in the following key respects.
\begin{enumerate}
\item
    There are no atomic propositions, and clock constraints in $\Lnu$ may use $=$ explicitly but not refer to any automaton clocks in $\clocksA$.
\item
    The binary $\exists_{\phi_1} \phi_2$ and $\forall_{\phi_1} \phi_2$ constructs are replaced by unary versions $\exists \phi$ and $\forall \phi$.
\item
    The least fixpoint operator $\mu X.$ is omitted.
\end{enumerate}
Note that every clock constraint appearing in $\Lnu$ that does not involve $=$ is an element of $\Phi_a(\clocksF)$, and thus is also a $\Lrelmunu$ formula. Those involving $=$ can be encoded in the obvious manner using a conjunction of inequalities.

We now give the semantics of $\Lnu$ in the same fashion as we did for $\Lrelmunu$.  Most of the operators are the same, and we omit these cases below.

\begin{definition}[$\Lnu$ semantics]\label{def:Lnu-semantics}
Let $\TA$ be a timed automaton over $\Sigma$ and $\AP$, with $\TS{\TA} = \ttsTA$, and let $\theta \in \Var \to 2^{\QTA}$.  Also let $\phi$ be an $\Lnu$ formula.  Then $\musemTAtheta{\phi} \subseteq 2^{\QTA}$ is defined inductively as follows.
\begin{align*}
\musemTAtheta{\exists \phi}
    &=  \{  q \in \QTA \mid
            \exists \delta \in \delays
            \colon
            \delta(q) \in \musemTAtheta{\phi}
        \}
\\
\musemTAtheta{\forall \phi}
    &=  \{  q \in \QTA \mid
            \forall \delta \in \delays
            \colon q \ttransTA{\delta} \implies
            \delta(q) \in \musemTAtheta{\phi}
        \}
\end{align*}
All other operators are interpreted as in $\Lrelmunu$ (Definition~\ref{def:Lrelmunu-semantics} and following).
\end{definition}
Semantically, the $\Lnu$ operator $\exists$ denotes a notion of ``eventually":  state $q$ satisfies $\exists \phi$ iff a state $q'$ reachable after some time-elapse $\delta$ from $q$ satisfies $\phi$.  $\Lnu$ operator $\forall$ can be seen to be the dual of $\exists$, and corresponds to a notion of ``always":  $q$ satisfies $\forall \phi$ iff every state reachable from $q$ via some time elapse satisfies $\phi$.

\paragraph{Logic $\Lmunu$.}
$\Lmunu$ is a generalization of a logic $L_\mu^t$ introduced in~\cite{sokolsky-local-model-1995}, which in turn extended $\Lnu$ with atomic propositions and least fixpoints while disallowing $=$ as a primitive clock-constraint operation.  $L_\mu^t$ also disallows so-called \emph{alternating} fixpoints.  We eliminate this restriction in order to obtain $\Lmunu$, which allows formulas of arbitrary alternation depth.  The formal definition of the formulas of $\Lmunu$ is as follows.

\begin{definition}[$\Lmunu$ syntax]\label{def:Lmunu-syntax}
The formulas of $\Lmunu$ extend those of $\Lnu$ (Definition~\ref{def:Lnu-syntax}) as follows, where $p \in \AP$.
\[
\phi ::= [ \text{operators of $\Lnu$} ]
\mid    p
\mid    \lnot p
\mid    \mu Y.\phi
\]
\end{definition}

The semantics of $\Lmunu$ extends that of $\Lnu$ in the obvious fashion by interpreting the additional operators in the same way they are in the semantics of $\Lrelmunu$.  If $\phi$ is a formula in $\Lmunu$ then we write $\musemTAtheta{\phi}$ for the set of states in $\TA$ satisfying $\phi$ in the context of $\theta$.

\paragraph{Logic $\Lc$.}
$\Lc$ modifies $\Lnu$ by making minor changes in the syntax of clock constraints, which we do not adopt here, and by replacing the unary $\exists$ and $\forall$ operators with versions of \emph{strong and weak until}~\cite{bouyer-timed-modal-2011}.  Those operators were notated $[\delta\rangle_s$ and $[\delta\rangle_w$ in that paper.  To avoid confusion with this paper's use of $\delta$ as a time-elapse metavariable, we instead use $\bouyerstrong$ and $\bouyerweak$.  The formal syntax of $\Lc$ is the following.

\begin{definition}[$\Lc$ syntax]\label{def:Lc-syntax}
The formulas of $\Lc$ modify those of $\Lnu$ (Definition~\ref{def:Lnu-syntax}) as follows, where $p \in \AP$.
\[
\phi ::= [ \text{operators of $\Lnu$ except $\exists$ and $\forall$} ]
\mid    \phi \bouyerstrong \phi
\mid    \phi \bouyerweak \phi
\]
\end{definition}

The semantics of $\Lc$ is given as follows; it differs from~\cite{bouyer-timed-modal-2011}, for reasons we explain below.
\begin{definition}[$\Lc$ semantics]\label{def:Lc-semantics}
Let $\TA$ be a timed automaton over $\Sigma$ and $\AP$, with $\TS{\TA} = \ttsTA$, and let $\theta \in \Var \to 2^{\QTA}$.  Also let $\phi$ be an $\Lc$ formula.  Then $\musemTAtheta{\phi} \subseteq 2^{\QTA}$ is defined inductively as follows.
\begin{align*}
\musemTAtheta{\phi_1 \bouyerstrong \phi_2}
    &=  \{  q \in \QTA \mid
            \exists \delta \in \delays
            \colon
            \delta(q) \in \musemTAtheta{\phi_2}\;\land
    \\
    &   \qquad\quad
        \forall \delta' \in \delays
        \colon \delta' < \delta \implies
        \delta'(q) \in \musemTAtheta{\phi_1 \lor \phi_2} \}
\\
\musemTAtheta{\phi_1 \bouyerweak \phi_2}
    &=  \musemTAtheta{\phi_1 \bouyerstrong \phi_2} \;\cup
\\
    &\qquad\quad
        \{  q \in \QTA \mid
            \forall \delta \in \delays
            \colon q \ttransTA{\delta} \implies
            \delta(q) \in \musemTAtheta{\phi_1}
        \}
\end{align*}
All other operators are interpreted as in $\Lnu$ (Definition~\ref{def:Lnu-semantics}).
\end{definition}

In~\cite{bouyer-timed-modal-2011} the operators $\bouyerstrong$ and $\bouyerweak$ are given a slightly different semantics.  In particular, $\bouyerstrong$ is interpreted there as $\bouyerstrong'$ given below.
\begin{align*}
\musemTAtheta{\phi_1 \bouyerstrong' \phi_2}
    &=  \{  q \in \QTA \mid
            \exists \delta \in \delays
            \colon
            \delta(q) \in \musemTAtheta{\phi_2} \;\land
    \\
    &   \qquad\quad
        \forall \delta' \in \delays
        \colon \delta' < \delta \implies
        \delta'(q) \in \underline{\musemTAtheta{\phi_1}} \}
\end{align*}
That is~\cite{bouyer-timed-modal-2011} defines $\bouyerstrong'$ so that $\phi_1$ must hold until the point at which $\phi_2$ becomes true, rather than allowing either $\phi_1$ or $\phi_2$ to be true, as in our definition of $\bouyerstrong$.
\begin{figure}
    \begin{subfigure}[t]{0.45\textwidth}
        \centering
        \begin{tikzpicture}[initial text=]
            \node[state,initial,align=center,minimum size=2.5cm] (l) {$\begin{array}{c}l \\ \\ x \geq 0\end{array}$};
        \end{tikzpicture}
    \caption{Timed automaton $\TA$, with location $l$ whose invariant is $x \geq 0$.}
    \label{fig:timed-automaton-TA}
    \end{subfigure}
    \hfill
    \begin{subfigure}[t]{0.45\textwidth}
        \centering
        \begin{tikzpicture}[initial text=, node distance=2.5cm]
            \node[state,initial by arrow,align=center] (lr1) {$(l,r_1)$};
            \node[state, below of=lr1] (lr2) {$(l,r_2)$};
            \draw
                (lr1) edge[left] node{$\varepsilon$} (lr2)
                (lr2) edge[loop right,right] node{$\varepsilon$} (lr2);
        \end{tikzpicture}
    \caption{Region automaton $R_{\{x\},0}(\TA)$; $r_1 = \{v \in \clockvals_{\clocks} \mid v(x) = 0\}$, and $r_2 = \{v \in \clockvals_{\clocks} \mid v(x) > 0\}$.}
    \label{fig:region-automaton-TA}
    \end{subfigure}
\caption{Timed automaton and corresponding region automaton.}
\label{fig:bouyer-semantics-issue}
\end{figure}
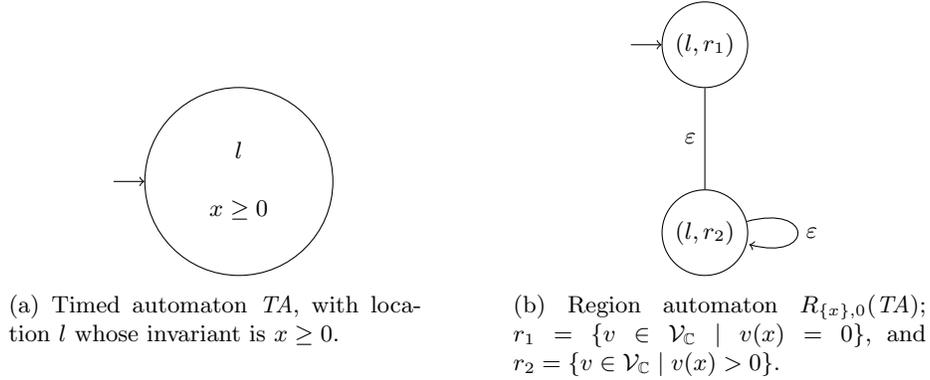
Unfortunately, that definition is inconsistent with the region-graph construction.  In particular, consider the timed automaton $\TA$ given in Figure~\ref{fig:timed-automaton-TA}; it consists of a single location $l$ and has a single clock $x$.  Note that $\bound(\TA) = 0$.  Now consider $\Lc$ formulas $\phi = (x \leq 0) \bouyerstrong (x > 0)$ and $ \phi' = (x \leq 0) \bouyerstrong' (x > 0)$.  Note that with obvious adaptations for $L_c$ to the functions $\bound$ and $cs$ we have that $\bound(\phi) = \bound(\phi') = 0$ and $cs(\phi) = cs(\phi') = \{x\}$.  Thus $R_{\TA,\phi} = R_{\TA,\phi'}$; this region automaton is given in Figure~\ref{fig:region-automaton-TA}.  From the definitions it is clear that $(l,\initval_\clocks) \in \musemTAtheta{\phi}$; however, $(l,\initval_\clocks) \not\in \musemTAtheta{\phi'}$, since no matter which time $\delta > 1$ is selected to witness the truth of $x > 1$, there are delays $\delta'$ in the range $1 < \delta' < \delta$ that cause the violation of $x \leq 1$.  It should be noted that, based on the region-graph construction, one would expect $(l,r_1) \in \musemsymTAtheta{\phi'}$, but this would violate the version of Lemma~\ref{lem:region-graph-consistency-Lrelmunu} for $\Lc$.

\paragraph{Logic $\Tmu$.}
$\Tmu$~\cite{henzinger-symbolic-model-1994} differs from the previous mu-calculi in that it does not distinguish among action labels on action transitions, and it includes a single modality that combines time and action behavior.  The formal syntax of $\Tmu$ we consider is as follows.
\begin{definition}[$\Tmu$ syntax]\label{def:Tmu-syntax}
The formulas of $\Tmu$ are generated by the following grammar, where $p \in \AP$, $cc \in \Phi_a(\clocks)$ and $Y \in \Var$.
\[
\phi ::=
p
\mid cc
\mid Y
\mid \lnot\phi
\mid \phi \lor \phi
\mid \phi_1 \rhd \phi_2
\mid z.\phi
\mid \mu Y.\phi
\]
Formulas must also satisfy the following restriction on the use of $\lnot$: in any formula $\mu Y.\phi$ every free occurrence of $Y$ in $\phi$ must be in the scope of an even number of negations in $\phi$.
\end{definition}

The semantics of $\Tmu$ may be given as follows.
\begin{definition}[$\Tmu$ semantics]\label{def:Tmu-semantics}
Let $\TA$ be a timed automaton over $\Sigma$ and $\AP$, with $\TS{\TA} = \ttsTA$, and let $\theta \in \Var \to 2^{\QTA}$.  Also let $\phi$ be an $\Tmu$ formula.  Then $\musemTAtheta{\phi} \subseteq 2^{\QTA}$ is defined inductively as follows.
\begin{align*}
&\musemTAtheta{\phi_1 \rhd \phi_2}\\
&=  \{ q \in \QTA \mid
        \exists \delta \in \delays, q' \in \delta(q), q'' \in \QTA
        \colon
        q' \ttrans{\Sigma} q''
        \land q'' \in \musemTAtheta{\phi_2} \;\land
        \\
&\qquad\quad
            \forall \delta' \in \delays
            \colon \delta' \leq \delta \implies
            \delta'(q) \in \musemTAtheta{\phi_1 \vee \phi_2}
        \}
\end{align*}
All other operators are interpreted as in $\Lrelmunu$ (Definition~\ref{def:Lrelmunu-semantics}).
\end{definition}

The $\rhd$ operator may be thought of as a ``timed next-step" operator.  Intuitively, a state satisfies $\phi_1 \rhd \phi_2$ if one of its time successors has an action transition  whose target state satisfies $\phi_2$ and every intermediate time successor (including this one) satisfies $\phi_1$ or $\phi_2$.

\section{Mu-Calculus Expressiveness Results} \label{sec:mu-calculus-expressiveness-results}

This section now establishes relative expressiveness results among the mu-calculi presented in the previous section.  We first precisely define the notions of relative expressiveness we will use.

\begin{definition}[Relative expressiveness]\label{def:relative-expressiveness}
Let $L_1$ and $L_2$ be logics such that for any formula $\phi$ in either $L_1$ or $L_2$, and timed automaton $\TA$ with environment $\theta$, $\musemTAtheta{\phi} \subseteq \QTA$ is the set of states in $\TS{\TA} = \ttsTA$ satisfying $\phi$ with respect to $\TA$ and $\theta$.
\begin{enumerate}
\item
    $L_1$ is \emph{no more expressive} than $L_2$ (equivalently, $L_2$ is \emph{at least as expressive} as $L_1$), notation $L_1 \subseteq L_2$, iff for every formula $\phi_1$ of $L_1$ there exists a formula $\phi_2$ of $L_2$ such that for every timed automaton $\TA$ and environment $\theta$, $\musemTAtheta{\phi_1} = \musemTAtheta{\phi_2}$.  We write $L_1 \not\subseteq L_2$ if it is not the case that $L_1 \subseteq L_2$.
\item
    $L_1$ and $L_2$ are \emph{equi-expressive}, notation $L_1 = L_2$, iff $L_1 \subseteq L_2$ and $L_2 \subseteq L_1$.
\item
    $L_1$ is \emph{strictly less expressive} than $L_2$ (equivalently, $L_2$ is \emph{strictly more expressive} than $L_1$), notation $L_1 \subsetneq L_2$, iff $L_1 \subseteq L_2$ and $L_2 \not\subseteq L_1$.
\item
    $L_1$ and $L_2$ are \emph{incomparable}, notation $L_1 \mathop{\#} L_2$, iff $L_1 \not\subseteq L_2$ and $L_2 \not\subseteq L_1$.
\end{enumerate}
\end{definition}

Figure~\ref{fig:expressiveness-results} summarizes the results for the five variants of the timed modal mu-calculi given above.  In the diagram, if there is a (directed) path from one logic to another then it means the first is strictly less expressive than the latter.  If two logics are not connected by a path then they are of incomparable expressive power.  Note that $\Lrelmunu$ is strictly more expressive than the other mu-calculi.  The rest of the section proves these results.

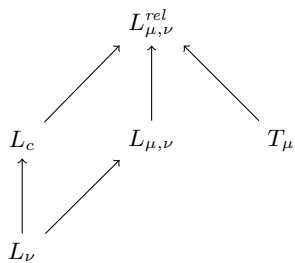
\begin{figure}
    \centering
    \begin{tikzpicture}
    \node   (Lrelmunu)  {$\Lrelmunu$};
    \node   (Lmunu)     [below=of Lrelmunu] {$\Lmunu$};
    \node   (Lc)        [left=of Lmunu]     {$\Lc$};
    \node   (Tmu)       [right=of Lmunu]    {$\Tmu$};
    \node   (Lnu)       [below=of Lc]       {$\Lnu$};
    \draw[<-] (Lrelmunu.south) -- (Lmunu.north);
    \draw[<-] (Lrelmunu.south west) -- (Lc.north east);
    \draw[<-] (Lrelmunu.south east) -- (Tmu.north west);
    \draw[<-] (Lc.south) -- (Lnu.north);
    \draw[<-] (Lmunu.south west) -- (Lnu.north east);
    \end{tikzpicture}
    \caption{Expressiveness results for timed modal mu-calculi.  If there is an edge from $L_1$ to $L_2$ then $L_1 \subsetneq L_2$.  If no path connects $L_1$ and $L_2$ then $L_1 \mathop{\#} L_2$.}
    \label{fig:expressiveness-results}
\end{figure}

\subsection{$\Lnu \subsetneq \Lc$}

A closely related result is established in~\cite{bouyer-timed-modal-2011}, albeit for a different semantics of $\bouyerstrong$ and $\bouyerweak$.  However, the arguments given in that paper can be adapted to the semantics given here.  Details are omitted.

\subsection{$\Lnu \subsetneq \Lmunu$}

That $\Lnu \subseteq \Lmunu$ is immediate, as every $\Lnu$ formula is a $\Lmunu$ formula, modulo the obvious translation of clock constraints involving $=$ that the former logic allows.

To show that $\Lnu \subsetneq \Lmunu$ we must show that $\Lmunu \not\subseteq \Lnu$.  Intuitively, this result is due to the fact that $\Lmunu$ includes a least fixpoint operator as well as a greatest fixipoint; more formally, it derives from a theorem of Bradfield~\cite{BRADFIELD1998133}, showing that the alternation-depth hierarchy for the modal mu-calculus is strict, even when formulas are interpreted over finite-state LTSs.
In particular, we define how to convert arbitrary finite-state LTSs into timed automata in a way that preserves satisfaction of so-called \emph{time-free} mu-calculus formulas, then apply Bradfield's result to arrive at the desired conclusion.  We first give the the LTS-to-TA translation.

\begin{definition}[$\TA_\mathcal{M}$]
Let $\Sigma$ be a time-safe sort and $\AP$ a set of atomic propositions, and let $\mathcal{M} = (Q, \to{}, \Lab, Q_0)$ be a LTS over $\Sigma$ and $\AP$ such that $|Q| < \infty$.  Then the \emph{pure timed automaton}, $\TAM$, has form $(Q, Q_0, \{x\}, I_\mathcal{M}, E_\mathcal{M}, \Lab)$, where:
\begin{itemize}
    \item $x \in \clocksA$;
    \item $I_{\mathcal{M}}(q) = x \leq 0$ for all $q \in Q$;
    \item $(q, a, \phi, \mathcal{C}, q') \in E_\mathcal{M}$ iff $q \ttrans{a} q'$ in $\mathcal{M}$, $\phi = \lgtrue$, and $\mathcal{C} = \emptyset$.
\end{itemize}
\end{definition}
In $\TAM$ the locations are taken to be the states of $\mathcal{M}$, of which there are only finitely many, and there is a single clock $x$ that is never allowed to advance because of the invariant ($x \leq 0$) associated with each location.  Each edge in $\TAM$ corresponds to a transition in $\mathcal{M}$, with transition guard $\lgtrue$ and reset set $\emptyset$.

From the definition of $\TS{\TAM} = \ttsTAM$ it is easy to see that:
\begin{itemize}
    \item $\QTAM = \{ (q, \nu) \in Q \times \clockvals_\clocks \mid \nu(x) = 0 \}$;
    \item $(q,\nu) \ttransTAM{a} (q',\nu')$ for $a \in \Sigma$ iff $q \ttrans{a} q'$ and $\nu = \nu'$; and
    \item $(q,\nu) \ttransTAM{\delta} (q',\nu')$ for $\delta \in \delays$ iff $q = q'$, $\nu = \nu'$, and $\delta = 0$.
\end{itemize}

We now argue that $\Lmunu \not\subseteq \Lnu$ as follows.  First, define the (untimed) modal mu-calculus to be the sublanguage of $\Lnu$ obtained by disallowing use of clock constraints, freeze quantification, $\exists$ and $\forall$.  The semantics of the (untimed) modal mu-calculus with respect to $\mathcal{M}$ and $\theta \in \Var \to 2^Q$ may be given as $\musem{\phi}{\mathcal{M}}{\theta} \subseteq Q$ in the standard fashion.  If $\theta \in \Var \to 2^Q$, then define $\theta_{\TAM} \in \Var \to 2^{\QTAM}$ by
$\theta_{\TAM}(Y) = \{(q,\nu) \in \QTAM \mid q \in \musem{\phi}{\mathcal{M}}{\theta} \land \nu(x) = 0\}$.
It can also be established that for any untimed formula $\phi$ and $\theta \in \Var \to 2^Q$,
\[
\musem{\phi}{\TAM}{\theta_{\mathcal{M}}} =
\{ (q,\nu) \in \QTAM \mid q \in \musem{\phi}{\mathcal{M}}{\theta} \land \nu(x) = 0\}.
\]
Intuitively, this observation shows that the semantics of untimed modal mu-calculus formulas, when interpreted with respect to $\TAM$, coincides in a very precise sense with the semantics of such a formula when interpreted directly with respect to the LTS $\mathcal{M}$.  The desired expressiveness result now follows from the fact that in the untimed fragment of $\Lnu$, only formulas of alternation-depth 1 may be defined, where as formulas of arbitrary alternation depth may be given in the untimed fragment of $\Lmunu$.  Bradfield's result establishes that the untimed fragment of $\Lmunu$ is strictly more expressive than the untimed fragment of $\Lnu$ over finite-state labeled transition system, and our previous observations then allow us to conclude that $\Lmunu$ is strictly more expressive than $\Lnu$.

\subsection{$\Lmunu \mathop{\#} \Lc$}

This result uses very similar arguments to those establishing that $\Lc \not\subseteq \Lnu$ and $\Lmunu \not\subseteq \Lnu$.
In particular, the $\Lc$ formula constructed in~\cite{bouyer-timed-modal-2011} that cannot be expressed in $\Lnu$ can also be shown not to be expressible in $\Lmunu$, while the untimed formulas of alternation-depth greater than 1 expressible in $\Lmunu$ cannot all be expressed in $\Lc$.

\subsection{$\Lc \subsetneq \Lrelmunu$}

That $\Lc \subseteq \Lrelmunu$ follows from the observation that the $\Lc$ time modalities, $\bouyerstrong$ and $\bouyerweak$, can be encoded in $\Lrelmunu$ using $\exists$ and $\forall$.  In particular, it is immediate that $\bouyerstrong$ coincides with the binary $\exists$ operator of $\Lrelmunu$.  As for $\bouyerweak$, we first note that $\Lrelmunu$ formula $\forall_{\lgfalse} \phi$ holds of a state in a timed automaton exactly when every time elapse possible from the state leads to a state satisfying $\phi$.  Assuming that $\Lc$ formulas $\phi_1$ and $\phi_2$ have been translated into $\phi'_1$ and $\phi'_2$ in $\Lrelmunu$, it is straightforward to verify that $\phi_1 \bouyerweak \phi_2$ can be encoded as $(\exists_{\phi_1'}\phi_2') \lor (\forall_{\lgfalse} \phi_1')$.

The argument that $\Lrelmunu \not\subseteq \Lc$ follows from the fact that $\Lrelmunu$ includes a least fixpoint operator and can thus encode untimed formulas of arbitrary alternation depth.  $\Lc$ does not have this capability.

\subsection{$\Lmunu \subsetneq \Lrelmunu$}

That $\Lmunu \subseteq \Lrelmunu$ is a direct consequence of the fact that the unary $\exists$ and $\forall$ operators in $\Lmunu$ can be encoded in $\Lrelmunu$.  In particular, if $\Lmunu$ formula $\phi$ can be encoded as $\phi'$ in $\Lrelmunu$ then $\exists \phi$ in $\Lmunu$ can be rendered as $\exists_{\lgtrue} \phi'$ in $\Lrelmunu$, and $\forall \phi$ in $\Lmunu$ as $\forall_{\lgfalse} \phi'$ in $\Lrelmunu$.

That $\Lrelmunu \not\subseteq \Lmunu$ is an immediate consequence of this fact and the facts that $\Lmunu \mathop{\#} \Lc$ and $\Lc \subsetneq \Lrelmunu$.

\subsection{$\Tmu \mathop{\#} \Lnu, \Tmu \mathop{\#} \Lc, \Tmu \mathop{\#} \Lmunu$}

In discussing the relative expressiveness of $\Tmu$ \emph{vis \`a vis} other timed modal mu-calculi one must first note an obvious difference between these logics:  $\Tmu$ does not distinguish between different action labels, whereas the other calculi do.  To ensure a fair comparison in this section and the next, we therefore limit the action modalities that can be used in the non-$\Tmu$ calculi to $\dia{\Sigma}$ and $[\Sigma]$, where $\Sigma$ is the entire time-safe sort.

Since $\Lnu \subsetneq \Lc$ and $\Lnu \subsetneq \Lmunu$ the incomparability of $\Tmu$ with respect to $\Lc$ and $\Lmunu$ all follow if $\Tmu \mathop{\#} \Lnu$.  That $\Tmu \not\subseteq \Lnu$ follows from the fact that $\Tmu$ can express least as well as greatest fixpoints; the arguments given earlier can be adapted in the obvious fashion.  We now focus on showing that $\Lnu \not\subseteq \Tmu$.  To do this we must exhibit $\Lnu$ formula $\phi$ with the property that for every $\Tmu$ formula $\Phi'$, there is $\TA$ and $\theta$ such that $\musemTAtheta{\phi} \neq \musemTAtheta{\phi}'$.

Consider $\Lnu$ formula $\exists x = 1$, where $x \in \clocksA$, and further consider two timed automata $\TA_1 = (\{l_1\}, \{l_1\}, \{x\}, I_1, \emptyset, \Lab_1)$ and $\TA_2 = (\{l_2\}, \{l_2\}, \{x\}, I_2, \emptyset, \Lab_2)$, where:
$I_1(l_1) = x < 1$, $\Lab_1(l_1) = \emptyset$, $I_2(l_2) = x \leq 1$ and $\Lab_2(l_2) = \emptyset$.
In essence, both $\TA_1$ and $\TA_2$ consist of single locations, which are also the start locations of the automata and which satisfy no atomic propositions.  Neither has any edges.  The only difference in the two is the invariant assigned to their individual locations:  $\TA_1$ assigns $x < 1$ to its location $l_1$, while $\TA_2$ assigns $x \leq 1$ to its location $l_2$.  It is straightforward to verify that $(l_1, \initval) \not\in \musem{\phi}{\TA_1}{\theta}$ while $(l_2, \initval) \in \musem{\phi}{\TA_2}{\theta}$ for any $\theta$.

We now note that for any $\Tmu$ formula $\phi'$, $(l_1,\initval) \in \musem{\phi'}{\TA_1}{\theta}$ iff $(l_2,\initval) \in \musem{\phi'}{\TA_2}{\theta}$ for any $\theta$.  The reason for this fact is that for any $\theta$ and $\Tmu$ formulas $\phi'_1$ and $\phi'_2$, $\musem{\phi'_1 \rhd \phi'_2}{\TA_1}{\theta} = \musem{\phi'_2}{\TA_2}{\theta} = \emptyset$; in other words, any formula involving $\rhd$ at the top-level is equivalent to $\lgfalse$ in the setting of $\TA_1$ and $\TA_2$.  This observation implies that any $\Tmu$ formula interpreted over $\TA_1$ or $\TA_2$ reduces to a boolean combination of clock constraints and atomic propositions, which is either satisfied by both $(l_1, \initval)$ and $(l_2,\initval)$ or by neither.  Thus, no $\Tmu$ formula is semantically equivalent to $\Lnu$ formula $\phi$, and $\Lnu \not\subseteq \Tmu$.  Consequently, $\Lnu \mathop{\#} \Tmu$.

\subsection{$\Tmu \subsetneq \Lrelmunu$}

We finish our timed modal mu-calculus results by establishing that $\Tmu \subsetneq \Lrelmunu$.  It suffices to show that $\Tmu \subseteq \Lrelmunu$, since the facts that $\Lnu \subsetneq \Lrelmunu$ and $\Lnu \mathop{\#} \Tmu$ imply that $\Lrelmunu \not\subseteq \Tmu$.  We do this by giving an encoding of the $\Tmu$ $\rhd$ operator; this completes the argument, as the other operators in $\Tmu$ are also in $\Lrelmunu$.  Assume that $\phi_1'$ and $\phi_2'$ are $\Lrelmunu$ encodings of $\Tmu$ formulas $\phi_1$ and $\phi_2$. Then
\[
\exists_{\phi_1' \lor \phi_2'} ((\phi_1' \lor \phi_2') \land \dia{\Sigma} \phi_2')
\]
is a $\Lrelmunu$ formula that is equivalent to $\phi_1 \rhd \phi_2$.  The proof of the correctness of this encoding follows from the definitions and is left to the reader.

\section{Timed Computation Tree Logic (TCTL)}\label{sec:timed-ctl}

In the remainder of this paper we consider the relative expressiveness of $\Lrelmunu, \Lmunu$ and $\Tmu$ with respect to Timed Computation Tree Logic (TCTL)~\cite{Alu1991,ACD1990,ACD1993,henzinger-symbolic-model-1994}.  (We restrict our attention to these mu-calculi because they include capabilities for least as well as greatest fixpoints, in contrast to $\Lnu$ and $\Lc$.)
 TCTL is a well-studied notation for expressing requirements on the behavior of timed automata; it extends (untimed) CTL with mechanisms for characterizing timing behavior.
The logic actually appears in different forms in the literature.  Traditionally, timing bounds on the modalities are used to limit their scope in time~\cite{ACD1990,ACD1993}. Other versions equip TCTL with  \emph{freeze quantification}~\cite{Alu1991,henzinger-symbolic-model-1994}. Based on results of Bouyer et al.~\cite{bouyer-on-the-2010} for timed linear-time temporal logics, it is the case that TCTL with freeze quantification is strictly more expressive than TCTL with time-constrained modalities. Consequently, in this paper, we focus on TCTL with freeze quantification.  This section gives the syntax and semantics of this variant of TCTL, while the section following gives expressiveness results of $\Lrelmunu$, $\Lmunu$ and $\Tmu$ \emph{vis \`a vis} the logic.

\subsection{Syntax of TCTL}

TCTL is parameterized with respect to clock structure $(\clocks,\clocksA,\clocksF)$, which will remain fixed, and a clock-safe set $\AP$ of atomic propositions.  Recall from Definition~\ref{def:Lrelmunu} that $\AP_\clocks = \mathcal{A} \cup \Phi_a(\clocks)$.

\begin{definition}[TCTL syntax]
TCTL formulas are given by the following grammar, where $A \in \AP_\clocks$ and $z \in \clocksF$.
\begin{eqnarray*}
\phi & ::= &
A \barsep
\lnot \phi \barsep
\phi \lor \phi \barsep
\tcA(\phi \tcU \phi) \barsep
\tcE(\phi \tcU \phi) \barsep
\tcfreeze{z}{\phi}
\end{eqnarray*}
\end{definition}
$\tcE$ and $\tcA$ are \emph{path quantifiers}, $\tcU$ is the until operator, and $\tcfreeze{z}{\phi}$ is the freeze-quantifier construct.
The operators \emph{release} ($\tcR$), \emph{eventually} ($\tcF$), and \emph{globally} ($\tcG$) can be derived as follows.
\begin{align*}
\tcA(\phi_1 \tcR \phi_2)
    & = \lnot \tcE(\lnot \phi_1 \tcU \lnot \phi_2)
&
\tcA(\tcF \phi)
    & = \tcA(\lgtrue \tcU \phi)
&
\tcA(\tcG \phi)
    & = \lnot \tcE(\tcF \lnot \phi)
\\
\tcE(\phi_1 \tcR \phi_2)
    & = \lnot \tcA(\lnot \phi_1 \tcU \lnot \phi_2)
&
\tcE(\tcF \phi)
    & = \tcE(\lgtrue \tcU \phi)
&
\tcE(\tcG \phi)
    & = \lnot \tcA(\tcF \lnot \phi)
\end{align*}

\subsection{Semantics of TCTL}

To give the semantics of TCTL formulas, we first review the standard notions of \emph{execution} and \emph{run} in timed transition systems.

\paragraph*{Executions and runs of timed transition systems}

Executions of TTSes are sequences of transitions, while runs are executions that are \emph{time-divergent}.
To formalize these notions, fix TTS $\T = \genTTS$ over (time-safe) $\Sigma$ and $\AP$.
If $t = (q,a,q')$ is a transition in $\T$ (i.e.\/ an element of $\ttrans{}$) then we write $\src(t) = q$, $\tgt(t) = q'$ and $\lab(t) = a$ for $t$'s source/target states and label.  We also use the following notions on sequences.

\begin{definition}[Sequences]
Let $S$ be a set.  We use $S^*$ and $S^\omega$ to represent the sets of finite and infinite sequences over $S$, respectively, and define $S^\infty = S^* \cup S^\omega$.
The empty sequence in $S^*$ is denoted $\emptyL$.
If $w, w' \in S^\infty$ then we adapt the usual notions over $S^*$ of
length, $|w|$, by taking $|w| = \infty$ iff $w \in S^\omega$,
and
concatenation, $\cdot$, by defining $w \cdot w' = w$ if $w \in S^\omega$.
If $w, w_1$ and $w_2$ are such that $w = w_1 \cdot w_2$ then we call $w_1$ a \emph{prefix} of $w$ and $w_2$ a \emph{suffix} of $w$; if in addition $w \neq w_1$ ($w \neq w_2$) we call $w_1$ a \emph{proper prefix} ($w_2$ a \emph{proper suffix}) of $w$.
We write $\indices(w) = \{i \in \nats \mid 1 \leq i \leq |w|\}$ for the set of \emph{indices} in $w$.
Note that $\indices(\emptyL) = \emptyset$ and that $\indices(w) = \{1,2,\ldots\}$ if $w \in S^\omega$.
If $w = s_1 \cdots$ and $i \in \indices(w)$ then $w[i] = s_i$ is the $i^{\text{th}}$ element in $w$,
$\prefix{w}{<i} = s_1 \ldots s_{i-1}$ ($\prefix{w}{\leq i} = s_1 \ldots s_i$) is the prefix of $w$ ending before (at) the $i^{\text{th}}$ element of $w$,
and $\suffix{w}{> i} = s_{i+1} \ldots$ ($\suffix{w}{\geq i} = s_i \ldots$) is the suffix of $w$ beginning after (at) the $i^{\text{th}}$ element of $w$.
If $|w| \in \indices(w)$ (i.e.\/ $w$ is finite and non-empty) then $w[\,|w|\,]$ is the last element of $w$.
\end{definition}

\begin{definition}[Execution, run]\label{def:execution}
Let $\pi = (q,ts)$ be an element of $Q \times (\ttrans{})^\infty$.
\begin{enumerate}
\item
$\pi$ is an \emph{execution of $\T$ from $q$} iff either $ts = \emptyL$, or $q = \src(ts[1])$ and
for all $i$ such that $i, i+1 \in \indices(ts)$ $\tgt(ts[i]) = \src(ts[i+1]))$.
\item
The \emph{duration}, $D(\pi) \in \delays \cup \{\infty\}$, of execution $\pi$ is defined as
$D(\pi) = \sum_{i=1}^{|ts|} D(ts[i])$,
where $D(ts[i]) = 0$ if $\lab(ts[i]) \in \Sigma$ and $D(ts[i]) = \lab(t[i])$ if $\lab(ts[i]) \in \nnegreals$.
\item
Execution $\pi$ is \emph{durationless} iff $D(\pi) = 0$.
\item
Execution $\pi$ is \emph{time-divergent} iff $D(\pi) = \infty$, and is \emph{time-convergent} otherwise.
\item
Execution $\pi$ is a \emph{run} of $\T$ from $q$ iff $\pi$ is time-divergent.
\end{enumerate}
We use $\executions{\T}{q}$ for the set of executions of $\T$ from $q$; $\runs{\T}{q} \subseteq \executions{\T}{q}$ for the set of runs of $\T$ from $q$; $\executions{\T}{Q}$ for $\bigcup_{q \in Q} \executions{\T}{q}$; and $\runs{\T}{Q}$ for $\bigcup_{q \in Q} \runs{\T}{q}$.
If $q \ttrans{a} q'$ and $q' \ttrans{a'} q''$, then we abuse notation and write $q \ttrans{a} q'$ for the single-transition execution $(q,(q,a,q'))$ and $q \ttrans{a} q' \ttrans{a'} q''$ for the two-transition execution $(q, (q,a,q') \cdot (q',a',q''))$.
\end{definition}

An execution is a state together with a sequence of transitions leading from that state.  The duration of an execution is the sum of the durations of its transitions, with the duration of an action transition taken to be 0.
An execution is \emph{durationless} if its duration is 0, meaning that all its transitions must either be action transitions or duration-0 time-elapse transitions.
An execution is \emph{time-divergent}, and thus is a \emph{run}, if its duration is $\infty$.
If $\pi = (q,ts)$ is an execution then source $\src(\pi) \in Q$ and target $\tgt(\pi) \in Q$ are defined by $\src(\pi) = q$, $\tgt(\pi) = q$ if $ts = \emptyL$, and $\tgt(\pi) = \tgt(ts[\,|ts|\,])$ iff $0 < |ts| < \infty$.  If $|ts| = \infty$ then $\tgt(\pi)$ is undefined.

We adapt our sequence notions to executions as follows.
\begin{definition}
Let $\pi = (q,ts)$ and $\pi' = (q',ts')$ be executions in $\executions{\T}{Q}$.
\begin{enumerate}
\item
The \emph{length}, $|\pi| \in \nats \cup \{\infty\}$, of $\pi$ is $|ts|$.  If $|\pi| < \infty$ then $\pi$ is \emph{finite}; otherwise, $\pi$ is \emph{infinite}.
\item
The \emph{transition indices}, $\indices_t(\pi) \subseteq \nats$, of $\pi$ are $\indices(ts)$.
If $i \in \indices_t(\pi)$ then $\pi[i] = ts[i]$.
\item
The \emph{action indices} of $\pi$, $\indices_\Sigma(\pi) \subseteq \indices_t(\pi)$, are $\indices_\Sigma(\pi) = \{i \in \indices_t(\pi) \mid \lab(\pi[i]) \in \Sigma\}$.
\item
\emph{Concatenation} $\pi \cdot \pi'$ is $(q, ts \cdot ts')$ if $|\pi| = \infty$ or $\tgt(\pi) = q'$, and is undefined otherwise.
\item
If $\pi = \pi_1 \cdot \pi_2$ then $\pi_1$ is a \emph{prefix} of $\pi$ and $\pi_2$ is a \emph{suffix} of $\pi$.
\item
Let $i \in \indices_t(\pi)$.  Then $\prefix{\pi}{< i}$ / $\prefix{\pi}{\leq i}$ / $\suffix{\pi}{> i}$ / $\suffix{\pi}{\geq i}$ are the prefixes / suffixes of $\pi$ defined respectively as $(q, \prefix{ts}{< i})$ / $(q, \prefix{ts}{\leq i})$ /  $(\tgt(ts[i]), \suffix{ts}{> i})$ / $(\src(ts[i]), \suffix{ts}{\geq i})$.
\end{enumerate}
\end{definition}

We now introduce the standard notions of Zeno execution and timelock state.

\begin{definition}[Zeno execution, timelock state]\label{def:run}
\mbox{}
\begin{enumerate}
\item
Execution $\pi \in \executions{\T}{Q}$ is \emph{Zeno} iff $\pi$ is time-convergent and $|\, \indices_\Sigma(\pi) \,| = \infty$.
\item
State $q \in Q$ is a \emph{timelock state} iff $\runs{\T}{q} = \emptyset$.
\end{enumerate}
\end{definition}
An execution is Zeno iff it contains an infinite number of action transitions and yet has finite duration, while a timelock state has no runs (i.e.\/ no time-divergent executions).

Later in the paper we will need a method for indexing the states in an execution.  This notion is complicated by the fact that at a given time-elapse $\delta$, the execution might traverse several states because it is in the midst of a durationless sub-execution.  We have the following.

\begin{definition}[State indexing in executions]\label{def:state-indices}
Let $\pi \in \executions{\T}{Q}$ and $\delta \in \delays$.
\begin{enumerate}
\item
    The \emph{$\delta$-prefix}, $\prefix{\pi}{\leq \delta}$, of $\pi$ is the maximum-length prefix $\pi'$ of $\pi$ such that $D(\pi') \leq \delta$.
\item
    Execution $\pi$ is \emph{defined at $\delta$} iff there exists finite prefix $\pi'$ of $\pi$ such that $\delta \leq D(\pi')$.
\item
    The \emph{durationless execution, $de(\pi,\delta)$, of $\pi$ at $\delta$} is defined iff $\pi$ is defined at $\delta$, in which case it is given as follows.
    \begin{enumerate}
    \item
        If $D(\prefix{\pi}{\leq\delta}) < \delta$ then $de(\pi,\delta) = (q',\emptyL)$, where $\delta' = \delta - D(\prefix{\pi}{\leq \delta})$ and $\tgt(\prefix{\pi}{\leq \delta})) \ttransTA{\delta'} q'$.
    \item
        If $D(\prefix{\pi}{\leq \delta}) = \delta$ then $de(\pi,\delta)$ is the maximum-length durationless suffix of $\prefix{\pi}{\leq \delta}$.
    \end{enumerate}
\item
    The \emph{state indices of $\pi$} are
    $
    \indices_s(\pi) = \{(\delta, i) \in \delays \times \nats \mid
    \text{$\pi$ is defined at $\delta$} \land
    i \leq |de(\pi,\delta)|\}.
    $
    If $(\delta_1, i_1), (\delta_2,i_2) \in \indices_s(\pi)$ then $(\delta_1, i_2) <_\pi (\delta_2,i_2)$ iff $\delta_1 < \delta_2$, or $\delta_1 = \delta_2$ and $i_1 < i_2$.
\item
    Let $(\delta,i) \in \indices_s(\pi)$ and $\pi' = de(\pi,\delta)$.  The \emph{state at index $(\delta,i)$ in $\pi$, $\pi[\delta,i] \in Q$}, is given by $\pi[\delta,0] = src(\pi')$, and $\pi[\delta, i] = \tgt(\pi'[i])$ for $i > 0$.
\end{enumerate}
\end{definition}
The $\delta$-prefix $\prefix{\pi}{\leq \delta}$ is the longest prefix of $\pi$ whose duration is bounded above by $\delta$.
If $|\pi| < \infty$ or $\pi$ is a run then it is easy to see that $|\,\prefix{\pi}{\leq \delta}\,| < \infty$ for all $\delta \in \delays$.  If, on the other hand, $|\pi| = \infty$ and $D(\pi) \leq \delta$ then $\prefix{\pi}{\leq \delta} = \pi$ and $|\,\prefix{\pi}{\leq \delta}\,| = \infty$.
Execution $\pi$ is defined a $\delta$ if there is a point in $\pi$ where time elapses to $\pi$.  Any finite-length $\pi$ is defined at $\delta$ iff $\delta \leq D(\pi)$, while any run is defined at every $\delta \in \delays$.
Execution $de(\pi,\delta)$
is the maximum-length durationless execution occurring at time $\delta$ in $\pi$; it is only defined if $\pi$ is defined at $\delta$.  There are two cases in the definition.  In the first, $D(\prefix{\pi}{\leq \delta}) < \delta$, meaning $\pi$ is in the middle of a time elapse at time $\delta$.  In this case the durationless execution is the zero-length one obtained by elapsing time from target state of $\prefix{\pi}{\leq \delta}$ until the over-all $\delta$ time limit is reached.  In the second case, $D(\prefix{\pi}{\leq \delta}) = \delta$, meaning $\pi$ is not in the middle of a time elapse at time $\delta$.  In this case, since $\prefix{\pi}{\leq \delta}$ is maximal the next transition in $\pi$ after $\prefix{\pi}{\leq \delta}$, if one exists, must be a positive-duration time elapse; consequently, the longest durationless suffix of $\prefix{\pi}{\leq \delta}$ is also the longest durationless execution embedded in $\pi$ at time $\delta$.  The state indices of $\pi$ consist of a duration $\delta$ and an index into the durationless execution at $\delta$, and $\pi[\delta,i]$ is the $i^{\text{th}}$ state in $de(\pi,\delta)$, where the $0^{\text{th}}$ state is the source state of $de(\pi,\delta)$ and subsequent states are the targets of the transitions in $de(\pi,\delta)$.

Later in the paper we also will use the notion \emph{extents} of executions associated with state indices of the execution.  Extents may be seen as analogous to prefixes and suffixes, with the given state index being the target of the extent to it, and the state at the index being the source of the extent from it.  These notions are formalized as follows.

\begin{definition}[Extents of executions]\label{def:extent}
Let $\pi \in \executions{\T}{Q}$.
\begin{enumerate}
\item
    The \emph{start time}, $st_\pi(i) \in \delays$, of transition $i \in \indices_s(\pi)$ is defined as $st_\pi(i) = D(\prefix{\pi}{< i})$, and the \emph{end time}, $et_\pi(i) \in \delays$ of transition $i$ is $st_\pi(i) + D(\pi[i])$.
\item
    The \emph{source- and target-state indices}, $si_\pi(i), ti_\pi(i) \in \indices_s(\pi)$, of transition $i \in \indices_t(\pi)$ are:
    \begin{align*}
    si_\pi(i)
        &=
        \begin{cases}
            (st_\pi(i), j)
                & \text{if $D(\pi[i]) > 0$ and $j = \max\{j' \in \nats \mid (st_\pi(i),j') \in \indices_s\}$}\\
            (\delta, j - 1)
                & \text{if $D(\pi[i]) = 0$ and $ti_{\pi}(i) = (\delta,j)$}
        \end{cases}
        \\
    ti_\pi(i)
        &=
        \begin{cases}
            (\delta + \lab(\pi[i]), 0)
                & \text{if $D(\pi[i]) > 0$ and $si_\pi(i) = (\delta,j)$} \\
            (st_\pi(i), i - |\pi'|)
                & \text{if $D(\pi[i]) = 0$ and $\pi = \pi' \cdot de(\pi,\delta)$.}
        \end{cases}
    \end{align*}
    In the definition of $ti_\pi(i)$, $i - |\pi'| $ is the index of transition $i$ in $\pi$ relative to $de(\pi,\delta)$.
\item
    Let $(\delta, i) \in \indices_t(\pi)$.
    Then the \emph{extent of $\pi$ \underline{to} $(\delta,i)$},
    $\prefix{\pi}{\leq (\delta,i)} \in \executions{\T}{q}$, is defined as follows.
    \begin{align*}
    \prefix{\pi}{\leq (\delta,i)} =
    \begin{cases}
    (\src(\pi), \emptyL)
        & \text{if $(\delta,i) = (0,0)$} \\
    \prefix{\pi}{\leq j}
        & \text{if $j \in \indices_t(\pi)$ and $ti_{\pi}(j) = (\delta,i)$} \\
    (\prefix{\pi}{< j}) \cdot
    cut(\pi[j], \delta - st_\pi(j))
            & \text{if $j \in \indices_t(\pi)$ and $st_{\pi}(j) < \delta < et_{\pi}(j)$}
    \end{cases}
    \end{align*}
    The \emph{extent of $\pi$ \underline{from} $(\delta,i)$}, $\suffix{\pi}{\geq (\delta,i)} \in \executions{\T}{Q}$, is defined as follows, where $q' = \pi[\delta,i]$.
    \begin{align*}
    \suffix{\pi}{\geq (\delta,i)} &=
        \begin{cases}
        \pi & \text{if $(\delta,i) = (0,0)$} \\
        \suffix{\pi}{\geq j}
            & \text{if $j \in \indices_t(\pi)$ and $si_\pi(j) = (\delta,i)$}
            \\
        rem(\pi[j], et_\pi(i) - \delta) \cdot \suffix{\pi}{\geq j}
            & \text{if $j \in \indices_t(\pi)$ and $st_{\pi}(j) < \delta < et_{\pi}(j)$}
        \end{cases}
    \end{align*}
\end{enumerate}
\end{definition}
The start and end times of transition $i$ in $\pi$ are the times in $\pi$ at which transition $\pi[i]$ begins and finishes; note that $de(\pi,\delta)$ is the maximum sub-execution of $\pi$ whose transitions are durationless and whose start and end times are all $\delta$.  The source-state index of transition $i$ in $\pi$ is the state index within $\pi$ associated with the source state of transition $i$ in $\pi$, while the target-state index is associated with the target state of transition $i$.  Finally, the extent of $\pi$ to $(\delta,i) \in \indices_s$ is an execution that includes exactly the behavior of $\pi$ up to and including state index $(\delta,i)$.
The notion is analogous to prefixing, but its definition is more complex because this extent may include time-elapse behavior that may be part of, but not wholly include, a delay transition in $\pi$.  Similarly, the extent of $\pi$ from $(\delta,i)$ is an execution capturing the behavior of $\pi$ beginning from $(\delta,i)$.
The following are easy to establish, where $(\delta_1, i_2) - (\delta_2, i_2) \in \reals \times \ints$ is defined as $(0, i_1 - i_2)$ if $\delta_1 = \delta_2$ and as $(\delta_1 - \delta_2, i_1)$ otherwise.
\begin{align*}
\indices_s(\prefix{\pi}{\leq (\delta,i)})
    &= \{ \delta',j) \in \indices \in \indices_s(\pi) \mid
    (\delta',j) \leq_{\pi} (\delta,i)
    \}
    \\
\indices_s(\suffix{\pi}{\geq (\delta,i)})
    &= \{ (\delta',j) - (\delta,i) \mid (\delta',j) \in \indices_s(\pi) \land (\delta,i) \leq_\pi (\delta',j)
    \}
\end{align*}
Note that the state indices of $\suffix{\pi}{\geq (\delta,i)}$ need to be ``adjusted downward" by $(\delta,i)$.  We also know that for all $(\delta', j) \in \indices_s(\prefix{\pi}{(\delta,i)})$, $(\prefix{\pi}{\leq (\delta,i)})[\delta',j] = \pi[\delta',j]$; that is, the extent of $\pi$ to $(\delta,i)$ contains the same states at the same indices as $\pi$.  An analogous result holds for $\suffix{\pi}{\geq (\delta,i)}$.  Define $(\delta_1, j_1) + (\delta_2, j_2)$ as $(\delta_2, j_1 + j_2)$ if $\delta_1 = 0$, and as $(\delta_1 + \delta_2, j_2)$ otherwise.  Then for all $(\delta',j) \in \indices_s(\suffix{\pi}{(\delta,i)}$, $(\suffix{\pi}{\geq (\delta,i)})[\delta',j] = \pi[(\delta',j) + (\delta,i)]$.  This says that the extent of $\pi$ from $(\delta,i)$ contains the same states and $\pi$ after accounting for the offsets in the state indices of the extent from $\pi$ \emph{vis \`a vis} those of $\pi$.

\paragraph*{TCTL formula semantics.}

Our semantics of TCTL is in the style of~\cite{henzinger-symbolic-model-1994}
and relies on \emph{until paths}.

\begin{definition}[Until path]
Let $\T = (Q,\ttrans{},\Lab,Q_0)$ be a TTS over time-safe $\Sigma$ and $\AP$ and let $Q_1, Q_2 \subseteq Q$.
Then $\pi \in \executions{\T}{Q}$ is an \emph{until path} from $Q_1$ to $Q_2$ in $\T$ iff there exists $(\delta,i) \in \indices_s(\phi)$ such that:
\[
    \pi(\delta,i) \in Q_2 \land
    \left(
        \forall\, (\delta', i') \in \indices_s(\pi) \colon
        (\delta',i') <_\pi (\delta,i) \implies \pi[\delta',i'] \in Q_1
    \right).
\]
We write $\untils{\T}{Q_1}{Q_2}$ for the set of until paths from $Q_1$ to $Q_2$ in $\T$.  If $\pi \in \untils{\T}{Q_1}{Q_2}$ is such that no proper prefix of $\pi$ is in $\untils{\T}{Q_1}{Q_2}$ then $\pi$ is called a \emph{minimal until path}.
\end{definition}
Intuitively, $\pi$ is an until path from $Q_1$ to $Q_2$ in $\T$ if $\pi$ eventually hits a state in $Q_2$, with all strictly preceding states in $\pi$ being in $Q_1$.
Note that if $\pi$ is a minimal until path then it must be of finite length, and that if $\pi$ is an until path then there is a unique prefix of $\pi$ that is also a minimal until path.
We use $mup_{\T}(\pi,Q_1,Q_2)$ for this unique prefix of $\pi$.

In what follows, if $\TA$ is a timed automaton with associated TTS $\TS{\TA}$ then we write $\executions{\TA}{q}, \runs{\TA}{q}$ and $\untils{\TA}{Q_1}{Q_2}$ instead of $\executions{\TS{\TA}}{q}, \runs{\TS{\TA}}{q}$ and $\untils{\TS{\TA}}{Q_1}{Q_2}$.
The semantics of TCTL can now be specified as follows.

\begin{definition}[Semantics of TCTL formulas]
Fix time-safe $\Sigma$ and clock-safe $\AP$, and let
$\TA = \genTA$ be a timed automaton over $\Sigma$ and $\AP$, with $\TS{\TA} = \ttsTA$. Then the semantics, $\tcsemTA{\phi} \subseteq \QTA$, of a TCTL formula $\phi$ over $\TA$ is defined inductively as follows, where $A \in \AP_\clocks$.
\begin{align*}
\tcsemTA{A}
& = \{ q \in \QTA \mid A \in \Lab_{\TA}(q) \}
\\
\tcsemTA{\lnot \phi}
& = \QTA \setminus \tcsemTA{\phi}
\\
\tcsemTA{\phi_1 \lor \phi_2}
& = \tcsemTA{\phi_1} \cup \tcsemTA{\phi_2}
\\
\tcsemTA{\tcE(\phi_1 \tcU \phi_2)}
& = \{ q \in \QTA \mid \exists r \in \runs{\TA}{q} \colon
r \in \untils{\TA}{\tcsemTA{\phi_1 \lor \phi_2}}{\tcsemTA{\phi_2}} \}
\\
\tcsemTA{\tcA(\phi_1 \tcU \phi_2)}
& = \{ q \in \QTA \mid \forall r \in \runs{\TA}{q} \colon
r \in \untils{\TA}{\tcsemTA{\phi_1 \lor \phi_2}}{\tcsemTA{\phi_2}} \}
\\
\tcsemTA{\tcfreeze{z}{\phi}}
& = \{ q \in \QTA \mid q[z := 0] \in \tcsemTA{\phi} \}
\end{align*}
If $q \in \QTA$ then we say $q$ \emph{satisfies} $\TCTL$ formula $\phi$ iff $q \in \tcsemTA{\phi}$.
\end{definition}
The semantics of most of the operators is straightforward.  Note that the definitions for $\tcE(\phi_1 \tcU \phi_2)$ and $\tcA(\phi_1 \tcU \phi_2)$ are given in terms of runs that are also until paths.  Specifically, for a state $q$ to be in $\tcsemTA{\tcE(\phi_1 \tcU \phi_2)}$ it must have a run with a state that satisfies $\phi_2$, with all preceding states in the run keeping either $\phi_1$ or $\phi_2$ true.  For $q$ to be in $\tcsemTA{\tcA(\phi_1 \tcU \phi_2)}$, all of its runs must have this property.

We close this section with a technical lemma about minimal until paths that will be used in the next section.

\begin{lemma}[Minimal until paths]\label{lem:mup}
Let $\T = \genTTS$ be a TTS over time-safe $\Sigma$ and $\AP$, let $Q_1, Q_2 \subseteq Q$, and let $\pi \in \untils{\T}{Q_1}{Q_2}$ be a minimal until path such that $|\pi| > 0$.  Then for all $(\delta,i) \in \indices_s(\pi_{<|\pi|})$, $\pi[\delta,i] \in Q_1 \setminus Q_2$.
\end{lemma}

This lemma asserts that if a minimal until-path $\pi \in \untils{\T}{Q_1}{Q_2}$ has at least one transition, then every state in $\pi_{<|\pi|}$, which is $\pi$ with its final transition removed, must be in $Q_1$ but not $Q_2$.  The proof of this lemma relies on the fact that $\pi$ is minimal, and thus no state in any proper prefix of $\pi$ can be in $Q_2$.

\section{TCTL Expressiveness Results}\label{sec:tctl-expressiveness}

This section compares the expressiveness of TCTL with the three mu-calculi in Section~\ref{sec:timed-modal-mu-calculi} --- $\Lrelmunu, \Lmunu$ and $\Tmu$ --- that can express least and greatest fixpoints.  Figure~\ref{fig:tctl-expressiveness-results} summarizes the results in this section, key among which is that $\Lrelmunu$ is strictly more expressive than TCTL while the other mu-calculi are incomparable.

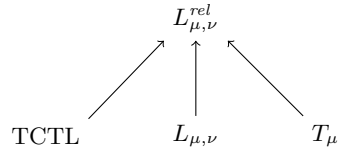
\begin{figure}
    \centering
    \begin{tikzpicture}
    \node   (Lrelmunu)  {$\Lrelmunu$};
    \node   (Lmunu)     [below=of Lrelmunu] {$\Lmunu$};
    \node   (TCTL)        [left=of Lmunu]     {TCTL};
    \node   (Tmu)       [right=of Lmunu]    {$\Tmu$};
    \draw[<-] (Lrelmunu.south) -- (Lmunu.north);
    \draw[<-] (Lrelmunu.south west) -- (TCTL.north east);
    \draw[<-] (Lrelmunu.south east) -- (Tmu.north west);
    \end{tikzpicture}
    \caption{Expressiveness results for timed modal mu-calculi \emph{vis \`a vis} TCTL.  If there is an edge from $L_1$ to $L_2$ then $L_1 \subsetneq L_2$.  If no path connects $L_1$ and $L_2$ then $L_1 \mathop{\#} L_2$.}
    \label{fig:tctl-expressiveness-results}
\end{figure}

In this section we refer to classes of \emph{non-Zeno} and \emph{timelock-free} timed automata.  Timed automaton $\TA$ is non-Zeno iff no state in $\TS{\TA}$ has any Zeno runs from it, and it is timelock-free if no state in  $\TS{\TA}$ is a timelock state.  It should be noted that these definitions differ from traditional definitions, which only require non-Zeno-ness and timelock freedom for states reachable from a start state.  However, in the presence of freeze quantification, the notion of reachable state becomes somewhat complex to formulate; we thus opt for these stronger, but easier to express, definitions.

\subsection{Fixpoints over Subset Lattices}

In this section we frequently need to establish set-theoretic relationships between semantics of formulas and fixpoints of monotonic functions defined over the complete lattice of the powerset of states in a TTS.  Accordingly, we briefly review characterizations of these fixpoints and highlight proof methods we use in what follows.

Let $X$ be a set.  Then we call $(2^X, \subseteq, \bigcup, \bigcap)$ the \emph{subset lattice} over $X$, with $\subseteq$ being the partial ordering over the carrier set $2^X$ and $\bigcup$ and $\bigcap$ being the join and meet operations, respectively.  A function $f \in 2^X \to 2^X$ is said to be \emph{monotonic} over this lattice provided that, whenever $X_1 \subseteq X_2 \subseteq S$, $f(X_1) \subseteq f(X_2)$  Such functions are guaranteed to have least and greatest fixpoints, $\mu f \subseteq X$ and $\nu f \subseteq X$ respectively, that Tarski and Knaster characterized as follows.
\begin{align*}
    \mu f &= \bigcap \{X' \subseteq X \mid f(X') \subseteq X'\} \\
    \nu f &= \bigcup \{X' \subseteq X \mid XX' \subseteq f(X')\}
\end{align*}
If $f(X') \subseteq X'$ then $X'$ is called a \emph{pre-fixpoint} of $f$; if instead $X' \subseteq f(X')$ then $X'$ is a \emph{post-fixpoint} of $f$.

A recently published paper~\cite{CK2023} gives an alternative characterization of $\mu f$ for monotonic $f$ in terms of \emph{well-found support structures}.  A support structure for $f$ is a pair $(X', {\prec})$ where $X' \subseteq X$, ${\prec} \subseteq X' \times X'$, and for all $x' \in X'$, $x' \in f(\preimg{{\prec}}{x})$, where $\preimg{{\prec}}{x} = \{ x'' \in X' \mid x'' \prec x' \}$ are the elements ``below'' $x'$ as defined by $\prec$.  Support structure $(X', {\prec})$ is well-founded iff relation $\prec$ is well-founded, and $X'$ is called well-supported for $f$ iff there is ${\prec} \subseteq X' \times X'$ such that $(S,{\prec})$ is a well-founded support structure for $f$.  Then
\[
\mu f
= \bigcup \{ X' \subseteq X \mid
\text{$X'$ is a well-supported for $f$} \}
\]

These characterizations imply the following proof strategies for relating $\mu f$ and $\nu f$ to a given subset $Y$ of $X$.
\begin{description}
\item[To prove $\mu f \subseteq Y$:]  Prove that $f(Y) \subseteq Y$.
\item[To prove $Y \subseteq \mu f$:]  Give well-founded ${\prec} \subseteq Y \times Y$ and prove that $(Y,{\prec})$ is a support structure for $f$.
\item[To prove $\nu f \subseteq Y$:]  Prove that for every $X'$ such that $X' \subseteq f(X')$, $X' \subseteq Y$.
\item[To prove $Y \subseteq \nu f$:]  Prove that $Y \subseteq f(Y)$.
\end{description}

\subsection{$\TCTL \mathop{\#} \Lmunu$, $\TCTL \mathop{\#} \Tmu$}

These results are already in the literature; we include them here for completeness.

Fontana's dissertation~\cite{fontana-2014} remarks that $\TCTL \mathop{\#} \Lmunu$.  In particular, he shows that $\Lmunu \not\subseteq TCTL$ by establishing $\Lmunu$ formula
\[
z.
\left(
    \nu X.
    \left(
        z \leq  0 \land \forall (z \leq 0 \land [\Sigma] X)
    \right)
\right)
\]
which asserts that time never advances, is impossible to express in TCTL.  That $TCTL \not\subseteq \Lmunu$ is a consequence of Kamp's seminal result on the expressive completeness of linear temporal logic that includes $\tcU$~ \cite{kamp68} and  the incompleteness of linear temporal logics that include only $\tcF$ and $\tcG$, which correspond to the $\exists$ and $\forall$ operators of $\Lmunu$.

It is also the case that $\Tmu \mathop{\#} \TCTL$.
Henzinger et al. showed that $\Tmu \not \subseteq \TCTL$~\cite{henzinger-symbolic-model-1994}. In particular, the $\Tmu$ formula $\nu X . \mu Y . ((x = 0 \land ((x = 0 \triangleright ((x > 0) \triangleright X)))) \lor (\true \triangleright Y))$ cannot be expressed in $\TCTL$. This is even the case for timelock-free, non-Zeno timed automata.
This formula
is an adaptation of Emerson's result that the mu-calculus cannot be expressed in CTL~\cite{emerson-1991,wolper-1983}. Likewise, in~\cite{henzinger-symbolic-model-1994} it is shown that over the model of real-time systems, $\TCTL \not \subseteq \Tmu$. This result extends to timed automata. However, $\TCTL \subsetneq \Tmu$
for the class of real-time programs~\cite{henzinger-symbolic-model-1994}, and also for timelock-free, non-Zeno timed automata~\cite{penczek-advances-in-2006}.


\subsection{$\TCTL \subsetneq \Lrelmunu$}\label{sec:tctl-vs-Lrelmunu}

We now establish a key result of this paper:  that over the class of all timed automata, $\TCTL \subsetneq \Lrelmunu$.  We begin by noting that since
$\Tmu \not \subseteq \TCTL$ and $\Tmu \subsetneq \Lrelmunu$, it must be the case $\Lrelmunu \not\subseteq  \TCTL$.
Therefore, to establish that $\Lrelmunu \subsetneq \TCTL$ it suffices to show that $\Lrelmunu \subseteq \TCTL$.  We do this by giving an embedding of TCTL into $\Lrelmunu$ in the remainder of this section.

The difficulty in the translation resides in the fact that $\Lrelmunu$ includes no direct mechanism for referencing runs of timed automata.  Thus any encoding of TCTL in this logic must account not only for the truth or falsity of a formula in individual states in the timed automaton, but also be sensitive to the possibility or impossibility of runs from those states.  To accommodate these subtleties, we present our translation in a staged fashion.  We first consider only the set of timelock-free (TF) and non-Zeno (nZ) timed automata and give translations of TCTL into $\Lrelmunu$.  We then explain how to adapt these translaton so that these restrictions may be removed.

Before defining the translations we first introduce notation for $\Lrelmunu$ that will be used in them.  First, we will abbreviate $\exists_{\lgtrue} \gamma$ and $\forall_{\lgfalse} \gamma$ as $\exists\, \gamma$ and $\forall\, \gamma$, respectively.  We also define the following.

\begin{definition}[Time-stopping $\TStop$, $\tcA(-\tcU-)$ template $\au$]\label{def:tstop-au}
\mbox{}
\begin{enumerate}
\item
    $\Lrelmunu$ formula $\TStop$ (``time stopping'') is defined as $\TStop = z.(\forall\, z < 1)$.
\item
    Let $\gamma_1, \gamma_2 \in \Phirelmunu$ be such that $z \not\in cs(\gamma_1) \cup cs(\gamma_2)$.  Then $\au(\gamma_1, \gamma_2) \in \Phirelmunu$ is defined as
    \[
    \au(\gamma_1, \gamma_2)
    =
    \exists_{\gamma_1} (\gamma_2 \lor (\TStop \land \forall\, \gamma_1)).
    \]
\end{enumerate}
\end{definition}

Intuitively,
$\TStop$ is a formula satisfied by states whose longest time-elapse transition is less than 1.  It is straightforward to see that a state satisfies $\exists\, \TStop$ iff the state is \emph{time-bounded}:  all of the state's time-elapse transitions have duration bounded above by $\delta$ for some $\delta \in \delays$.  These facts are formalized as follows for timed automaton $\TA$ with $\TS{\TA} = (\QTA, \ldots)$.
\begin{align*}
\musemTAnotheta{\TStop}
    &= \{q \in \QTA \mid \forall\,\delta \in \delays \colon q \ttransTA{\delta} \implies \delta < 1 \}
    \\
\musemTAnotheta{\exists\, \TStop}
    &= \{q \in \QTA \mid \exists \delta' \in \delays \colon \forall\,\delta \in \delays \colon q \ttransTA{\delta} \implies \delta \leq \delta' \}
\end{align*}
The operator $\au$ can be seen as a bounded weak-until operator over time elapses.  Specifically, a state in a timed automaton satisfies $\au(\gamma_1,\gamma_2)$ iff either there is a time-elapse from the state leading to a state satisfying $\gamma_2$, with every intervening state satisfying $\gamma_1 \lor \gamma_2$, or there is bound $\delta$ on the length of time elapses from the state, and every time elapse from the state leads to a state satisfying $\gamma_1$.  The following lemma formalizes this intuition.

\begin{lemma}[Semantics of $\au$]\label{lem:semantics-of-au}
Let $\TA$ be a timed automaton, $\TS{\TA} = (\QTA, \ldots)$, and let $\theta \in \Var \to 2^{\QTA}$.  Then for any $\gamma_1, \gamma_2 \in \Phirelmunu$,
\[
\musemTAtheta{\au(\gamma_1,\gamma_2)}
=
\musemTAtheta{\exists_{\gamma_1} \gamma_2}
\cup \musemTAtheta{(\exists\, \TStop) \land (\forall\, \gamma_1) }.
\]
\end{lemma}
\begin{proof}
Fix $\TA$, $\TS{\TA} = \ttsTA$, $\theta \in \Var \to 2^{\QTA}$, and $\gamma_1, \gamma_2 \in \Phirelmunu$.
We first note that,
based on the semantics of $\Lrelmunu$, it is the case that for any $\gamma'_1, \gamma'_2, \gamma'_3 \in \Phirelmunu$,
\[
\musemTAtheta{\exists_{\gamma'_1} (\gamma'_2 \lor \gamma'_3)}
= \musemTAtheta{(\exists_{\gamma'_1} \gamma'_2) \lor (\exists_{\gamma'_1} \gamma'_3)}.
\]
Consequently,
\begin{align*}
\musemTAtheta{\au(\gamma_1,\gamma_2)}
    &= \musemTAtheta{\exists_{\gamma_1} (\gamma_2 \lor (\TStop \land \forall\, \gamma_1))}
    \\
    &= \musemTAtheta{(\exists_{\gamma_1} \gamma_2) \lor \exists_{\gamma_1} (\TStop \land \forall\, \gamma_1)}
    \\
    &=  \musemTAtheta{\exists_{\gamma_1} \gamma_2} \cup
        \musemTAtheta{\exists_{\gamma_1} (\TStop \land \forall\, \gamma_1)}.
\end{align*}
To finish the proof, we show that
$\musemTAtheta{\exists_{\gamma_1} (\TStop \land \forall\, \gamma_1)}
= \musemTAtheta{(\exists\, \TStop) \land (\forall\, \gamma_1)}$.
So assume that $q \in \musemTAtheta{\exists_{\gamma_1} (\TStop \land \forall\, \gamma_1)}$.
This means there exists $\delta_q \in \delays$ such that
$\delta_q(q) \in \musemTAtheta{\TStop \land \forall\, \gamma_1}$
and for all $\delta' < \delta_q, \delta'(q) \in \musemTAtheta{\gamma_1 \lor (\TStop \land \forall\, \gamma_1)}$.
As $\musemTAtheta{\gamma_1 \lor (\TStop \land \forall\, \gamma_1)} = \musemTAtheta{\gamma_1}$,
we have that for all $\delta' < \delta_q, \delta'(q) \in \musemTAtheta{\gamma_1}$.
Since $\delta_q(q) \in \musemTAtheta{\TStop \land \forall\, \gamma_1}$, there exists $\delta_1$ such that $0 \leq \delta_1 < 1$ and such that for all $\delta_1' \in \delays$ and $q' \in \QTA$ if $\delta_q(q) \ttransTA{\delta_1'} q'$ then $\delta_1' < \delta_1$ and $q \in \musemTAtheta{\gamma_1}$.
Thus $q \in \musemTAtheta{\exists\, \TStop}$, and for all $\delta \in \delays$ and $q' \in \QTA$, if $q \ttransTA{\delta} q'$ then $q' \in \musemTAtheta{\gamma_1}$.
Consequently, $q \in \musemTAtheta{\forall\, \gamma}$, and we have $q \in \musemTAtheta{(\exists\, \TStop) \land (\forall\, \gamma_1)}$ and $\musemTAtheta{\exists_{\gamma_1} (\TStop \land \forall\, \gamma_1)} \subseteq \musemTAtheta{(\exists\, \TStop) \land (\forall\, \gamma_1)}$.

Now assume that $q \in \musemTAtheta{(\exists\, \TStop) \land (\forall\, \gamma_1)}$.
Since $q \in \musemTAtheta{\exists\, \TStop}$
it follows that there is a minimum $\delta_m \in \delays$
such that for all $\delta \in \delays$, if $q \ttransTA{\delta}$ then $\delta \leq \delta_m$.
Let $\delta_q$ be such that $\delta_m - 1 < \delta_q < \delta_m$.  It can be seen that $\delta_q(q) \in \musemTAtheta{\TStop \land \forall\, \gamma_1}$ and for all $\delta' < \delta_q, \delta'(q) \in \musemTAtheta{\gamma_1} = \musemTAtheta{\gamma_1 \lor (\TStop \land \forall\, \gamma_1)}$.
Consequently, $q \in \musemTAtheta{\exists_{\gamma_1} (\TStop \land \forall\, \gamma_1)}$ and $\musemTAtheta{(\exists\, \TStop) \land (\forall\, \gamma_1)} \subseteq \musemTAtheta{\exists_{\gamma_1} (\TStop \land \forall\, \gamma_1)}$, whence $\musemTAtheta{\exists_{\gamma_1} (\TStop \land \forall\, \gamma_1)}
= \musemTAtheta{(\exists\, \TStop) \land (\forall\, \gamma_1)}$.  This completes the proof.\qedhere
\end{proof}

\paragraph{Translation for TF/nZ timed automata.}

In a TF / nZ timed automaton, every state is the source of at least one run, and every execution with infinitely many action transitions is time-divergent.  We define the translation of TCTL to $\Lrelmunu$ for TF / nZ timed automata as a mapping $\embed_{tn}$.

\begin{definition}[TF / nZ translation]
Let $\phi$ be a TCTL formula.  Then $\Lrelmunu$ formula $\embed_{tn}(\phi)$ is defined inductively as follows, where $A \in \AP_\clocks$.
\begin{align*}
\embed_{tn}(A)
    &= A
    \\
\embed_{tn}(\lnot \phi)
    &= \lnot \,\embed_{tn}(\phi)
    \\
\embed_{tn}(\phi_1 \lor \phi_2)
    &= \embed_{tn}(\phi_1) \lor \embed_{tn}(\phi_2)
    \\
\embed_{tn}(\tcE(\phi_1 \tcU \phi_2))
    &=  \mu X . \exists_{\embed_{tn}(\phi_1)}
        (\embed_{tn}(\phi_2) \lor
            (\embed_{tn}(\phi_1) \land \dia{\Sigma} X))
    \\
\embed_{tn}(\tcA (\phi_1 \tcU \phi_2))
    &=  \mu X . \au(\embed_{tn}(\phi_1) \land [\Sigma]X, \embed_{tn}(\phi_2))
    \\
\embed_{tn}(z.\phi)
    &= z.\embed_{tn}(\phi)
\end{align*}
\end{definition}
%
It can be seen that for any state in a TF / nZ timed automaton satisfying $\TStop$ and any time-elapse possible from that state, there must be an additional (possibly 0) time elapse time leading to a state having at least one action transition; otherwise, that state would be a timelock state.
In addition, for any TCTL formula $\phi$, $\embed_{tn}(\phi)$ is \emph{closed}; there are no unbound propositional variables in $\embed_{tn}(\phi)$ for any $\TCTL$ formula $\phi$.  If a $\Lrelmunu$ formula $\gamma$ is closed, it follows that for any environments $\theta$ and $\theta'$, $\musemTAtheta{\gamma} = \musemTA{\gamma}{\theta'}$.  We will write $\musemTAnotheta{\gamma}{}$ for this common value in what follows when $\gamma$ is closed.

The next theorem states the correctness of $\embed_{tn}(-)$.

\begin{theorem}[Translation correctness for TF / nZ]\label{thm:divergent-ta-Lrelmunu-subsumes-TCTL}
Let $\TA$ be a TF / nZ timed automaton over time-safe $\Sigma$ and clock-safe $\AP$.
Then for all TCTL formulas $\phi$,
$
\tcsemTA{\phi} = \musemTAnotheta{\embed_{tn}(\phi)}.
$
\end{theorem}

The proof of this theorem is based on the following intuitions.  First, we note that the translations of TCTL operators that are also operators in $\Lrelmunu$ (atomic propositions, clock constraints, boolean connectives, freeze quantification) are immediate.  For $\tcE(\phi_1 \tcU \phi_2)$, because $\TA$ is TF / nZ, any execution to a state satisfying $\phi_2$ along which $\phi_1 \lor \phi_2$ holds can be extended into a run, because the state where $\phi_2$ is true must be the source of at least one run.
The case of $\tcA(\phi_1 \tcU \phi_2)$ is more complicated.
Note that in TF / nZ timed automata, any execution with an infinite number of action transitions is a run, i.e., it is time divergent.
Furthermore, any run from a state in such a timed automaton is either (1) a time-divergent sequence of time elapses, or (2) a finite (possibly empty) sequence of time elapses, followed by an action transition, followed by a run from the target of the action transition.
For a state $q$ to satisfy $\tcA(\phi_1 \tcU \phi_2)$, every run from that state must satisfy $\phi_1 \tcU \phi_2$.
The analysis now breaks into two cases.
First, if there is a $\delta$ delay in the current state $q$ such that $\delta(q)$ satisfies $\phi_2$, and all preceding states satisfy $\phi_1$, this identifies prefixes of runs of type (1) and those runs of type (2) in which $\phi_2$ becomes true before an action transition.
For any action transition that is taken before $\phi_2$ becomes true, we are in a run of type (2), and we need to check that from the target location, all runs satisfy $\phi_1 \tcU \phi_2$. Together, this is captured by $\exists_{\embed(\phi_1) \land [\Sigma]X}(\embed(\phi_2) \cdots)$.
Second, observe that if $\phi_2$ is never true after any time delay from $q$ then $q$ must be time-bounded.  In this case, $\phi_1$ needs to hold after every time-elapse from $q$, and all action transitions must lead to states satisfying the formula recursively.
This is captured by $\cdots \lor (\TStop \land \forall(\embed(\phi_1) \land [\Sigma]X))$.
We now prove Theorem~\ref{thm:divergent-ta-Lrelmunu-subsumes-TCTL} formally.

\begin{proof}[Proof of Theorem~\ref{thm:divergent-ta-Lrelmunu-subsumes-TCTL}]
Fix time-safe $\Sigma$ and clock-safe $\AP$, and let $\TA = \genTA$ be a TF / nZ timed automaton over $\Sigma$ and $\AP$.  Also let $\TS{\TA} = \ttsTA$ be the TTS associated with $\TA$.
Most cases are routine and left to the reader.  We consider here the cases when $\phi = \tcE(\phi_1 \tcU \phi_2)$ and $\phi = \tcA (\phi_1 \tcU \phi_2)$.  We use the following abbreviations.
\begin{align*}
Q_\phi
    &= \tcsemTA{\phi}
    \\
Q_1
    &= \tcsemTA{\phi_1}
    \\
Q_2
    &= \tcsemTA{\phi_2}
\\
Q_{12}
    &= \tcsemTA{\phi_1 \lor \phi_2}
\\
\runs{\phi}{q}
    &= \{ r \in \runs{\TA}{q} \mid r \in \untils{\TA}{Q_{12}}{Q_2}\}
\end{align*}
Note that $\runs{\phi}{q}$ consists of all runs from $q \in Q$ that are also until paths from $Q_{12}$ to $Q_2$.  If $r \in \runs{\phi}{q}$ then Lemma~\ref{lem:mup} guarantees the existence of a minimum prefix $\pi$ of $r$ such that $\pi \in \untils{\phi}{Q_{12}}{Q_2}$; we use $mup_\phi(r) \in \executions{\TA}{q}$ to refer to this minimum until-prefix of $r$.
\begin{description}
\item[$\phi = \tcE(\phi_1 \tcU \phi_2)$.]
    The induction hypothesis guarantees that $Q_1 = \musemTAnotheta{\embed_{tn}(\phi_1)}$ and $Q_2  = \musemTAnotheta{\embed_{tn}(\phi_2)}$, and thus also that $Q_{12} = \musemTAnotheta{\embed_{tn}(\phi_1 \lor \phi_2)}$.
    We introduce the following.
    \begin{align*}
    \gamma_1
        &= \embed_{tn}(\phi_1)
    \\
    \gamma_2
        &= \embed_{tn}(\phi_2)
    \\
    \gamma
        &=  \exists_{\gamma_1}
            (\gamma_2 \lor
                (\gamma_1 \land \dia{\Sigma} X))
    \\
    f_\gamma(S)
    &= \musemTA{\gamma}{\theta[x := S]}
    \end{align*}
    Note that $\gamma \in \Phirelmunu$ is a formula in $\Lrelmunu$,
    that $\embed_{tn}(\phi) = \mu X.\gamma$, that the only free propositional variable in $\gamma$ is $X$, and that $\gamma_1$ and $\gamma_2$ contain no free propositional variables.
    Also note that $f_\gamma \in 2^{\QTA} \to 2^{\QTA}$.
    We know that $\musemTAnotheta{\embed_{tn}(\phi)} = \mu f_\gamma$, where $\mu f_\gamma \subseteq \QTA$ is the least fixpoint of $f_\gamma$ over the complete lattice $(2^{\QTA}, \subseteq, \cup, \cap)$.

    We prove $Q_\phi = \musemTAnotheta{\embed_{tn}(\phi)}$ by showing that $Q_\phi \subseteq \musemTAnotheta{\embed_{tn}(\phi)}$ and $\musemTAnotheta{\embed_{tn}(\phi)} \subseteq Q_\phi$.
    For the former, it suffices to give a well-founded relation ${\prec} \subseteq Q_\phi \times Q_\phi$ such that $(Q_\phi, \prec)$ is a support structure for $f_\gamma$.
    We define $\prec$ by first introducing the notion of \emph{minimum action count}, $mac_{\phi}(q) \in \nats$, for $q \in S_\phi$, which is defined by
    \[
    mac_{\phi}(q) = \min \{\; |\, \indices_\Sigma(mup_\phi(r)) \,|
        \;\mid r \in \runs{\phi}{q} \}.
    \]
    Intuitively, $mac_{\phi}(q)$ is  the minimum, over all runs $r$ ensuring that $q \in Q_\phi$, of the number of action transitions in $r$'s minimum until-path prefix.
    We now define $q_1 \prec q_2$ iff $mac_\phi(q_1) < mac_\phi(q_2)$.  Clearly $\prec$ is well-founded.  We now must establish that $(Q_\phi, {\prec})$ is a support structure for $f_\gamma$.
    So pick $q \in S_\phi$; we must show that $q \in f_\gamma(\preimg{{\prec}}{q})$, where $\preimg{{\prec}}{q} = \{q' \in S_\phi \mid q' \prec q\}$.
    The proof consists of two subcases.
    \begin{itemize}
    \item $mac_\phi(q) = 0$.
        In this subcase, $\preimg{{\prec}}{q} = \emptyset$, and there is an $r \in R_\phi(q)$ such that $mup_\phi(r)$ contains no action transitions.
        From the semantics of $\Lrelmunu$ it can be be seen that $f_\gamma(\emptyset) = \musemTAnotheta{\exists_{\gamma_1} \gamma_2}$.
        That $q \in f_\gamma(\emptyset)$ is immediate from the semantics of $\Lrelmunu$ and the induction hypothesis.
    \item $mac_\phi(q) > 0$.
        In this case $\preimg{{\prec}}{q} \neq \emptyset$, and there exists $r \in \runs{\phi}{q}$ such that $|\indices_{\Sigma}(mup_\phi(r))| = mac_\phi(q) > 0$.
        Pick such an $r$, and let $\pi = mup_\phi(r)$; note that $|\pi| \geq 1$ since it must contain at least one action transition.  Define $r' = \suffix{r}{> |\pi|}$; note that $r = \pi \cdot r'$.
        Now let $i_\pi = \min (\indices_\Sigma(\pi))$ be the index of the first action transition in $\pi$; note that $\pi = \prefix{\pi}{< i_\pi} \cdot \pi[i_\pi] \cdot \suffix{\pi}{> i_\pi}$.
        We know that for every $(\delta,i) \in \indices_s(\prefix{\pi}{<i_\pi})$, $(\prefix{\pi}{<i_\pi})[\delta,i] = \pi[\delta,i]$, and thus Lemma~\ref{lem:mup} guarantees that for every $(\delta,i) \in \indices_s(\prefix{\pi}{<i_\pi})$, $\pi[\delta,i] \in Q_1 \setminus Q_2$.
        Now consider $q' = \tgt(\prefix{\pi}{<i_\pi}) = \src(\pi[i_\pi])$ and $q'' = \tgt(\pi[i_\pi])$; we will show that $q'' \in \preimg{{\prec}}{q}$ and from this conclude that $q \in f_\gamma(\preimg{{\prec}}{q})$, which is the desired result.  To see that $q'' \in \preimg{{\prec}}{q}$, first define $r'' = \suffix{\pi}{> i_\phi} \cdot r'$ and observe that $r'' \in \runs{\phi}{q''}$ and $\suffix{\pi}{> i_\pi} = mup_\phi(r'')$.
        As $|\indices_\Sigma(\pi')| < |\indices_\Sigma(\pi)|$ it therefore must be the case that $mac_\phi(q'') < mac_\phi(q)$, and thus $q'' \in \preimg{{\prec}}{q}$.  To see why this fact implies that $q \in f_\gamma(\preimg{{\prec}}{q})$, we first note that $\prefix{\pi}{<i}$ consists entirely of delay transitions, meaning that, taking $\delta = D(\prefix{\pi}{<i})$, $\delta(q) = \{ q' \}$, and for every $\delta' \leq \delta$, $\delta'(q) \cap (Q_1 \setminus Q_2) \neq \emptyset$.  Since $q' \ttrans{\Sigma} q''$ and $q'' \in \preimg{{\prec}}{q}$ the semantics of $\Lrelmunu$ and the induction hypothesis guarantee that $q' \in \musemTA{\gamma_1 \land \dia{\Sigma} X}{\theta[X := \preimg{{\prec}}{q}]}$ and that $q \in f_\gamma(\preimg{{\prec}}{q})$.
    \end{itemize}

    To finish the $\phi = \tcE(\phi_1 \tcU \phi_2)$ case we now establish that $\musemTAnotheta{\embed_{tn}(\phi)} \subseteq Q_\phi$.  It suffices to show that $f_\gamma(Q_\phi) \subseteq Q_\phi$, as the definition of $\musemTAnotheta{\mu X.\gamma} = \musemTAnotheta{\embed_{tn}(\phi)}$ gives the desired result.  So suppose that $q \in f_\gamma(Q_\phi)$; we show that $q \in Q_\phi$ by exhibiting a run $r \in \runs{\TA}{q}$ such that $r \in \untils{\TA}{Q_{12}}{Q_2}$.  Since $q \in f_\gamma(Q_\phi)$ the semantics of $\Lrelmunu$ ensure that there exists $\delta_q \in \delays$ such that the following hold.
    \begin{enumerate}
    \item
        $\delta_q(q) \cap \musemTA{\gamma_2 \lor (\gamma_1 \land \dia{\Sigma}X)}{\theta[X := Q_\phi]} \neq \emptyset$
    \item
        For all $\delta' < \delta_q$, $\delta'(q) \cap \musemTA{\gamma_1 \lor \gamma_2 \lor (\gamma_1 \land \dia{\Sigma}X)}{\theta[X := Q_\phi]}\neq \emptyset$
    \end{enumerate}
    The semantics and the induction hypothesis also imply the following.
    \begin{align*}
    \musemTA{\gamma_2 \lor (\gamma_1 \land \dia{\Sigma}X)}{\theta[X := Q_\phi]}
        &= Q_2 \cup (Q_1 \cap \musemTA{\dia{\Sigma} X}{\theta[X := Q_\phi]})
        \\
    \musemTA{\gamma_1 \lor \gamma_2 \lor (\gamma_1 \land \dia{\Sigma}X)}{\theta[X := Q_\phi]}
        &= Q_1 \cup Q_2 = Q_{12}
    \end{align*}
    Define $\{ q' \} = \delta_q(q)$.
    We now consider two subcases.  In the first $q' \in Q_2$.  Note that the single-transition execution $\pi = q \ttransTA{\delta_q} q'$ is such that $\pi \in \untils{\TA}{Q_{12}}{Q_2}$, and as $\TA$ is TF it further follows that there is a run $r' \in \runs{\TA}{q'}$.  Then $\pi \cdot r \in \runs{\TA}{q}$ and $\pi \cdot r \in \untils{\TA}{Q_{12}}{Q_2}$, and $q \in Q_\phi$.
    In the second subcase $q' \in Q_1 \cap \musemTA{\dia{\Sigma} X}{\theta[X := Q_\phi]}$, meaning $q' \in Q_1 \subseteq Q_{12}$ and $q' \ttransTA{a} q''$ for some $a \in \Sigma$ and $q'' \in Q_\phi$.  Since $q'' \in Q_\phi$ there exists a run $r'' \in \runs{\TA}{q'}$ such that $r'' \in \untils{\TA}{Q_{12}}{Q_2}$.  Now define execution $\pi' = q \ttransTA{\delta} q' \ttrans{a} q''$.  Clearly $r = \pi' \cdot r''$ is such that $r \in \runs{\TA}{q}$ and $r \in \untils{\TA}{Q_{12}}{Q_2}$, and thus $q \in Q_\phi$.  Consequently, $f_\gamma(Q_\phi) \subseteq Q_\phi$, and $\musemTAnotheta{\embed_{tn}(\phi)} \subseteq Q_\phi$.

\item[$\phi = \tcA(\phi_1 \tcU \phi_2)$.]
    The induction hypothesis guarantees that $Q_1 = \musemTAnotheta{\embed_{tn}(\phi_1)}$ and $Q_2 = \musemTAnotheta{\embed_{tn}(\phi_2)}$, and thus also that $Q_{12} = \musemTAnotheta{\embed_{tn}(\phi_1 \lor \phi_2)}$.
    In the argument to follow we will use the following definitions.
    \begin{align*}
    \eta_1
        &= \embed_{tn}(\phi_1) \land [\Sigma] X
    \\
    \eta_2
        &= \embed_{tn}(\phi_2)
    \\
    \eta
        &=  \au(\eta_1, \eta_2)
    \\
    f_{\eta}(S)
        &= \musemTA{\eta}{\theta[x := S]}
    \end{align*}
    Note that $\eta \in \Phirelmunu$ is a formula in $\Lrelmunu$,
    that $\embed_{tn}(\phi) = \mu X.\eta$, that the only free propositional variable in $\eta$ is $X$, and that $\embed_{tn}(\phi_1)$ and $\embed_{tn}(\phi_2)$ contain no free propositional variables.
    Also note that $f_{\eta} \in 2^{\QTA} \to 2^{\QTA}$, and that $\musemTAnotheta{\embed_{tn}(\phi)} = \mu f_{\eta}$.

    The proof now proceeds in a similar manner as the $\tcE(\phi_1 \tcU \phi_2)$ case:  we prove $Q_\phi = \musemTAnotheta{\embed_{tn}(\phi)}$ by establishing that $Q_\phi \subseteq \musemTAnotheta{\embed_{tn}(\phi)}$ and $\musemTAnotheta{\embed_{tn}(\phi)} \subseteq Q_\phi$.
    The proof of the former entails constructing a well-founded support structure for $f_\eta$ from $Q_\phi$; the latter relies on showing that $f_\eta(Q_\phi) \subseteq Q_\phi$.

    To continue the proof that $\tcsemTA{\phi} \subseteq \musemTAnotheta{\embed_{tn}(\phi)}$ we build a well-founded ${\prec} \subseteq Q_\phi \times Q_\phi$ such that $(Q_\phi, {\prec})$ is a support structure for $f_\eta$.
    To define $\prec$, let $q, q' \in Q_\phi$; note that since $q, q' \in Q_\phi$, $\runs{\TA}{q} = \runs{\phi}{q}$ and $\runs{\TA}{q'} = \runs{\phi}{q'}$.  Then $q' \prec q$ iff there exists $r \in \runs{\phi}{q}$, with $\pi = mup_\phi(r)$, such that there is $i \in \indices_\Sigma(\pi)$ with $\tgt(\pi[i]) = q'$.  In this case we call $\prefix{\pi}{\leq i}$ a \emph{witnessing execution} for $q' \prec q$. In words, $q' \prec q$ iff $q'$ is the target state of some action transition appearing in the minimum until-prefix of a run $r$ from $q$.

    We now establish that $\prec$ is well-founded.  So assume, by way of contradiction, that there is an infinite descending chain $\cdots \prec q_1 \prec q_0$.  From the definition of ${\prec}$ and Lemma~\ref{lem:mup} it must be the case that $q_i \in Q_1 \setminus Q_2$ for all $i \in \nats$, for if there is $i$ such that $q_i \in Q_2$ then $\preimg{{\prec}}{q_i} = \emptyset$.  Now let $\pi_{i, i+1}$ a witnessing execution for $q_{i+1} \prec q_i$; note that each $\pi_{i,i+1}$ contains at least one action transition.  We can now construct an infinite execution $r = \pi_{0,1} \cdot \pi_{1,2} \cdot \cdots$; note $r \in \executions{\TA}{q}$.  This execution contains an infinite number of action transitions, and since $\TA$ is nZ, it follows that it is a run; thus $\pi \in \runs{\TA}{q}$.  Lemma~\ref{lem:mup} and the fact that $q_i \in Q_1 \setminus Q_2$ guarantee that there is no $(\delta,i) \in \indices_s(\pi)$ such that $(\delta,i) \in Q_2$; thus $\pi \not\in \untils{\TA}{Q_{12}}{Q_2}$.  But this contradicts the fact that $q \in Q_\phi$ and thus $\pi \in \runs{\phi}{q}$; consequently, $\prec$ must be well-founded.

    We now finish the proof of the $Q_\phi \subseteq \musemTAnotheta{\embed_{tn}(\phi)}$ case by establishing that $(Q_\phi, {\prec})$ is a support structure for $f_\eta$.  So fix $q \in Q_\phi$, and let $S_q = \preimg{{(\prec)}}{q}$; we must show that $q \in f_\eta(S_q)$.
    We begin our argument by first noting that since $\runs{\TA}{q} = \runs{\phi}{q}$, $mup_\phi(r)$ exists for every $r \in \runs{\TA}{q}$.
    Now consider $r \in \runs{\TA}{q}$, with $\pi = mup_\phi(r)$, and suppose that that $(\delta,i) \in \indices_s(\pi)$ is such that for all $(\delta',i') \leq_\pi (\delta,i)$, $\pi[\delta',i'] \not\in Q_2$.
    Also suppose that
    $\pi[\delta,i] \ttransTA{\sigma} q'$ for some $\sigma \in \Sigma$ and $q' \in Q$.
    Note that since $\TA$ is TF, there exists run $r' \in \runs{\TA}{q'}$.  Then $r''$ defined as
    \[
        r'' =
            ext(r,\delta,i)
            \cdot (\pi[\delta,i] \ttransTA{\sigma} q')
            \cdot r'
    \]
    satisfies:
    $r'' \in \runs{\TA}{q}$.
    Intuitively, $r''$ matches $r$ until state $\pi[\delta,i] = r[\delta,i]$, at which point $r''$ takes one of the action transitions available to $r[\delta,i]$ to transition to $q'$ and then follows $r'$.
    Clearly, $r'' \in \untils{\TA}{Q_{12}}{Q_2}$, which, based on the construction of $r''$, further implies that $q' \in Q_\phi$.  From the definitions it then follows that $q' \in S_q$.
    Thus, for any $q \in Q_\phi$ such that, for some run $r \in \runs{\TA}{q}$ and $(\delta,i) \in \indices_s(r)$ with the property that for all $(\delta',i') \leq_r (\delta,i)$, $r[\delta',i'] \not\in Q_2$, we have $r[\delta,i] \in \musemTA{[\Sigma]X}{\theta[X := S_q]}$.
    We now recall that from the definition of $f_\eta$ and Lemma~\ref{lem:semantics-of-au},
    \begin{align*}
        f_\eta(S_q)
            &= \musemTA{\au(\eta_1, \eta_2)}{\theta[X := S_q]}
        \\
            &=  \musemTA{\exists_{\eta_1} \eta_2}{\theta[X := S_q]}
                \cup
                \musemTA{\exists\, \TStop \land \forall\, \eta_1}{\theta[X := S_q]}
    \end{align*}
    To prove that $q \in f_\eta(S_q)$, assume $q \not\in \musemTA{\exists_{\eta_1} \eta_2}{\theta[X := S_q]}$; we must show that $q \in \musemTA{\exists\, \TStop \land \forall\, \eta_1}{\theta[X := S_q]}$.  To see that $q \in \musemTA{\exists\, \TStop}{\theta[X := S_q]}$, assume that this is not the case, i.e.\/ assume $q$ is not time-bounded and thus has a run $r \in \runs{\TA}{q}$ consisting entirely of time-elapse transitions.  Since $r$ is a run and $Q_\phi$, it must be the case that $r \in \untils{\TA}{Q_{12}}{Q_2}$, which also implies that $q \in \musemTA{\exists_{\eta_1} \eta_2}{\theta[X := S_q]}$, yielding a contradiction.  Thus $q \in \musemTA{\exists\, \TStop}{\theta [X:=S_q]}$.
    We now show that $q \in \musemTA{\forall\, \eta_1}{\theta[X := S_a]}$ by showing that for all $\delta, q'$ such that $q \ttransTA{\delta} q'$, $q' \in \musemTA{\eta_1}{\theta[X := S_q]}$.  So assume $q \ttransTA{\delta} q'$.
    We know that execution $q \ttransTA{\delta} q' \not\in \untils{\TA}{Q_{12}}{Q_2}$ since $q \not\in \musemTA{\exists_{\eta_1} \eta_2}{\theta[X := S_q]}$.  Now suppose that $\sigma \in \Sigma$ and $q''$ are such that $q' \ttransTA{\sigma} q''$.  Since $\TA$ is TF $\runs{\TA}{q''} \neq \emptyset$, and every run $r'' \in \runs{\TA}{q''}$ is such that $(q \ttransTA{\delta} q' \ttransTA{\sigma} q'') \cdot r'' \in \untils{\TA}{Q_{12}}{Q_2}$, with $q \ttransTA{\delta} q' \ttransTA{\sigma} q''$ a prefix of $mup_\phi(r')$.  This implies that $q'' \in Q_\phi$ and $q'' \prec q$ and thus $q'' \in \musemTA{X}{\theta[X := S_q]}$.  It then follows that $q' \in \musemTA{[\Sigma] X}{\theta[X := S_q]}$.  It also implies that $q' \in Q_{12}$, and since $q'$ cannot be in $Q_2$ (otherwise $q \in \musemTA{\exists_{\eta_1} \eta_2}{\theta[X := S_1]}$), we have that $q \in Q_1$.  Consequently, $q \in \musemTA{\eta_1}{\theta[X := S_q]}$, thereby completing the proof that $(Q, \prec)$ is a well-founded support ordering for $f_\mu$.

    To finish the proof that $Q_\phi = \musemTAnotheta{\embed_{tn}(\phi)}$, we show that $\musemTAnotheta{\embed_{tn}(\phi)} \subseteq Q_\phi$ by establishing that $f_\eta(Q_\phi) \subseteq Q_\phi$.  To this end, suppose that $q \in f_\eta(Q_\phi)$; we show that $q \in Q_\phi$.
    From the definition of $f_\eta$ and Lemma~\ref{lem:semantics-of-au} we know that $q \in \musemTA{\exists_{\eta_1} \eta_2}{\theta[X := S_q]} \cup \musemTA{(\exists\, \TStop) \land (\forall\, \eta_1)}{\theta[X := S_1]}$; we must show that $q \in Q_\phi$.  The proof proceeds by a case split:  the first case involves $q \in \musemTA{\exists_{\eta_1} \eta_2}{\theta[X := S_q]}$.  In this case there exists $\delta \in \delays$ such that $\delta(q) \cap \musemTA{\eta_1}{\theta[X := Q_\phi]} \neq \emptyset$ and for all $\delta' < \delta, \delta'(q) \cap \musemTA{\eta_1 \lor \eta_2}{\theta[X := Q_\phi]}$.  Define $\Delta_U$ to be the set of all $\delta$ with this property; that is,
    \begin{align*}
    &\Delta_U = \{ \delta \in \delays \mid
    \delta(q) \cap \musemTA{\eta_1}{\theta[X := Q_\phi]} \neq \emptyset \\
    & \qquad\qquad \land
        \forall \delta' \in \delays \colon
        \delta' < \delta \implies \delta'(q) \cap \musemTA{\eta_1 \lor \eta_2}{\theta[X := Q_\phi]}  \neq \emptyset
    \}
    \end{align*}
    Since $\Delta_U \neq \emptyset$ is bounded below (by 0), $\delta_q = \inf \Delta_U$ is well-defined, and has the property that $\delta_q \leq \delta$ for all $\delta \in \Delta_U$.
    Note that it can be the case that either $\delta_q \in \Delta_U$ or $\delta_q \not\in \Delta_U$.
    We now establish that $q \in Q_\phi$ by showing that for all $r \in \runs{\TA}{q}$, $r \in \untils{\TA}{Q_{12}}{Q_2}$.  The argument uses a case split on the form of $r \in \runs{\TA}{q}$.
    \begin{itemize}
    \item
        For all $(\delta,i) \in \indices_s(r)$ such that $\delta \leq \delta_q$, $i=0$.
        In this case $r$ begins with time-elapse transitions whose duration is at least $\delta_q$ and which has no action transitions at $\delta_q$ (because $(\delta_q, i) \not\in \indices_s$ if $i \neq 0$).  It then follows that for all $(\delta,i) \in \indices_s(r)$, $q \ttransTA{\delta} r[\delta,i]$. That $r \in \untils{\TA}{Q_{12}}{Q_2}$ follows from the fact that $\musemTA{\eta_1 \lor \eta_2}{\theta[X := Q_\phi]} \subseteq Q_{12}$.
    \item
        There exists $(\delta,i) \in \indices_s(r)$ such that $\delta \leq \delta_q$ and $i \neq 0$.  This implies that there is $i \in \indices_\Sigma(r)$ such that $st_r(i) \leq \delta_q$; in other words, $r$ contains an action transition that occurs at or before $\delta_q$ time has elapsed.  Let $i_q = \min I_\Sigma(r)$ be the index of the first such transition in $r$.  There are now two subcases to consider.  In the first, $D(\prefix{r}{< i_q}) = \delta_q$ and $\delta_q \in \Delta_U$; in other words, the first action transition happens at time $\delta_q$, but with $\delta_q(q) \in Q_2$, and the reasoning is the same as above.  In the second subcase $r = \prefix{r}{<i_q} \cdot r[i_q] \cdot \suffix{r}{> i_q}$.  It is straightforward to see that for all $(\delta,i) \in \indices_s(\prefix{r}{<i_q}), r[\delta,i] \in \musemTA{\eta_1}{\theta[X := Q_\phi]}$ and thus $\src(r[i_q]) \in \musemTA{\eta_1}{\theta[X := Q_\phi]}$.  This fact also implies that $\tgt(r[i_q]) \in Q_\phi$, which implies that $\suffix{r}{>i_q} \in \untils{\TA}{Q_{12}}{Q_2}$ and thus $r \in \untils{\TA}{Q_{12}}{Q_2}$.
    \end{itemize}
    Thus, if $q \in \musemTA{\exists_{\eta_1} \eta_1}{\theta[X := Q_\phi]}$ then all $r \in \runs{\TA}{q}$ are in $\untils{\TA}{Q_{12}}{Q_2}$ and $q \in Q_\phi$.

    The second case involves $q \in \musemTA{(\exists\, \TStop) \land (\forall\, \eta_1)}{\theta[X := S_1]}$.  Given the above argument, we may further assume that $q \not\in \musemTA{\exists_{\eta_1} \eta_2}{\theta[X := Q_\phi]}$.  To show that $q \in Q_\phi$ we show that $\runs{\TA}{q} \subseteq \untils{\TA}{Q_{12}}{Q_2}$.  So pick $r \in \runs{\TA}{q}$.  Based on the assumptions this run may be written as $r = \pi_1 \cdot \pi_2 \cdot r'$, where $\pi_1$ consists of a finite number of time-elapse transitions, $\pi_2$ is a single action transition, and $r'$ is a run from $\tgt(\pi_2)$.  Since $q \not\in \musemTA{\exists_{\eta_1} \eta_2}{\theta[X := Q_\phi]}$ it follows that for every $(\delta,i) \in \indices_s(\pi_1), r[\delta,i] \in \musemTA{\eta_1}{\theta[X := Q_\phi]}$.  This in particular holds for $\tgt(\pi_1) = \src(\pi_2)$, meaning that $\src(\pi_2) \in \musemTA{[\Sigma] X}{\theta[X := Q_\phi]}$ and $\tgt(\pi_2) \in Q_\phi$.  But then $r' \in \runs{\TA}{\tgt(\pi_2)}$ and $r' \in \untils{\TA}{Q_{12}}{Q_2}$, and consequently $r \in \untils{\TA}{Q_{12}}{Q_2}$.  Therefore, if $q \in \musemTA{(\exists\, \TStop) \land (\forall\, \eta_1)}{\theta[X := S_1]}$ and $q \not\in \musemTA{\exists_{\eta_1} \eta_2}{\theta[X := Q_\phi]}$ then $q \in Q_\phi$.
    We have shown that if $q' \in f_\eta(Q_\phi)$ then $q' \in Q_\phi$.  Thus $\musemTAtheta{\embed_{tn}(\phi)} \subseteq Q_\phi$, and the proof is complete.\qedhere
\end{description}
\end{proof}

\paragraph{Translation for general timed automata.}
We now show how the translations of $TCTL$ into $\Lrelmunu$ given above can be adapted so that they are correct for all timed automata, and not just those satisfying the TF / nZ assumption.  In particular, we will define a transformation $\embed$ so that for all $\TCTL$ formulas $\phi$ and all timed automata $\TA$, $\embed(\phi) \in \Phirelmunu$ has the same semantics as $\phi$.  We will first consider how to translate $\tcE(\phi_1 \tcU \phi_2)$, then focus on $\tcA(\phi_1 \tcU \phi_2)$.

In the case of $\tcE(\phi_1 \tcU \phi_2)$, the proof of Theorem~\ref{thm:divergent-ta-Lrelmunu-subsumes-TCTL} only requires the timelock-free property to show the correctness of the translation for TF / nZ timed automata.
Timelock-freedom in particular allows the extension of any finite execution $\pi \in \untils{\TA}{Q_{12}}{Q_2}$ into a run $r \in \untils{\TA}{Q_{12}}{Q_2}$. This fact was crucial in the proof of Theorem~\ref{thm:divergent-ta-Lrelmunu-subsumes-TCTL}.
If a timed automaton is not TF, then such extensions are not guaranteed, and the translation $\embed_{tn}(\tcE \phi_1 \tcU \phi_2)$ is not guaranteed to be correct.  To handle $\tcE(\phi_1 \tcU \phi_2)$ for general timed automata, we introduce a formula, $\TDIV$, that holds of states that have at least one run, and then incorporate this appropriately into the translation.  $\TDIV$ is given as follows, where $X$ and $Y$ are distinct variables in $\Var$; note that $z \in \clocksF$ from the definition of $\Lrelmunu$ and thus cannot appear in any timed automaton.  It is adapted from similar formulas in~\cite{fontana-2014,henzinger-symbolic-model-1994}.
\[
\TDIV = \nu X . z . (\mu Y . \exists\,((z \geq 1 \land X) \lor \dia{\Sigma} Y))
\]
This formula contains an alternating fixpoint; in particular, the bound variable, $X$, of the outer greatest fixpoint appears free in the body of the inner, least fixpoint.  States satisfying this formula are allowed to perform a finite number of action and delay transitions during any time interval whose duration is at least 1.
Note that $\TDIV$ is closed, and thus for any $\TA$ and environments $\theta, \theta'$, $\musemTAtheta{\TDIV} = \musemTA{\TDIV}{\theta'}$.  As before, we write $\musemTAnotheta{\TDIV}$ for this environment-independent value.  The next lemma characterizes the semantics of $\TDIV$.

\begin{lemma}[Semantics of $\TDIV$]\label{lem:semantics-of-TDIV}
Let $\TA$ be a timed automaton, and let $\TS{\TA} = (\QTA, \ldots)$.  Then $\musemTAnotheta{\TDIV} = \{q \in \QTA \mid \runs{\TA}{q} \neq \emptyset\}$.
\end{lemma}
\begin{proof}
Fix $\TA$ and $\TS{\TA} = \ttsTA$.  Let $\theta \in \Var \to 2^{\QTA}
$ be arbitrary, and define the following, where $S, S' \subseteq \QTA$.
\begin{align*}
Q_{\TDIV}
    &= \{q \in \QTA \mid \runs{\TA}{q} \neq \emptyset\}
    \\
\gamma_1
    &= z \geq 1 \land X
    \\
\gamma_2
    &= \dia{\Sigma}Y
    \\
\gamma_Y
    &= \mu Y.\exists\,(\gamma_1 \lor \gamma_2)
    \\
f'_S(S')
    &= \musemTA{\exists\,(\gamma_1 \lor \gamma_2)}{(\theta[X := S])[Y := S']}
    \\
f_{\TDIV}(S)
    &= \musemTA{z . \gamma_Y}{\theta[X := S]}
\end{align*}
Note that, for any $S \subseteq \QTA$, $f'_S \in 2^{\QTA} \to 2^{\QTA}$ and $\musemTA{\gamma_Y}{\theta[X := S]} = \mu f'_S$, where $\mu f'_S \subseteq \QTA$ is the least fixpoint of $f'_S$ over complete lattice $(2^{\QTA}, \subseteq, \bigcup, \bigcap)$.
Similarly, $f \in 2^{\QTA} \to 2^{\QTA}$, and $\musemTAnotheta{\TDIV} = \nu f_{\TDIV}$, the greatest fixpoint of $f_{\TDIV}$ over the same lattice.  The Tarksi-Knaster characterization of this greatest fixpoint is as follows.
\[
\nu f_{\TDIV} = \bigcup \{S \subseteq \QTA \mid S \subseteq f_{\TDIV}(S) \}
\]
Before turning to the proof of the lemma, we first establish the following characterization of $f_{\TDIV}(S)$, where $S \subseteq \QTA$.
\begin{align}\label{eq:f-tdiv}
f_{\TDIV}(S) = \{ q \in \QTA \mid \exists \pi \in \executions{\TA}{q} \colon |\pi| < \infty \land D(\pi) \geq 1 \land \tgt(\pi) \in S \}
\end{align}
In other words, states are in $f_{\TDIV}(S)$ iff they have a finite-length execution of duration at least 1 that ends in a state in $S$.
For notational convenience, we define the following.
\begin{align*}
Q_S
    &= \{ q \in \QTA \mid \exists \pi \in \executions{\TA}{q} \colon |\pi| < \infty \land D(\pi) \geq 1 \land \tgt(\pi) \in S \}
    \\
Q'_S
    &= \{ q \in \QTA \mid \exists \pi \in \executions{\TA}{q} \colon |\pi| < \infty \land D(\pi) \geq 1-q(z) \land \tgt(\pi) \in S \}
\end{align*}
$Q'_S$ differs from $Q_S$ in that the constraint on the duration of $\pi$ in the definition is relative to $1- q(z)$, where recall that $q(z)$ is the value state $q$ assigns to clock $z$.  Since $(q[z:=0])(z) = 0$, it follows that $q[z := 0] \in Q'_S$ iff $q[z:=0] \in Q_S$.  Likewise, since $z \in \clocksF$ it follows that $q[z := 0] \in Q_S$ iff $q \in Q_S$.  Consequently, $q \in Q_S$ iff $q[z := 0] \in Q'_S$.

To prove the validity of Equation~\ref{eq:f-tdiv} we now must establish that $f_{\TDIV}(S) = Q_S$ for any $S \subseteq \QTA$.  So fix $S \subseteq \QTA$; it suffices to show that $f_{\TDIV}(S) \subseteq Q_S$ and $Q_S \subseteq f_{\TDIV}(S)$.
For the former, assume $q \in f_{\TDIV}(S)$; we must show that $q \in Q_S$.
The semantics of the $z.$\/ operator guarantees that since $q \in f_{\TDIV}(S)$, it is the case that $q[z := 0] \in \musemTA{\gamma_Y}{\theta[X := S]} = \mu f'_S$.
If $\mu f'_S \subseteq Q'_S$ then we have $q[z := 0] \in Q'_S$, whence $q \in Q_S$ and the set inclusion $f_{\TDIV}(S) \subseteq Q_S$ has been proved.
So we now show that $\mu f'_S \subseteq Q'_S$.  Based on the definition of $\mu f'_S$ this follows if $f'_S(Q'_S) \subseteq Q'_S$.  So assume $q' \in f'_S(Q'_S)$; we must show that $q' \in Q'_S$.
Letting $\theta' = (\theta[X := S])[Y := Q'_S]$, the semantics of $\Lrelmunu$ guarantee that
\[
f'_S(Q'_S) =
\musemTA{\exists\,\gamma_1}{\theta'}
\cup
\musemTA{\exists\,\gamma_2}{\theta'}.
\]
We now do a case analysis on whether $q' \in \musemTA{\exists\,\gamma_1}{\theta'}$ or $q' \in \musemTA{\exists\,\gamma_2}{\theta'}$.  In the former case there exists $\delta' \in \delays$ and $q'' \in S$ such that $q' \ttransTA{\delta} q''$ and $q'' \in \musemTA{\gamma_1}{\theta'}$.  From the definition of $\gamma_1$ this means $q''(z) = q'(z) + \delta \geq 1$ and $q'' \in \musemTA{X}{\theta'} = S$.  If we take single-transition execution $\pi = q' \ttransTA{\delta'} q''$ we see that $D(\pi) = \delta \geq 1 - q'(z)$ and $\tgt(\pi) \in S$, and thus $q' \in Q'_S$.
In the latter case $q' \in \musemTA{\exists\,\gamma_2}{\theta'}$.  This means there is an execution $\pi' = q' \ttransTA{\delta} q_1 \ttransTA{a} q_2$ for some $\delta \in \delays$, $a \in \Sigma$, $q_1 \in \QTA$ and $q_2 \in Q'_S$.  Since $q_2 \in Q'_S$ there must be a finite $\pi_2 \in \executions{\TA}{q_2}$ with $| \pi| < \infty$, $D(\pi_2) \geq 1-q_2(z)$ and $\tgt(\pi_2) \in S$.  Now consider execution $\pi = \pi' \cdot \pi_2$.  Clearly $\pi \in \executions{\TA}{q'}$ and $|\pi| < \infty$.  Also, $q_2(z) = q'(z) + \delta$, meaning $D(\pi) = D(\pi_2) + \delta$ and thus $D(\pi) \geq 1 - q'(z)$.  Since $\tgt(\pi) = \tgt(\pi_2)$, $\tgt(\pi) \in S$, and it follows that $q' \in Q'_S$ in this case also.
Consequently $f'_S(Q'_S) \subseteq Q'_S$, $\mu f'_S \subseteq Q'_S$, $q[z := 0] \in Q'_S$, $q \in Q_S$, and $f_{\TDIV}(S) \subseteq Q_S$.

To finish proving the validity of Equation~(\ref{eq:f-tdiv}) we now must show that $Q_S \subseteq f_{\TDIV}(S)$.  Based on the semantics of $z.$, this holds if and only if $Q_S[z := 0] \subseteq \musemTA{\gamma_Y}{\theta[X := S]} = \mu f'_S$, where $Q_S[z := 0] = \{ q[z := 0] \mid q \in Q_S \}$.
Note that $Q_S[z := 0] \subseteq Q_S$ since $z \in \clocksF$; this implies that $Q_S[z := 0] \subseteq Q'_S$.
Consequently, $Q_S \subseteq f_{\TDIV}$ follows if $Q'_S \subseteq \mu f'_S$.  We prove this by defining a well-founded relation ${\prec} \subseteq Q'_S \times Q'_S$ such that $(Q'_S, {\prec})$ is a support structure for $f'_S$.
To define $\prec$, let $q' \in Q'_S$, and let $mac(q) \in \nats$ be given as follows.
\[
mac(q) = \min \{ |\indices_\Sigma(\pi)| \mid \pi \in \executions{\TA}{q} \land |\pi| < \infty \land D(\pi) \geq 1-q(z) \land \tgt(\pi) \in S\}
\]
The measure $mac(q)$ is the minimum number of action transitions appearing in a run from $q$ whose length is finite, whose duration is at least $1-q(z)$, and whose target lies in $S$.
Since every $q \in Q'_S$ has at least one $\pi \in \executions{\TA}{q}$ such that $|\pi| < \infty, D(\pi) \geq 1-q(z)$ and $\tgt(\pi) \in S$, $mac(q) \in \nats$ is well-defined.  We now define $q' \prec q$ iff $mac(q') < mac(q)$.  This relation is clearly well-founded.  We now establish that $(Q'_S, {\prec})$ is a support structure for $f'_S$.  So assume $q \in Q'_S$; we must show that $q \in f'_S(S_q)$, where $S_q = \preimg{{\prec}}{q}$.
We know that $f'_S(S_q) = \musemTA{\exists\,\gamma_1}{\theta'} \cup \musemTA{\exists\, \gamma_2}{\theta'}$, where $\theta' = (\theta[X := S])[Y := S_1]$.
Thus, to prove that $q \in f'_S(S_q)$ it suffices to show that, under the assumption that $q' \not\in \musemTA{\exists\,\gamma_1}{\theta'}$, $q \in \musemTA{\exists\, \gamma_2}{\theta'}$.  So assume $q \not\in \musemTA{\exists\,\gamma_1}{\theta'}$.
It can be seen that this implies $mac(q) \geq 1$.  Now pick $\pi \in \runs{\TA}{q}$ such that $|\pi| < \infty$, $D(\pi) \geq 1-q(z)$, $\tgt(\pi) \in S$ and $\indices_\Sigma(\pi) = mac(q)$.  Let $i \in \indices_t(\pi)$ be the smallest $i$ such that $\lab(\pi[i]) \in \Sigma$.  Clearly there exists $\delta \in \delays$ such that $q \ttransTA{\delta} \src(\pi[i])$; also, $q' = \tgt(\pi[i])$ satisfies $q(z') = q(z) + \delta$ and $mac(q') < mac(q)$, meaning $q' \in \preimg{{\prec}}{q}$.  Consequently, $q \in \musemTA{\dia{\Sigma}Y}{\theta'} = \musemTA{\gamma_2}{\theta'}$, and $(Q, {\prec})$ is a support structure for $f'_S$.
Since ${\prec}$ is well-founded and $(Q'_S, {\prec})$ is a support structure for $f'_S$, $Q'_S \subseteq \mu f'_S$ and thus $Q_S \subseteq f_{\TDIV}(S)$.
This completes the proof of Equation~(\ref{eq:f-tdiv}).

We now turn to proving the main lemma, whose conclusion can be rephrased as: $Q_{\TDIV} = \nu f_{\TDIV}$.  We establish this by showing that $Q_{\TDIV} \subseteq \nu f_{\TDIV}$ and $\nu f_{\TDIV} \subseteq Q_{\TDIV}$.  For the former it suffices to show that $Q_{\TDIV} \subseteq f_{\TDIV}(Q_{\TDIV})$.  So suppose $q \in Q_{\TDIV}$.  It follows that there is a run $r \in \runs{\TA}{q}$, and it also immediately follows that there is a state $q'$, execution $\pi \in \executions{\TA}{q}$, and run $r' \in \runs{\TA}{q'}$ such that $|\pi| < \infty$, $D(\pi) \geq 1$, $\tgt(\pi) = q'$ and $r = \pi \cdot r'$.  But then $q' \in Q_{\TDIV}$, and Equation~(\ref{eq:f-tdiv}) then guarantees $q \in f_{\TDIV}(Q_{\TDIV})$.  Thus $Q_{\TDIV} \subseteq \nu f_{\TDIV}$.
To show that $\nu f_{\TDIV} \subseteq Q_{\TDIV}$ it suffices to show that for all $S \subseteq \QTA$ such that $S \subseteq f_{\TDIV}(S)$, $S \subseteq Q_{\TDIV}$.
So fix $S \subseteq f_{\TDIV}(S)$, and consider $q_0 \in S$.  Equation~(\ref{eq:f-tdiv}) guarantees that there is a $\pi_{0,1} \in \executions{\TA}{q}$ such that $|\pi_{0,1}| < \infty$, $D(\pi_{0,1}) \geq 1$, and $\tgt(\pi_{0,1}) = q_1$ is such that $q_1 \in S$.  We may similarly construct executions $\pi_{1,2}, \pi_{2,3}, \ldots$, each being of finite length and duration at least 1, with source state $q_i$ and target state $q_{i+1}$.
From these we may (co-inductively) construct run $r = \pi_{0,1} \cdot \pi_{1,2} \cdot \cdots \cdot \pi_{i,i+1} \cdots$.
By construction $r \in \runs{\TA}{q_0}$ and thus $q \in Q_{\TDIV}$, thereby establishing that $S \subseteq Q_{\TDIV}$ for every $S$ such that $S \subseteq f_{\TDIV}(S)$.
Thus $\nu f_{\TDIV} \subseteq Q_{\TDIV}$.  This completes the proof.\qedhere
\end{proof}

Based Lemma~\ref{lem:semantics-of-TDIV} one can also see that if $q \not\in \musemTAnotheta{\TDIV}$ then for any $a \in \Delta(\Sigma)$ and $q' \in \QTA$ such that $q \ttransTA{a} q'$, $q' \not\in \musemTAnotheta{\TDIV}$.  If this were not the case for some $q$ and $q'$ then $q'$ would have a run, and so would $q$, which would contradict Lemma~\ref{lem:semantics-of-TDIV}.

%
We can now define the generalized translation for $\tcE(\phi_1 \tcU \phi_2)$ as follows; the addition to $\embed_{tn}(\tcE(\phi_1 \tcU \phi_2))$ is \underline{underlined}.
\[
\embed(\tcE(\phi_1 \tcU \phi_2)) = \mu X . \exists_{\embed(\phi_1)}((\embed(\phi_2)\: \underline{\land\: \TDIV}) \lor (\embed(\phi_1) \land \dia{\Sigma} X))
\]

The following lemma states the correctness of this translation.

\begin{lemma}[Correctness of $\embed(\tcE (\phi_1 \tcU \phi_2)$]\label{lem:correctness-of-EU}
Let $\TA$ be a timed automaton, and suppose that $\phi_1, \phi_2$ are $\TCTL$ formulas such that $\tcsemTA{\phi_1} = \musemTAnotheta{\embed(\phi_1)}$ and $\tcsemTA{\phi_2} = \musemTAnotheta{\embed(\phi_2)}$.  Then $\tcsemTA{\tcE(\phi_1 \tcU \phi_2)} = \musemTAnotheta{\embed(\tcE(\phi_1 \tcU \phi_2))}$.
\end{lemma}
\begin{proof}
Formula $\embed(\tcE(\phi_1 \tcU \phi_2))$ replaces $\embed_{tn}(\phi_2)$ in $\embed_{tn}(\tcE(\phi_1 \tcU \phi_2))$ with $\embed(\phi_2) \land \TDIV$. This requires that any state state satisfying $\phi_2$ as part of a determination that another state satisfies $\embed(\tcE(\phi_1 \tcU \phi_2))$ also must have a run.  The correctness of the translation is very similar to the argument given for the correctness of $\embed_{tn}(\tcE(\phi_1 \tcU \phi_2))$ in the proof of Theorem~\ref{thm:divergent-ta-Lrelmunu-subsumes-TCTL} and is left to the reader.\qedhere
\end{proof}

For $\tcA(\phi_1 \tcU \phi_2)$, the situation is more complicated, and we take a staged approach to define $\embed(\tcA(\phi_1 \tcU \phi_2))$. First, we drop the TF assumption and give a translation, $\embed_n(\tcA(\phi_1 \tcU \phi_2))$, for nZ timed automata. The definition is given as follows, with additions to $\embed_{tn}(\tcA(\phi_1 \tcU \phi_2))$ \underline{underlined}.
\begin{align*}
\embed_n(\tcA(\phi_1 \tcU \phi_2))
&=
\mu X. \au(\underline{\TDIV \implies} \,(\embed(\phi_1) \land [\Sigma]X), \embed(\phi_2))
\end{align*}
The intuitions are as follows. We first note that every timelock state trivially satisfies $\tcA(\phi_1 \tcU \phi_2)$, since such states have no runs.
Such states also immediately satisfy any $\Lrelmunu$ implication of form $\TDIV \implies \cdots$. Now consider a state that satisfies $\TDIV$, and thus has at least one run emanating from it.  There are two ways that this state can satisfy $\tcA(\phi_1 \tcU \phi_2)$ in this case.  In the first, there is a time-elapse from the state that makes $\phi_2$ true, with every intervening state before this point required to keep either $\phi_1$ or $\phi_2$ true; in addition, every state reachable from one of these intervening states via an action transition must also satisfy $\tcA(\phi_1 \tcU \phi_2)$.  The first condition ensures that any run from the original state that begins with a sequence of time-elapse transitions and which transitions through a state satisfying $\phi_2$ as a result keeps the until property true.  The second condition ensures that any runs that exercise an action transition before reaching this $\phi_2$ state make the until condition true.
In the second case there is no such state reachable via time-elapse transitions from the original state that makes $\phi_2$ true.  In this case, every run contains at least one action transition and must keep $\phi_1$ true before that transition occurs.  In addition, time is bounded in the state, as otherwise a time-elapse-only run would violate the until property.
These cases are handled by the rest of the implication following $\TDIV \implies \cdots$.

The next lemma established the correctness of $\embed_n(\tcA(\phi_1 \tcU \phi_2))$ for nZ timed automata.

\begin{lemma}[Correctness of $\embed_n(\tcA(\phi_1 \tcU \phi_2))$]\label{lem:correctness-of-nau}
Let $\TA$ be a non-Zeno timed automaton, and suppose that $\phi_1, \phi_2$ are $\TCTL$ formulas such that $\tcsemTA{\phi_1} = \musemTAnotheta{\embed(\phi_1)}$ and $\tcsemTA{\phi_2} = \musemTAnotheta{\embed(\phi_2)}$.  Then
\[
\tcsemTA{\tcA(\phi_1 \tcU \phi_2)}
=
\musemTAnotheta{\embed_n(\tcA(\phi_1 \tcU \phi_2))}.
\]
\end{lemma}

\begin{proof}
Fix $\TA$, $\phi_1$ and $\phi_2$ as above, and let $\TS{\TA} = \ttsTA$.
The proof proceeds in two steps.  In the first, we show that $\musemTAnotheta{\lnot\TDIV} \subseteq \tcsemTA{\tcA(\phi_1 \tcU \phi_2)}$  and $\musemTAnotheta{\lnot\TDIV} \subseteq \musemTAnotheta{\embed_{n}(\tcA(\phi_1 \tcU \phi_2))}$.  Then, in the second, we show that $\musemTAnotheta{\TDIV} \cap \tcsemTA{\tcA(\phi_1 \tcU \phi_2)} = \musemTAnotheta{\TDIV} \cap \musemTAnotheta{\embed_{n}(\tcA(\phi_1 \tcU \phi_2))}$.  It immediately follows that $\tcsemTA{\tcA(\phi_1 \tcU \phi_2)} = \musemTAnotheta{\embed_n(\tcA(\phi_1 \tcU \phi_2))}$.

For the first of these, note that based on Lemma~\ref{lem:semantics-of-TDIV}, $q \in \musemTAnotheta{\lnot\TDIV}$ iff $\runs{\TA}{q} = \emptyset$.  Now assume that $q \in \musemTAnotheta{\lnot\TDIV}$.  From the definition of $\tcsemTA{\tcA (\phi_1 \tcU \phi_2)}$, $q \in \tcsemTA{\tcA (\phi_1 \tcU \phi_2)}$ holds vacuously, since $q$ has no runs.  To show that $q \in \musemTAnotheta{\embed_n(\tcA(\phi_1 \tcU \phi_2))}$ we must show that $q \in \musemTAnotheta{\mu X.\au(\TDIV \implies (\embed_s(\phi_1) \land [\Sigma]X), \embed_s(\phi_2))}$.  This follows immediately from the following facts and Lemma~\ref{lem:semantics-of-au}.
\begin{itemize}
\item
    $q \in \musemTAtheta{\TDIV \implies \phi}$ for any $\Lrelmunu$ formula and environment $\theta$.
\item
    $q \in \musemTAnotheta{\exists\,\TStop}$ (since $q$ has no runs and thus no time-divergent executions).
\item
    $q \in \musemTAtheta{\forall\, (\TDIV \implies \phi)}$ for any $\Lrelmunu$ formula and environment $\theta$, since for every $\delta$ and $q'$ such that $q \ttransTA{\delta} q'$, $q' \in \musemTAnotheta{\lnot\TDIV}$.
\end{itemize}

Now we prove that $\musemTAnotheta{\TDIV} \cap \tcsemTA{\tcA(\phi_1 \tcU \phi_2)} = \musemTAnotheta{\TDIV} \cap \musemTAnotheta{\embed_{n}(\tcA(\phi_1 \tcU \phi_2))}$.  The proof of this result follows very similar lines to the proof of the characterization of $\tcA(\phi_1 \tcU \phi_2)$ in Theorem~\ref{thm:divergent-ta-Lrelmunu-subsumes-TCTL}.  In particular, one may construct a \emph{sub-TTS}%
\footnote{
TTS $\T_1 = (Q_1, \ttrans{}_1, \Lab_1, Q_{0,1})$ is a sub-TTS of TTS $\T_2 = (Q_2, \ttrans{}_2, \Lab_2, Q_{0,1})$ iff $Q_1 \subseteq Q_2$, ${\ttrans{}_1} \subseteq {\ttrans{}_2}$, $\Lab_1 (q) = \Lab_2(q)$ for all $q \in Q_1$, and $Q_{0,1} \subseteq Q_{0,2}$.
}
$\T = (Q, \ttrans{}, \Lab, Q_0)$ of $\TS{\TA}$ as follows: $Q = \QTA \cap \musemTAnotheta{\TDIV}, q \ttrans{a} q'$ iff $q \ttransTA{a} q'$, $\Lab = \LabTA$, and $Q_0 = Q_{0,TA} \cap \musemTAnotheta{\TDIV}$. All the states in $\T$ satisfy $\TDIV$, and it is also the case that $\runs{\TA}{q} = \runs{\T}{q}$ for any $q \in Q$.  These facts, combined with a slight adaptation of the reasoning in Theorem~\ref{thm:divergent-ta-Lrelmunu-subsumes-TCTL} to define the semantics of formulas over sub-TTSes of $\TS{\TA}$, gives the desired result.\qedhere
\end{proof}

\begin{figure}
    \centering
    \begin{tikzpicture}[initial text=]
        \node[state,initial,align=center,minimum size=2.0cm] (l) {$l$ \\ \\ $x \geq 0$};
        \draw
            (l) edge[loop right,right] node[align=center]{$x \leq 0$ \\ $a$ \\ $\emptyset$} (l);
    \end{tikzpicture}
    \caption{A timelock-free timed automaton over $\Sigma = \{ a \}$ with Zeno runs.}
    \label{fig:zeno-TA}
\end{figure}
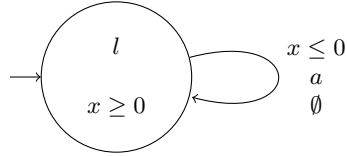

Finally, we show how to drop the nZ restriction in the encoding of $\tcA(\phi_1 \tcU \phi_2)$.
Zeno executions are time-convergent and as such should not influence whether or not a state satisfies $\tcA(\phi_1 \tcU \phi_2)$.  However, they also contain an infinite number of action transitions, and this fact exposes problems with the translations $\embed_{tn}(\tcA (\phi_1 \tcU \phi_2))$ and $\embed_{n}(\tcA (\phi_1 \tcU \phi_2))$ given above, which rely exclusively on least-fixpoint operators that in effect are violated by Zeno executions.  As an example, consider the timed automaton $\TA$ given in Figure~\ref{fig:zeno-TA} and the $\TCTL$ formula $\tcA\,(\tcF\, x \geq 1)$.  It can be seen that $\TS{\TA}$ contains no timelock states.  It does, however, contain non-Zeno executions:  one example is $(l, \initval) \ttransTA{a} (l, \initval) \ttransTA{a} \cdots$.  As every $r \in \runs{\TA}{l,\initval}$ is such that eventually $x \geq 1$, $(l,\initval) \in \tcsemTA{\tcA\,(\tcF\, x \geq 1))}$.  However, $(l,\initval) \not\in \musemTAnotheta{\embed_n(\tcA\,(\tcF\, x \geq 1)))}$.  To see why, note that since $\TA$ is timelock-free and no state in $\TS{\TA}$ satisfies $\TStop$, $\musemTAnotheta{\embed_n(\tcA\,(\tcF\, x \geq 1)))} = \musemTAnotheta{\mu X. \exists_{[\Sigma] X} x \geq 1}$.  A state $q \in \TS{\TA}$ can satisfy $\mu X. \exists_{[\Sigma] X} x \geq 1$ iff either $q(x) \geq 1$ or $q$ is incapable of an infinite sequence of $a$ transitions. Neither condition holds for $(l, \clockval)$, and thus $(l,\initval) \not\in \musemTAnotheta{\embed_n(\tcA\,(\tcF\, x \geq 1)))}$.

Operationally, a Zeno execution may be seen as \emph{unfair} to the passage of time:  while an infinite number of transitions happens, time only elapses finitely.  The correct $\Lrelmunu$ formula for $\tcA(\phi_1 \tcU \phi_2)$ in effect uses an \emph{alternating fixpoint} to rule out spurious inferences involving these unfair executions.
The encoding for non-Zeno timed automata is the following, where $z \notin cs(\phi_1) \cup cs(\phi_2)$ is a clock not appearing in $\phi_1$ or $\phi_2$ and we have \underline{underlined} the part of the translation differing from $\embed_n(\tcA (\phi_1 \tcU \phi_2))$.
\begin{align*}
&\embed(\tcA(\phi_1 \tcU \phi_2)) =
\\
&\mu X . \underline{z . \nu Y . }
\\
&\qquad\au (
    \TDIV \implies
        (
            \embed(\phi_1)
            \land
            \underline{(z \geq 1 \implies [\Sigma]X)
                \land (z < 1 \implies [\Sigma]Y)
            }
        ),
    \embed(\phi_2)
    )
\end{align*}

This encoding modifies $\embed_n(\tcA (\phi_1 \tcU \phi_2)$ in two signficant ways.  The first is that the single fixpoint $\mu X. \ldots$ in the $\embed_n$ translation is replaced by two alternating fixpoints separated by a clock-reset operator: $\mu X . z. \nu Y. \ldots$.  The second is that the subformula $[\Sigma] X$ in the $\embed_n$ translation is replaced by $(z \geq 1 \implies [\Sigma] X) \land (z < 1 \implies [\Sigma] Y)$.
We can now state and prove the correctness of this translation of $\tcA(\phi_1 \tcU \phi_2)$.\jk{I think we can add some more intuition here, explaining why a Zeno-path makes the formula true}\rc{Agreed, but running out of time for arXiv ... maybe for the conference submission?}

\begin{lemma}[Correctness of $\embed(\tcA (\phi_1 \tcU \phi_2))$]\label{lem:correctness-of-AU}
Let $\TA$ be a timed automaton, and suppose that $\phi_1, \phi_2$ are $\TCTL$ formulas such that $\tcsemTA{\phi_1} = \musemTAnotheta{\embed(\phi_1)}$ and $\tcsemTA{\phi_2} = \musemTAnotheta{\embed(\phi_2)}$.  Then $\tcsemTA{\tcA (\phi_1 \tcU \phi_2)} = \musemTAnotheta{\embed(\tcA (\phi_1 \tcU \phi_2))}$.
\end{lemma}

\begin{proof}
Fix $\TA = \genTA$, with $\TS{\TA} = \ttsTA$, and let $\phi_1$ and $\phi_2$ be $\TCTL$ formulas such that $\tcsemTA{\phi_1} = \musemTAnotheta{\embed(\phi_1)}$ and $\tcsemTA{\phi_2} = \musemTAnotheta{\embed(\phi_2)}$.
Let $\theta \in \Var \to 2^{\QTA}$ be arbitrary.
For notational convenience we introduce the following, where $q \in \QTA$ and $S, S' \subseteq \QTA$.
\begin{align*}
\phi
    &= \tcA(\phi_1 \tcU \phi_2)
    \\
Q_{\phi}
    &= \tcsemTA{\phi}
    \\
Q_1
    &= \tcsemTA{\phi_1}
    \\
Q_2
    &= \tcsemTA{\phi_2}
    \\
Q_{12}
    &= \tcsemTA{\phi_1 \lor \phi_2}
    \\
\runs{\phi}{q}
    &= \{ r \in \runs{\TA}{q} \mid r \in \untils{\TA}{Q_{12}}{Q_2} \}
    \\
\gamma_1
    &= \embed(\phi_1)
    \\
\gamma_2
    &= \embed(\phi_2)
    \\
\gamma_{z, \geq 1}
    &= z \geq 1 \implies [\Sigma]X
    \\
\gamma_{z, < 1}
    &= z < 1 \implies [\Sigma]Y
    \\
\gamma'_1
    &= \TDIV \implies (\gamma_1 \land \gamma_{z, \geq 1} \land \gamma_{z, < 1})
    \\
\gamma_Y
    &= \nu Y. (\au(\gamma'_1, \gamma_2))
    \\
f'_S(S')
    &= \musemTA{\au(\gamma'_1, \gamma_2)}{(\theta[X := S])[Y := S']}
    \\
f_\phi(S)
    &= \musemTA{z.\gamma_Y}{\theta[X := S]}
\end{align*}
Note that $f_\phi \in 2^{\QTA} \to 2^{\QTA}$ and that $\musemTAnotheta{\embed(\phi)} = \mu f_\phi$.  Also, for every $S \subseteq \QTA$, $f'_S \in 2^{\QTA} \to 2^{\QTA}$ and $\musemTA{\gamma_Y}{\theta[X := S]} = \nu f'_S$.  Finally, since $z \not\in cs(\phi_1) \cup cs(\phi_2)$, it is the case that for all $q \in \QTA$ and $\delta \in \delays$, $q \in Q_1$ iff $q[z:=\delta] \in Q_1$, $q \in Q_2$ iff $q[z := \delta] \in Q_2$, and $q \in Q_{12}$ iff $q[z := \delta] \in Q_{12}$.
Moreover, since $z \in \clocksF$, we know $z \not\in \CX$, and this implies obvious relationships between $\executions{\TA}{q}$ and $\executions{\TA}{q[z := \delta]}$ for $\delta \in \delays$.  Specifically, suppose $\pi \in \executions{\TA}{q}$ and $\delta \in \delays$.  Then there is an execution $\pi_{z,\delta} \in \executions{\TA}{q[z := \delta]}$ defined as follows.  Let $d = \delta - q(z)$ and $\pi[i] = q_i \ttransTA{a} q_{i+1}$.  Then $\pi_{z,\delta}[i] = (q_i[z := q_i(z) + d]) \ttransTA{a} (q_{i+1}[z := q_{i+1}(z) + d])$.  Intuitively, $\pi_{z,\delta}$ is the same as $\pi$ except that the times associated with $z$ in the transitions of $\pi_{z,\delta}$ are uniformly shifted from those in $\pi$ by $d$, the difference between $\delta$ and the time associated with $z$ in the source state of $r$.  It is straightforward to establish that $\pi_{z,\delta}$ is indeed an execution in $\executions{\TA}{q[z := \delta]}$.  This same construction also guarantees that if execution $\pi \in \runs{\TA}{q[z := \delta]}$, then $\pi_{z,q(z)} \in \runs{\TA}{q}$.  It can be seen that $\pi \in \untils{\TA}{Q_{12}}{Q_2}$ iff $\pi_{z,\delta} \in \untils{\TA}{Q_{12}}{Q_2}$, and that if $\pi, \pi_{z,\delta} \in \untils{\TA}{Q_{12}}{Q_2}$ then $mup(\pi_{z,\delta}) = (mup(\pi))_{z,\delta}$.
Finally, we have that $\indices_t(\pi) = \indices_t(\pi_{z,\delta})$ and that for every $(\delta',i) \in \indices_s(\pi), \pi_{z,\delta}[\delta',i] = (\pi[\delta',i])[z := z']$, where $z' = (\pi[\delta',i])(z) + d$.  This implies that $q \in \tcsemTA{\phi}$ iff $q[z := \delta] \in \tcsemTA{\phi}$.

We begin the proof by first establishing a characterization of $f_\phi(S)$ for $S \subseteq \QTA$.  For notational convenience we first define the following sets of executions, where $S \subseteq \QTA$ and $\Pi, \Pi' \subseteq \executions{\TA}{Q_{\TA}}$.  The reader should note the (intended) similarities between these definitions and the operators $\mathsf{X}$ (``next''), $\mathsf{U}$ (``until''), and $\mathsf{G}$ (``always'') of Linear Temporal Logic.
\begin{align*}
\executions{\TA}{S}
    &= \{ \pi \in \executions{\TA}{Q_{\TA}} \mid \src(\pi) \in S \}
    \\
X_{\TA}(\Pi)
    &= \{ \pi \in \executions{\TA}{Q_{\TA}} \mid 1 \in \indices_\Sigma(\pi) \land \suffix{\pi}{\geq (0,1)} \in \Pi \}
    \\
\untils{\TA}{\Pi}{\Pi'}
    &=  \{ \pi \in \executions{\TA}{Q_{\TA}} \mid
    \exists (\delta,i) \in \indices_S(\pi) \colon (\suffix{\pi}{\geq (\delta,i)} \in \Pi'
    \;\land
    \\
    &\qquad
    (\forall (\delta',j) \in \indices_s(\pi) \colon (\delta',j) <_\pi (\delta,i) \implies \suffix{\pi}{(\delta',j)} \in \Pi))
    \}
    \\
G_{\TA}(\Pi)
    &= \{ \pi \in \executions{\TA}{Q_{\TA}} \mid \forall (\delta,i) \in \indices_s(\pi) \colon \suffix{\pi}{\geq (\delta,i)} \in \Pi \}
\end{align*}
Here $\Pi_{\TA}(S)$ consists of executions in $\TA$ whose source state is in $S$.  $X_{\TA}(\Pi)$ consists of executions whose first transition is an action transition (i.e.\/ $1 \in \indices_\Sigma(\pi)$) and whose suffix after this transition is in $\Pi$.
The set $\untils{\TA}{\Pi}{\Pi'}$ consists of executions in $\TA$ containing a suffix in $\Pi'$, with every other suffix longer than this suffix being in $\Pi$.
Finally, $G_{\TA}(\Pi)$ consists of sequences whose every suffix is in $\Pi$.
Note that $\untils{\TA}{\Pi,\Pi'}$ overloads the notation $\untils{\TA}{S}{S'}$ introduced earlier, where $S, S' \subseteq \QTA$.
It is easy to see that $\pi \in \untils{\TA}{S}{S'}$ iff $\pi \in \untils{\TA}{\executions{\TA}{S}}{\executions{\TA}{S'}}$.
We will avail ourselves of similar short-hand for the other operators above, writing e.g. $G_{\TA}(S)$ when $S \subseteq \QTA$ for $G_{\TA}(\executions{\TA}{S})$, etc.

We also introduce time-elapse versions of the above.  Define $\executions{\TA,E}{Q_{\TA}}
= \{ \pi \in \executions{\TA}{Q_{\TA}} \mid \forall i \in \indices_t(\pi) \colon \lab(\pi[i]) \in \delays \}$ to be the set of executions consisting only of time-elapse transitions.  Then:
\begin{align*}
U_{\TA,E}(\Pi, \Pi')
    &=  \{ \pi \in \executions{\TA}{Q_{\TA}}
        \mid
    \\
    & \qquad\qquad\exists (\delta,i) \in \indices_s(\pi) \colon
            \prefix{\pi}{\leq (\delta,i)} \in \executions{\TA,E}{Q_{\TA}} \cap \untils{\TA}{\Pi}{\Pi'}
        \}
    \\
G_{\TA,E}(\Pi)
    &=  \executions{\TA,E}{Q_\TA} \cap G_{\TA}(\Pi)
\end{align*}
$U_{\TA,E}(\Pi,\Pi')$ consists of executions in $\TA$ that satisfy the ``until requirement'' using a prefix containing only time-elapse transitions, although the rest of the execution after this prefix may include action transitions.  $G_{\TA,E}(\Pi)$ consists of time-elapse-only executions whose suffixes are entirely in $\Pi$.
Also in what follows, we will use $Q_{TL} = \{q \in \QTA \mid \runs{\TA}{q} = \emptyset\}$ for the set of timelock states in $\TA$
and
$\executions{\TA,\infty}{q} = \{ \pi \in \executions{\TA}{q} \mid |\pi| = \infty \}$ for the infinite executions from $q$.
Note that $\runs{\TA}{q} \subseteq \executions{\TA,\infty}{q}$.
Define $Q_2' = Q_2 \cup Q_{TL}$, $Q_{12}' = Q_1 \cup Q_2' = Q_{12} \cup Q_{TL}$, and $Q_1' = Q_1 \setminus Q_2'$.  We now give our alternative characterization of $f_\phi(S)$ as follows.

\begin{align}
f_\phi(S)
    &=  \{ q \in \QTA
        \mid \executions{\TA,\infty}{q} \subseteq \Pi_1 \cup \Pi_2
        \cup \Pi_3(S)\}, \;\text{where}
    \label{eq:f-phi}
    \\
\Pi_1
    &=  \{ \pi \in \executions{\TA}{q}
        \mid \exists (\delta,i) \in \indices_s \colon
            \delta < 1 \land \prefix{\pi}{\leq (\delta,i)} \in \untils{\TA}{Q_{12}'}{Q_2'}
        \}
    \nonumber
    \\
\Pi_2
    &=  \{ \pi \in \executions{\TA}{q}
        \mid D(\pi) < 1 \land \pi \in G_{\TA}(Q_1')
        \}
    \nonumber
    \\
\Pi_3(S)
    &=  \Pi_{3,1} \cap (\Pi_{3,2} \cup \Pi_{3,3}(S))
    \nonumber
    \\
\Pi_{3,1}
    &=  \{ \pi \in \executions{\TA}{q}
        \mid
        D(\pi) \geq 1 \land
        \forall (\delta,i) \in \indices_s(\pi) \colon
            \delta < 1 \implies \pi[\delta,i] \in Q_1'
        \}
    \nonumber
    \\
\Pi_{3,2}
    &= \{ \pi \in \executions{\TA}{q}
        \mid D(\pi) < \infty \land
        \suffix{\pi}{\geq (1,0)} \in G_{\TA,E}(Q_{1}')
        \}
    \nonumber
    \\
\Pi_{3,3}(S)
    &=  \{ \pi \in \executions{\TA}{q}
        \mid
        \suffix{\pi}{\geq (1,0)} \in
        \untils{\TA,E}{\executions{\TA}{Q_{12}'}}{\executions{\TA}{Q'_2} \cup X_{\TA}(S)}
        \}
    \nonumber
\end{align}
Intuitively, this characterization asserts that $q \in f_\phi(S)$ iff each infinite-length execution $\pi \in \executions{\TA, \infty}{q}$ satisfies one of the following.
\begin{description}
\item[$\pi \in \Pi_1$.]
    This holds if $\pi$ has an initial of duration $< 1$ satisfying the until property.  Note that the definition of $\untils{\TA}{Q_{12}'}{Q_2'}$ guarantees that $\pi \in \untils{\TA}{Q_{12}'}{Q_2'}$ in this case.
\item[$\pi \in \Pi_2$.]
    This holds if the duration of $\pi$ is $< 1$ and every state in $\pi$ is in set $Q_1'$.  In this case, $\pi$ is time-convergent.
\item[$\pi \in \Pi_3(S)$.]
    In this case $\pi$ must have duration $\geq 1$, and every state in $\pi$ occurring at time $< 1$ must be in $Q_1'$ ($\Pi_{3,1}$).  The extent of $\pi$ after $(1,0)$ must then either be of finite duration and only pass through states in $Q_1'$ ($\Pi_{3,2}$), or must pass through states in $Q_{12}'$ until hitting a state that is either in $Q_2'$ or is the source of an action transition in $\pi$ whose target in set $S$, the argument to $f_\phi(S)$ ($\Pi_{3,3}$).
\end{description}
We now establish that Equation~\ref{eq:f-phi} is valid.  So fix $S \subseteq \QTA$.  For notational convenience, let $Q_{\phi,S}$ be the right-hand side of Equation~\ref{eq:f-phi}.  We prove Equation~(\ref{eq:f-phi}) by showing $f_\phi(S) \subseteq Q_{\phi,S}$ and $Q_{\phi,S} \subseteq f_\phi(S)$.
For the former, assume that $q \in f_\phi(S)$; we must show that $q \in Q_{\phi,S}$.  We note that since $z \not\in cs(\phi_1) \cup cs(\phi_2)$, $q \in f_\phi(S)$ iff $q[z := 0] \in \musemTA{\gamma_Y}{\theta[X := S]} = \nu f'_S$, and, since $z \not\in \CX$, $q \in Q_{\phi,S}$ iff $q[z := 0] \in Q'_{\phi,S}$, where $Q'_{\phi,S}$ modifies $Q_{\phi,S}$ as follows (here $\monus$ is the usual ``monus'' operator adapted to $\delays$:  $\delta \monus \delta' = \delta - \delta'$ if $\delta \geq \delta'$, and is $0$ otherwise).
\begin{align*}
Q'_{\phi,S}
    &=  \{ q \in \QTA \mid
            \executions{\TA}{q} \subseteq \Pi'_1 \cup \Pi'_2 \cup \Pi'_3(S)
        \} \text{, where}
    \\
\Pi'_1
    &=  \{ \pi \in \executions{\TA}{q}
        \mid \exists (\delta,i) \in \indices_s \colon
            \delta < 1 \monus q(z)  \land \prefix{\pi}{\leq (\delta,i)} \in \untils{\TA}{Q_{12}'}{Q_2'}
        \}
    \nonumber
    \\
\Pi'_2
    &=  \{ \pi \in \executions{\TA}{q}
        \mid D(\pi) < 1 \monus q(z) \land \pi \in G_{\TA}(Q_1')
        \}
    \nonumber
    \\
\Pi'_3(S)
    &=  \Pi'_{3,1} \cap (\Pi'_{3,2} \cup \Pi'_{3,3}(S))
    \nonumber
    \\
\Pi'_{3,1}
    &=  \{ \pi \in \executions{\TA}{q}
        \mid
    \\
    &\qquad
        D(\pi) \geq 1 \monus q(z) \land
        \forall (\delta,i) \in \indices_s(\pi) \colon
            \delta < 1 \monus q(z) \implies \pi[\delta,i] \in Q_1'
        \}
    \nonumber
    \\
\Pi'_{3,2}
    &= \{ \pi \in \executions{\TA}{q}
        \mid D(\pi) < \infty \land
        \suffix{\pi}{\geq (1 \monus q(z),0)} \in G_{\TA,E}(Q_{1}')
        \}
    \nonumber
    \\
\Pi'_{3,3}(S)
    &=  \{ \pi \in \executions{\TA}{q}
        \mid
        \suffix{\pi}{\geq (1 \monus q(z),0)} \in
        \untils{\TA,E}{\executions{\TA}{Q_{12}'}}{\executions{\TA}{Q'_2} \cup X_{\TA}(S)}
        \}
    \nonumber
\end{align*}
In effect, $Q'_{\phi,S}$ differs from $Q_{\phi,S}$ in that occurrences of delay $1 \in \delays$ in the definitions of $\Pi_1$, $\Pi_2$ and $\Pi_3(S)$ are replaced by $1 \monus q(z)$ in $\Pi'_1$, $\Pi'_2$ and $\Pi'_3(S)$.
Now, if we can show $\nu f'_S \subseteq Q'_{\phi,S}$ then $q[z:=0] \in Q'_{\phi,S}$ and $q \in Q_{\phi,S}$, thereby establishing the set inclusion $f_\phi(S) \subseteq Q_{\phi,S}$.
Since $\nu f'_S = \bigcup \{S' \subseteq \QTA \mid f'_S(S') \subseteq S'\}$ it suffices to show that for any $S' \subseteq f'_S(S')$, $S' \subseteq Q'_{\phi,S}$.  So fix $S' \subseteq f'_S(S')$ and $q' \in S'$; we establish that $q' \in Q'_{\phi,S}$ by showing that for each $\pi \in \executions{\TA,\infty}{q'}$, either $\pi \in \Pi'_1$, $\pi \in \Pi'_2$, or $\pi \in \Pi'_3(S)$.  To do this, assume $\pi \not\in \Pi'_1 \cup \Pi'_2$; we must show that $\pi \in \Pi'_3(S)$.  In what follows, let $\theta' = (\theta[X := S])[Y := S']$.  From the definition of $f'_S(S')$ and Lemma~\ref{lem:semantics-of-au} we know that
\[
f'_S(S') = \musemTA{\exists_{\gamma'_1} \gamma_2}{\theta'} \cup \musemTA{(\exists\, \TStop) \land (\forall\, \gamma'_1)}{\theta'}.
\]
To establish that $\pi \in \Pi'_{3}(S)$ we first show that $\pi \in \Pi'_{3.1}$.  Since $q' \in S' \subseteq f'_S(S')$, it follows that either $q' \in \musemTA{\exists_{\gamma'_1} \gamma_2}{\theta'}$ or $q' \in \musemTA{(\exists\, \TStop) \land (\forall\, \gamma'_1)}{\theta'}$.  In the former case, there must exist $\delta' \in \delays$ such that $\delta'(q') \in \gamma_2$ and for all $\delta'' < \delta'$, $\delta''(q') \in \musemTA{\gamma'_1 \lor \gamma_2}{\theta'}$.  Because $\pi \not\in \Pi'_1$, it must follow that $\delta \geq 1 \monus q'(z)$, which implies that $D(\pi) \geq 1 \monus q'(z)$.
Also for this reason, there can be no $(\delta'',j) \in \indices_s(\pi)$ such that $\delta < 1 \monus q(z)$ and $\pi[\delta'',j] \in Q_2'$.
Since $q' \in \musemTA{\exists_{\gamma'_1} \gamma_2}{\theta'}$ the semantics of $\exists$ guarantees that for all such $(\delta'',j)$, $\pi[\delta'',j] \in \musemTA{\gamma'_1}{\theta'}$; these facts imply $\pi[\delta'',j] \in Q_1'$.
For the latter case, assume $q' \in \musemTA{(\exists\, \TStop) \land (\forall\, \gamma'_1)}{\theta'}$ but $q' \not\in \musemTA{\exists_{\gamma'_1} \gamma_2}{\theta'}$.
It follows that there exists a smallest $\delta' \in \delays$ such that for all $\delta'' \in \delays$, if $q' \ttransTA{\delta'}$ then $\delta'' \leq \delta'$.  We also note that for all $(\delta''',j) \in \indices_s(\pi)$, $\pi[\delta''',j] \not\in Q'_2$ since otherwise $\pi$ would be in $\Pi_1'$, which is a contradiction.  Now, since since $\pi \not\in \Pi'_2$ it must be the case that $D(\pi) \geq 1$.  Consequently, in this case $\pi \in \Pi'_{3,1}$ as well.

To finish proving that $\pi \in \Pi'_3$ we show that $\pi \in \Pi'_{3,2} \cup \Pi'_{3,3}(S)$.  So assume $\pi \not\in \Pi'_{3,2}$; we must establish that $\pi \in \Pi'_{3,3}(S)$.  Since $\pi \not\in \Pi'_{3,2}$ we know $D(\pi) = \infty$.  We have already observed that either $q' \in \musemTA{\exists_{\gamma'_1} \gamma_2}{\theta'}$ or $q' \in \musemTA{(\exists\, \TStop) \land (\forall\, \gamma'_1)}{\theta'}$. We now do a case analysis.
First assume that $q' \in \musemTA{\exists_{\gamma'_1} \gamma_2}{\theta'}$.  The argument above established that there must exist $\delta \geq 1 \monus q'(z)$ such that $\delta(q) \in Q_2'$ and for all $\delta' < \delta$, $\delta'(q) \in Q'_{12}$.
From the definitions, we observe that
\begin{align*}
\suffix{\pi}{\geq 1 \monus q'(z)}
    &\in
    \untils{\TA,E}{\executions{\TA}{Q_{12}'}}{\executions{\TA}{Q_2'}}
    \\
    &\subseteq
    \untils{\TA,E}{\executions{\TA}{Q_{12}'}}{\executions{\TA}{Q_2'} \cup X_{\TA}(S)}.
\end{align*}
Thus $\pi \in \Pi'_{3,3}(S)$.
Now assume that $q' \in \musemTA{(\exists\, \TStop) \land (\forall\, \gamma'_1)}{\theta'}$ but $q' \not\in \musemTA{\exists_{\gamma'_1} \gamma_2}{\theta'}$.
Using reasoning given above, we know there exists a smallest $\delta' \in \delays$ such that for all $\delta'' \in \delays$, if $q' \ttransTA{\delta''}$ then $\delta'' \leq \delta$, and that $\delta' \geq 1 \monus q(s)$.  Since $D(\pi) = \infty$ it follows that $\pi[q'] \ttransTA{1 \monus q(s)}$.
The facts that $D(\pi) = \infty$, $q' \in \musemTA{(\exists\, \TStop) \land (\forall\, \gamma'_1)}{\theta'}$, and $q' \not\in \musemTA{\exists_{\gamma'_1} \gamma_2}{\theta'}$ guarantee that there exists a $(\delta''',j) \in \indices_s$ such that $\delta \geq 1 \monus q'(z)$ and $\suffix{\pi}{\geq(\delta''',j)} \in X_{\TA}(S)$.  Consequently, $\pi \in \Pi'_{3,3}(S)$ in this case as well.
We have therefore established for any $\pi \in \executions{\TA,\infty}{q'}, \pi \in \Pi'_1 \cup \Pi'_2 \cup \Pi'_3(S)$,
thereby showing that $q' \in Q'_{\phi,S}$ and $S' \cup Q'_{\phi,S}$ when $S' \subseteq f'_S(S')$.
Consequently, $\nu f'_S \subseteq Q'_{\phi,S}$ and $f_{\phi}(S) \subseteq Q_{\phi,S}$.

To finish establishing the validity of Equation~\ref{eq:f-phi} we now show that $Q_{\phi,S} \subseteq f_\phi(S)$.  So assume that $q \in Q_{\phi,S}$; we must show that $q \in f_\phi(S)$, which in turn holds iff $q[z:=0] \in \musemTA{\gamma_Y}{\theta[X := S]}$.  It is immediate to see that $q[z:=0] \in Q'_{\phi,S}$, so if $Q'_{\phi,S} \subseteq \musemTA{\gamma_Y}{\theta[X := S]}$ then $q[z:=0] \in \musemTA{\gamma_Y}{\theta[X := S]}$, $q \in f_\phi(S)$, and we are done.
It suffices to show that $Q'_{\phi,S} \subseteq f_S(Q'_{\phi,S})$.  So fix $q' \in Q'_{\phi,S}$; we note that $q' \in f_S(Q'_{\phi,S})$ iff $q' \in \musemTA{\au(\gamma'_1,\gamma_2)}{\theta'}$, where $\theta' = (\theta[X := S])[Y := Q'_{\phi,S}]$.  Based on Lemma~\ref{lem:semantics-of-au}, it suffices to show that $q' \in \musemTA{\exists_{\gamma'_1} \gamma_2}{\theta'} \cup \musemTA{(\exists\, \TStop) \land (\forall\gamma'_1)}{\theta'}$.  To this end, assume $q' \not\in \musemTA{\exists_{\gamma'_1} \gamma_2}{\theta'}$; we must establish that $q' \in \musemTA{(\exists\, \TStop) \land (\forall\gamma'_1)}{\theta'}$, i.e.\/ that $q' \in \musemTA{\exists\, \TStop}{\theta'}$ and $q' \in \musemTA{\forall\gamma'_1}{\theta'}$.
For the former, suppose by way of contradiction that $q \not\in \musemTA{\exists\, \TStop}{\theta'}$.  This means there exists run $\pi \in \executions{\TA}{q'}$ consisting only of time-elapse transitions.  Since $q' \not\in \musemTA{\exists_{\gamma'_1} \gamma_2}{\theta'}$ it also follows that no state along $\pi$ is in $Q_2$.  From the definition of $Q'_{\phi,S}$ it can be seen that $\pi \not\in \Pi'_1 \cup \Pi'_2 \cup \Pi'_3(S)$, and thus $q' \not\in Q'_{\phi,S}$, which is a contradiction.  Consequently, $q' \in \musemTA{\exists\, \TStop}{\theta'}$.  To show that $q' \in \musemTA{\forall\gamma'_1}{\theta'}$ we again argue by contradiction.
So suppose to the contrary that $q' \not\in \musemTA{\forall\gamma'_1}{\theta'}$.  This means there exists $\delta' \in \delays$ such that $\delta'(q') \not \in Q'_1$.  Since $q' \not\in \musemTA{\exists_{\gamma'_1} \gamma_2}{\theta'}$ it also follows that for all $\delta'' < \delta'$, $\delta''(q) \not\in Q_2$.
Now consider any infinite execution $\pi \in \executions{\TA,\infty}{q'}$ beginning with transition $q' \ttransTA{\delta'} \delta'(q)$.  It can again be seen that $\pi \not\in \Pi'_1 \cup \Pi'_2 \cup \Pi'_3(S)$.  This would imply that $q' \in Q'_{\phi,S}$, which is the desired contradiction, and $q' \in \musemTA{\forall\gamma'_1}{\theta'}$.
Thus $q' \in \musemTA{\forall\gamma'_1}{\theta'}$; $q' \in \musemTA{\exists\, \TStop}{\theta'}$; $q' \in \musemTA{\au(\gamma'_1, \gamma_2)}{\theta'}$; $Q'_{\phi,S} \subseteq f'_S(Q'_{\phi,S})$; $Q'_{\phi,S} \subseteq \musemTA{\gamma_Y}{\theta[X := S]}$; $q[z:=0] \in \musemTA{\gamma_Y}{\theta[X := S]}$; $q \in f_\phi(S)$; and $Q_{\phi,S} \subseteq f_{\phi,S}$.
This finishes the proof of validity of Equation~\ref{eq:f-phi}.

We now finish the proof of this lemma by showing that $\tcsemTA{\tcA (\phi_1 \tcU \phi_2)} = \musemTAnotheta{\embed(\tcA (\phi_1 \tcU \phi_2))}$, i.e.\/ that $Q_\phi \subseteq \mu f_\phi$ and $\mu f_\phi \subseteq Q_\phi$.  For the former, it suffices to give a well-founded relation ${\prec} \subseteq Q_\phi \times Q_\phi$ such that $(Q_\phi, {\prec})$ is a support structure for $f_\phi$.  We use the same definition for $\prec$ given in the proof of Theorem~\ref{thm:divergent-ta-Lrelmunu-subsumes-TCTL} for the translation for $\tcA(\phi_1 \tcU \phi_2)$:  $q' \prec q$ iff there is a run $r \in \runs{\TA}{q}$, with $\pi = mup_\phi(r)$, and action-transition index $i \in \indices_\Sigma(\pi)$ such that $\tgt(\pi[i]) = q'$.  The same argument as before establishes that $\prec$ is well-founded.  We now must show that $(Q_\phi,{\prec})$ is a support structure for $f_\phi$.  So fix $q \in Q_\phi$, and define $S_q = \preimg{{\prec}}{q}$; we must show that $q \in f_\phi(S_q)$.  We begin by noting that $q \in f_\phi(S_q)$ iff $q[z:=0] \in \musemTA{\gamma_Y}{\theta[X := S_q]}$, and that $q[z:=0] \in Q'_{\phi,S_q}$ defined above.  The previous argument also established that $Q'_{\phi,S_q} \subseteq \musemTA{\gamma_Y}{\theta[X := S_q]}$, meaning $q[z:=0] \in \musemTA{\gamma_Y}{X := S_q}$ and $q \in f_\phi(S_q)$.  Thus $(Q_\phi, \prec)$ is a well-founded support structure for $f_\phi$, and $Q_\phi \subseteq \mu f_\phi$.

We now establish that $\mu f_\phi \subseteq Q_\phi$.  For this, it suffices to show that $f_\phi(Q_\phi) \subseteq Q_\phi$.  So suppose that $q \in f_\phi(Q_\phi)$; we must establish that $q \in Q_\phi$.  Based on Equation~\ref{eq:f-phi} we know that $\executions{\TA,\infty}{q} \subseteq \Pi_1 \cup \Pi_2 \cup \Pi_3(Q_\phi)$.
To demonstrate that $q \in Q_\phi$ we must show that every $r \in \runs{\TA}{q}$ is an element of $U_{\TA}(Q_{12}, Q_2)$.  Since $r \in \runs{\TA}{q}$ it follows that $r \in \executions{\TA,\infty}{q}$ and thus $r \in \Pi_1 \cup \Pi_2 \cup \Pi_3(Q_\phi)$.
Moreover, $D(r) = \infty$, so $r \not\in \Pi_2$ and $r \not\in \Pi_{3,2}$.  Thus $r \in \Pi_1 \cup (\Pi_{3,1} \cap \Pi_{3,3}(Q_\phi))$.  Suppose $r \in \Pi_1$.  This means there is $\delta < 1$ such that $\prefix{r}{\leq (\delta,i)} \in \untils{\TA}{Q_{12}'}{Q_2'}$.  Since $D(r) = \infty$ it follows that for every $(\delta,i) \in \indices_s(r)$ we have $r[\delta,i] \not\in Q_{TL}$.  This implies that $\prefix{r}{\leq (\delta,i)} \in \untils{\TA}{Q_{12}}{Q_2}$, whence $r \in U_{\TA}(Q_{12},Q_2)$.
Now suppose that $r \in \Pi_{3,1} \cap \Pi_{3,3}(Q_\phi)$.  From the definition of $\Pi_{3,1}$ it follows that for every $(\delta,i) \in \indices_s(r)$ such that $\delta < 1$, $r[\delta,i] \in Q'_1 \subseteq Q_1$.  From the definition of $\Pi_{3,3}(Q_\phi)$ we also know that either $\suffix{r}{\geq (1,0)} \in U_{\TA,E}(\Pi_{\TA}(Q'_{12}),\Pi_{\TA}(\Pi_{\TA}(Q'_2))$ or $\suffix{r}{\geq (1,0)} \in U_{\TA,E}(\Pi_{\TA}(Q'_{12}),X_{\TA}(Q_\phi))$.
In the former case, using the same reasoning as above, we know that $\suffix{r}{\geq (1,0)} \in U_{\TA,E}(\Pi_{\TA}(Q_{12}),\Pi_{\TA}(\Pi_{\TA}(Q_2))$, and this fact and the earlier observation about $r[\delta,i]$ when $\delta < 1$ implies that $r \in U_{\TA}(Q_{12},Q_2)$.
In the latter case, based on the definitions of $U_{\TA,E}$ and $X_{\TA}$, it also follows $r \in U_{\TA}(Q_{12},Q_2)$.  Thus $r \in U_{\TA}(Q_{12}, Q_2)$ and $q \in f_{\phi}(Q_\phi)$, thereby completing the proof.\qedhere
\end{proof}

We may now give the formal definition of $\embed(-)$ as follows.
\begin{definition}[Translation of $\TCTL$ to $\Lrelmunu$]
Let $\phi$ be a TCTL formula.  Then $\Lrelmunu$ formula $\embed(\phi)$ is defined inductively as follows, where $A \in \AP_\clocks$.
\begin{align*}
\embed(A)
    &= A
\\
\embed(\lnot \phi)
    &= \lnot\, \embed(\phi)
\\
\embed(\phi_1 \lor \phi_2)
    &= \embed(\phi_1) \lor \embed(\phi_2)
\\
\embed(\tcE(\phi_1 \tcU \phi_2))
    &=  \mu X . \exists_{\embed(\phi_1)}((\embed(\phi_2)\: \land\: \TDIV) \lor (\embed(\phi_1) \land \dia{\Sigma} X))
\\
\embed(\tcA (\phi_1 \tcU \phi_2))
    &=  \mu X . z . \nu Y .
    (
        \TDIV \implies \phi_1'
    ,\;
        \embed(\phi_2)
    )
\\
\phi_1'
    &=  \embed(\phi_1)
        \land
        (z \geq 1 \implies [\Sigma]X)
        \land
        (z < 1 \implies [\Sigma]Y)
\\
\embed(z.\phi)
    &= z.\embed(\phi)
\end{align*}
\end{definition}

We now state and prove the correctness of this translation.
\begin{theorem}\label{thm:Lrelmunu-subsumes-TCTL}
Let $\TA = (L, L_0, \CX, I, E, \Lab)$ be a timed automaton over time-safe $\Sigma$ and $\AP$.
Then for all TCTL formulas $\phi$,
$
\tcsemTA{\phi} = \musemTAnotheta{\embed(\phi)}.
$
\end{theorem}
\begin{proof}
Proceeds by structural induction on $\phi$.  Most cases are straightforward.  The correctness of $\embed(\tcE(\phi_1 \tcU \phi_2))$ follows from Lemma~\ref{lem:correctness-of-EU}.  The correctness of $\embed(\tcA(\phi_1 \tcU \phi_2)$ is established in Lemma~\ref{lem:correctness-of-AU}.\qedhere
\end{proof}

\section{Conclusion and Directions for Future Work} \label{sec:conclusion}

This paper has presented a timed modal mu-calculus, $\Lrelmunu$, and shown that it is strictly more expressive than other timed mu-calculi given in the literature for the model of timed automata.  It is also strictly more expressive than the timed branching-time temporal logic TCTL over arbitary timed automa, in contrast with other timed modal mu-calculi.  $\Lrelmunu$ extends the traditional untimed modal mu-calculus with modalities for capturing the passage of time; these modalities have the flavor of well-known ``until" and ``release" modalities from classical temporal logic.  Model checking of $\Lrelmunu$ over timed automata is decidable and can be implemented via well-known region-graph constructions.

Regarding future work, it would be interesting to explore these expressiveness results in the more general setting of timed transition systems (TTSes).  Formalisms such as hybrid automata~\cite{alur-et-al-1995} also have a semantics in terms of TTSes; such expressiveness results would illustrate the power of $\Lrelmunu$ for reasoning about such systems as well.

\bibliographystyle{splncs04}
\bibliography{references}

\begin{thebibliography}{10}
\providecommand{\url}[1]{\texttt{#1}}
\providecommand{\urlprefix}{URL }
\providecommand{\doi}[1]{https://doi.org/#1}

\bibitem{aceto-is-your-2002}
Aceto, L., Laroussinie, F.: Is your model checker on time? {O}n the complexity of model checking for timed modal logics. Journal of Logic and Algebraic Programming  \textbf{52--53}(0),  7--51 (2002). \doi{10.1016/S1567-8326(02)00022-X}

\bibitem{Alu1991}
Alur, R.: Techniques for Automatic Verification of Real-Time Systems. Ph.D. thesis (1991), \url{https://www.cis.upenn.edu/~alur/Thesis91.pdf}

\bibitem{ACD1990}
Alur, R., Courcoubetis, C., Dill, D.: Model-checking for real-time systems. In: [1990] {{Proceedings}}. {{Fifth Annual IEEE Symposium}} on {{Logic}} in {{Computer Science}}. pp. 414--425 (1990). \doi{10.1109/LICS.1990.113766}

\bibitem{ACD1993}
Alur, R., Courcoubetis, C., Dill, D.: Model-checking in dense real-time. Information and Computation  \textbf{104}(1),  2--34 (1993). \doi{10.1006/inco.1993.1024}

\bibitem{alur-et-al-1995}
Alur, R., Courcoubetis, C., Halbwachs, N., Henzinger, T., Ho, P.H., Nicollin, X., Olivero, A., Sifakis, J., Yovine, S.: The algorithmic analysis of hybrid systems. Theor. Comput. Sci.  \textbf{138}(1),  3--34 (1995). \doi{10.1016/0304-3975(94)00202-T}, hybrid Systems

\bibitem{alur-a-theory-1994}
Alur, R., Dill, D.L.: A theory of timed automata. Theor. Comput. Sci.  \textbf{126}(2),  183--235 (April 1994). \doi{10.1016/0304-3975(94)90010-8}

\bibitem{andersen-model-checking-1994}
Andersen, H.: Model checking and boolean graphs. Theor. Comput. Sci.  \textbf{126}(1),  3--30 (1994). \doi{10.1016/0304-3975(94)90266-6}

\bibitem{behrmann-a-tutorial-2004}
Behrmann, G., David, A., Larsen, K.: A tutorial on \textsc{Uppaal}. In: Bernardo, M., Corradini, F. (eds.) Formal Methods for the Design of Real-Time Systems, International School on Formal Methods for the Design of Computer, Communication and Software Systems (SFM-RT '04). {LNCS}, vol.~3185, pp. 200--236. Springer Berlin Heidelberg, Bertinoro, Italy (September 2004). \doi{10.1007/b110123}

\bibitem{bhat-efficient-model-1996}
Bhat, G., Cleaveland, R.: Efficient model checking via the equational $\mu$-calculus. In: Proceedings of the 11th Annual IEEE Symposium on Logic and Computer Science (LICS '96). pp. 304--312. IEEE Computer Society, New Brunswick, NJ, USA (July 1996). \doi{10.1109/LICS.1996.561358}

\bibitem{bouyer-timed-modal-2011}
Bouyer, P., Cassez, F., Laroussinie, F.: Timed modal logics for real-time systems. Journal of Logic, Language and Information  \textbf{20}(2),  169--203 (2011). \doi{10.1007/s10849-010-9127-4}

\bibitem{bouyer-on-the-2010}
Bouyer, P., Chevalier, F., Markey, N.: On the expressiveness of {TPTL} and {MTL}. Inf. Comput.  \textbf{208}(2),  97--116 (2010). \doi{10.1016/j.ic.2009.10.004}

\bibitem{BRADFIELD1998133}
Bradfield, J.: The modal mu-calculus alternation hierarchy is strict. Theoretical Computer Science  \textbf{195}(2),  133--153 (1998). \doi{10.1016/S0304-3975(97)00217-X}, concurrency Theory

\bibitem{clarke-automatic-verification-1986}
Clarke, E., Emerson, E., Sistla, A.: Automatic verification of finite-state concurrent systems using temporal logic specifications. ACM Transactions on Programming Languages and Systems (TOPLAS)  \textbf{8}(2),  244--263 (1986). \doi{10.1145/5397.5399}

\bibitem{cleaveland-a-linear-time-1993}
Cleaveland, R., Steffen, B.: A linear-time model-checking algorithm for the alternation-free modal mu-calculus. Formal Methods in System Design  \textbf{2}(2),  121--147 (1993). \doi{10.1007/BF01383878}

\bibitem{CK2023}
Cleaveland, R., Keiren, J.J.A.: Extensible proof systems for infinite-state systems. ACM Trans. Comput. Logic  (sep 2023). \doi{10.1145/3622786}

\bibitem{emerson-1991}
Emerson, E.: Temporal and Modal Logic, p. 995–1072. MIT Press, Cambridge, MA, USA (1991)

\bibitem{emerson-halpern-1986}
Emerson, E., Halpern, J.: “sometimes” and “not never” revisited: On branching versus linear time temporal logic. J. ACM  \textbf{33}(1),  151–178 (jan 1986). \doi{10.1145/4904.4999}

\bibitem{emerson-efficient-model-1986}
Emerson, E., Lei, C.L.: Efficient model checking in fragments of the propositional mu-calculus. In: Proceedings of the 1st Symposium on Logic in Computer Science (LICS '86). pp. 267--278. IEEE Computer Society (June 1986)

\bibitem{fontana-2014}
Fontana, P.: Towards a Unified Theory of Timed Automata. Ph.D. thesis, University of Maryland (2014)

\bibitem{fontana-data-structure-2011}
Fontana, P., Cleaveland, R.: Data structure choices for on-the-fly model checking of real-time systems. In: Ganai, M., Biere, A. (eds.) Proceedings of the First International Workshop on Design and Implementation of Formal Tools and Systems (DIFTS '11). CEUR Workshop Proceedings, vol.~832, pp. 13--21. Austin, TX, USA (November 2011), \url{http://ceur-ws.org/Vol-832/Difts11Proceedings.pdf\#page=17}

\bibitem{fontana-the-power-2014}
Fontana, P., Cleaveland, R.: The power of proofs: New algorithms for timed automata model checking. In: Formal {{Modeling}} and {{Analysis}} of {{Timed Systems}}. pp. 115--129. {Springer, Cham} (Sep 2014). \doi{10.1007/978-3-319-10512-3_9}

\bibitem{graf1986modal}
Graf, S., Sifakis, J.: A modal characterization of observational congruence on finite terms of {CCS}. Information and Control  \textbf{68}(1-3),  125--145 (1986)

\bibitem{henzinger-symbolic-model-1994}
Henzinger, T., Nicollin, X., Sifakis, J., Yovine, S.: Symbolic model checking for real-time systems. Inf. Comput.  \textbf{111}(2),  193--244 (1994). \doi{10.1006/inco.1994.1045}

\bibitem{kamp68}
Kamp, H.: Tense logic and the theory of linear order. Ph.D. thesis, {UCLA} (1968)

\bibitem{kozen-results-on-1983}
Kozen, D.: Results on the propositional $\mu$-calculus. Theor. Comput. Sci.  \textbf{27}(3),  333--354 (1983). \doi{10.1016/0304-3975(82)90125-6}

\bibitem{laroussinie-cmc:-a-1998}
Laroussinie, F., Larsen, K.: {CMC}: A tool for compositional model-checking of real-time systems. In: Budkowski, S., Cavalli, A., Najm, E. (eds.) Proceedings of the Joint International Conference on Formal Description Techniques and Protocol Specification, Testing and Verification (FORTE/PSTV '98). pp. 439--456. The International Federation for Information Processing (IFIP), Springer US, Paris, France (1998). \doi{10.1007/978-0-387-35394-4_27}

\bibitem{laroussinie-from-timed-1995}
Laroussinie, F., Larsen, K., Weise, C.: From timed automata to logic --- and back. In: Wiedermann, J., H{\'a}jek, P. (eds.) Proceedings of the 20th Annual Symposium on the Mathematical Foundations of Computer Science (MFCS '95). {LNCS}, vol.~969, pp. 529--539. Springer Berlin Heidelberg, Prague, Czech Republic (August 1995). \doi{10.1007/3-540-60246-1_158}

\bibitem{mader-verification-of-1997}
Mader, A.: Verification of Modal Properties Using Boolean Equation Systems. Edition versal 8, Bertz Verlag, Berlin, Germany (1997), \url{http://doc.utwente.nl/64253/}

\bibitem{mateescu-efficient-on-the-fly-2003}
Mateescu, R., Sighireanu, M.: Efficient on-the-fly model-checking for regular alternation-free mu-calculus. Science of Computer Programming  \textbf{46}(3),  255--281 (2003). \doi{10.1016/S0167-6423(02)00094-1}

\bibitem{penczek-advances-in-2006}
Penczek, W., P{\'o}lrola, A.: Advances in Verification of Time Petri Nets and Timed Automata, Studies in Computational Intelligence, vol.~20. Springer Berlin Heidelberg, Secaucus, NJ, USA (2006). \doi{10.1007/978-3-540-32870-4}

\bibitem{pnueli-1977}
Pnueli, A.: The temporal logic of programs. In: 18th Annual Symposium on Foundations of Computer Science (sfcs 1977). pp. 46--57 (1977). \doi{10.1109/SFCS.1977.32}

\bibitem{seshia-unbounded-fully-2003}
Seshia, S., Bryant, R.: Unbounded, fully symbolic model checking of timed automata using boolean methods. In: {Hunt Jr.}, W., Somenzi, F. (eds.) Proceedings of the 15th International Conference on Computer Aided Verification (CAV '03). {LNCS}, vol.~2742, pp. 154--166. Springer Berlin Heidelberg (2003). \doi{10.1007/978-3-540-45069-6_16}

\bibitem{sokolsky-local-model-1995}
Sokolsky, O., Smolka, S.: Local model checking for real-time systems. In: Wolper, P. (ed.) Proceedings of the 7th International Conference on Computer Aided Verification (CAV '95). {LNCS}, vol.~939, pp. 211--224. Springer Berlin Heidelberg (July 1995). \doi{10.1007/3-540-60045-0_52}

\bibitem{steffen1994characteristic}
Steffen, B., Ingolfsdottir, A.: Characteristic formulas for processes with divergence. Information and Computation  \textbf{110}(1),  149--163 (1994)

\bibitem{tarski-a-lattice-theoretical-1955}
Tarski, A.: A lattice-theoretical fixpoint theorem and its applications. Pacific Journal of Mathematics  \textbf{5}(2),  285--309 (1955)

\bibitem{tripakis-analysis-of-2001}
Tripakis, S., Yovine, S.: Analysis of timed systems using time-abstracting bisimulations. Formal Methods in System Design  \textbf{18}(1),  25--68 (January 2001). \doi{10.1023/A:1008734703554}

\bibitem{wang-efficient-verification-2004}
Wang, F.: Efficient verification of timed automata with {BDD}-like data structures. International Journal on Software Tools for Technology Transfer  \textbf{6}(1),  77--97 (2004). \doi{10.1007/s10009-003-0135-4}

\bibitem{wolper-1983}
Wolper, P.: Temporal logic can be more expressive. Information and Control  \textbf{56}(1),  72--99 (1983). \doi{10.1016/S0019-9958(83)80051-5}

\bibitem{zhang-fast-generic-2005}
Zhang, D., Cleaveland, R.: Fast generic model-checking for data-based systems. In: Wang, F. (ed.) Proceedings of the International Conference on the Formal Techniques for Networked and Distributed Systems (FORTE '05). {LNCS}, vol.~3731, pp. 83--97. Springer Berlin Heidelberg (2005). \doi{10.1007/11562436_8}

\end{thebibliography}




\end{document}